\documentclass[11pt]{amsart}
\pdfoutput=1
\usepackage[english]{babel}
\usepackage{graphicx,amsmath,relsize}
\usepackage{enumitem}
\usepackage{MnSymbol}
\usepackage{subfig}
\usepackage{float}
\usepackage{floatflt,young,youngtab}
\usepackage{tikz}
\theoremstyle{plain}
\numberwithin{equation}{section}

\newtheorem{theorem}{Theorem}[subsection]
\newtheorem{corollary}[theorem]{Corollary}
\newtheorem{lemma}[theorem]{Lemma}
\newtheorem{proposition}[theorem]{Proposition}
\newtheorem{remark}[theorem]{Remark}

\newtheorem{definition}[theorem]{Definition}

\newtheorem{construction}[theorem]{Construction}
\newtheorem{main}{Main construction -- Part}

\def \R {{\mathbb R}}

\def \GTNN {{Gr^{\mbox{\tiny TNN}} (k,n)}}
\def \GTP {{Gr^{\mbox{\tiny TP}} (k,n)}}
\newcommand\mycom[2]{\genfrac{}{}{0pt}{}{#1}{#2}}
\def \DKP {{\mathcal D}_{\textup{\scriptsize KP},\Gamma}}

\def \DS {{\mathcal D}_{\textup{\scriptsize S},\Gamma_0}}

\def \DDN {{\mathcal D}_{\textup{\scriptsize dr},{\mathcal N}}}
\def \S {{\mathcal S}_{\mathcal M}^{\mbox{\tiny TNN}}}

\def \gdr {\gamma_{\textup{\scriptsize dr},V_l}}

\def \Pdr {P^{\textup{\scriptsize (dr)}}}

\title[KP divisors on $\mathtt M$--curves]{ Real regular KP divisors on $\mathtt M$--curves and totally non-negative Grassmannians}
\author{Simonetta Abenda}
\address{Dipartimento di Matematica and Alma Mater Research Center on Applied Mathematics, Universit\`a di Bologna, Italy\\ INFN, sez. di Bologna, Italy
}
\email{simonetta.abenda@unibo.it
}
\author{Petr G. Grinevich}
\address{Steklov Mathematical Institute of Russian Academy of Sciences, Moscow, Russia\\
L.D.Landau Institute for Theoretical Physics, Chernogolovka, Russia\\
Lomonosov Moscow State University, Faculty of Mechanics and Mathematics, Moscow, Russia}
\email{pgg@landau.ac.ru}

\thanks{
This research has been partially supported by GNFM-INDAM and RFO University of Bologna, by the Russian Foundation for Basic Research, grant 20-01-00157, by the program ``Fundamental problems of nonlinear dynamics'', Presidium of RAS.}

\begin{document}
\begin{abstract}
  {In this paper we construct an explicit map from planar bicolored (plabic) trivalent graphs representing a given irreducible positroid cell $\S$ in the totally non--negative Grassmannian $Gr^{\mbox{\tiny TNN}}(k,n)$ to the spectral data
	for the relevant class of real regular Kadomtsev-Petviashvili II (KP) solutions, thus completing search of
	real algebraic-geometric data for the KP equation started in \cite{AG1,AG3}.  The spectral curve is modeled on Krichever construction for degenerate finite-gap solutions, and is a rationally degenerate $\mathtt M$-curve, $\Gamma$, dual to the graph. The divisors are real regular KP divisors in the ovals of $\Gamma$, i.e. they fulfill the conditions for selecting real regular finite--gap solutions KPII solutions in \cite{DN}. Since the soliton data are described by points in $\S$, we establish a bridge between real regular finite-gap KP solutions \cite{DN} and real regular multi-line KP solitons which are known to be parameterized by points in $Gr^{\mbox{\tiny TNN}}(k,n)$ \cite{CK,KW2}.

We use the geometric characterization of spaces of relations on plabic networks introduced in  \cite{AG4} to prove the invariance of this construction with respect to the many gauge freedoms on the network. Such systems of relations were proposed in \cite{Lam2} for the computation of scattering amplitudes on on--shell diagrams $N=4$ SYM \cite{AGP1} and govern the totally non--negative amalgamation of the little positive Grassmannians, $Gr^{\mbox{\tiny TP}}(1,3)$ and $Gr^{\mbox{\tiny TP}}(2,3)$, into any given positroid cell $\S\subset \GTNN$. In our setting they rule the reality and regularity properties of the KP divisor.

Finally, we explain the transformation of the curve and the divisor both under Postnikov moves and reductions and under amalgamation of positroid cells, and apply our construction to some examples.}


\end{abstract}
\maketitle

\tableofcontents

\section{Introduction}

Totally non--negative Grassmannians $\GTNN$ are a special case of the generalization to reductive Lie groups by Lusztig  \cite{Lus1,Lus2} of the classical notion of total positivity \cite{GK,GK2,Sch,Kar}. As for  classical total positivity, $\GTNN$ naturally arise in relevant problems in different areas of mathematics and physics \cite{AGP1,AGP2,ADM,BG,CW,FPS,Lam1,OPS,Pos2,Sc}. 
In particular, the deep connection of the combinatorial structure of $\GTNN$ with KP real soliton theory was unveiled in a series of papers by Chakravarthy, Kodama and Williams (see \cite{CK,KW1,KW2} and references therein). In \cite{KW1}
it was proven that multi-line soliton solutions of the Kadomtsev-Petviashvili 2 (KP) equation are real and regular in space--time if and only if their soliton data correspond to points in the irreducible part of totally non--negative Grassmannians, whereas the combinatorial structure of the latter was used in \cite{CK,KW2} to classify the asymptotic behavior in space-time of such solutions.

In \cite{AG1,AG3} we started to investigate a connection of different nature between this family of KP solutions and total positivity in the framework of the finite-gap approach, using the fact that any such solution may also be interpreted as a potential in a degenerate spectral problem for the KP hierarchy. In particular in \cite{AG3}, for any positroid cell $\S\subset Gr^{\mbox{\tiny TNN}}(k,n)$, we used the fact that its Le--graph \cite{Pos} is dual to an $\mathtt M$--curve $\Gamma$ of genus $g$ equal to the dimension of $\S$ to associate a degree $g$ real regular divisor on $\Gamma$ to any real regular multi--line KP soliton solution whose soliton data belong to $\S$. In \cite{AG3} the construction depends explicitly on a recursion associated to a specific acyclic orientation of a special graph representing the soliton data, and it cannot be generalized to the full class of graphs representing the soliton data. Moreover we were unable to prove the invariance of the KP divisor with respect to the many gauge freedoms on the network.

In this paper we complete the research starting in \cite{AG3}, constructing an explicit map from the networks representing the soliton solutions to their spectral data fulfilling the reality and regularity conditions in \cite{DN}. We prove the invariance of the KP divisor and explain the effect of moves, reductions and amalgamation of networks on the spectral data. We remark that our construction of real regular KP divisors on $\mathtt M$--curves is completely explicit. At this aim we use the full rank geometric system of relations on the network introduced in \cite{AG4} to fix the values of the KP wave function at the nodes of the spectral curve. Since such system of relations provides the value of Postnikov boundary measurement map for any choice of positive edge weights and is modeled on the amalgamation procedure for totally non--negative Grassmannians \cite{FG1,AGP1,AGP2,Lam2}, we conjuncture that a purely cluster algebraic approach should be possible for the characterization of the KP divisor. Finally, in \cite{A3} one of the authors found a natural connection between our construction and dimer models in the disk.

Before continuing, let us briefly recall that the finite-gap approach to soliton systems was first suggested by Novikov \cite{Nov} for the Korteveg-de Vries equation, and extended to the 2+1 KP equation by Krichever in \cite{Kr1,Kr2}, where it was shown that finite-gap KP solutions correspond to non special divisors on arbitrary algebraic curves. Dubrovin and Natanzon \cite{DN} then proved that real regular KP finite gap solutions correspond to divisors on smooth $\mathtt M$--curves satisfying natural reality and regularity conditions. In \cite{Kr4} Krichever developed, in particular, the direct scattering transform for the real regular parabolic operators associated with KP and proved that the corresponding spectral curves are always $M$-curves, and divisor points are located in the ovals as in \cite{DN}. In \cite{Kr3, KV} finite gap theory was extended to reducible curves in the case of degenerate solutions. Applications of singular curves to the finite-gap integration are reviewed in \cite{Taim}.

In our setting the degenerate solutions are the real regular multiline KP solitons studied in \cite{BPPP,CK,KW1,KW2}: the real regular KP soliton data correspond to a well defined reduction of the Sato Grassmannian \cite{S}, and they are parametrized by pairs $(\mathcal K, [A])$, {\sl i.e.} $n$ ordered phases $\mathcal K =\{ \kappa_1<\kappa_2 <\cdots <\kappa_n\}$ and a point in an irreducible positroid cell $[A]\in \S \subset Gr^{\mbox{\tiny TNN}}(k,n)$. We recall that the irreducible part of $Gr^{\mbox{\tiny TNN}}(k,n)$ is the natural setting for the minimal parametrization of such solitons \cite{CK,KW2}.

Following \cite{Mal}, to the soliton data $(\mathcal K, [A])$ there is associated a rational spectral curve $\Gamma_0$ (Sato component), with a marked point $P_0$ (essential singularity of the wave function), and $k$ simple real poles $\DS =\{ \gamma_{S,r},\  r\in [k] \}$, such that $\gamma_{S,r}\in [\kappa_1,\kappa_n]$ (Sato divisor). However, due to a mismatch between the dimension of $Gr^{\mbox{\tiny TNN}}(k,n)$ and that of the variety of Sato divisors, generically the Sato divisor is not sufficient to determine the corresponding KP solution. 

In \cite{AG1, AG3} we proposed a completion of the Sato algebraic--geometric data based on the degenerate finite gap theory of \cite{Kr3} and constructed divisors on reducible curves for the real regular multiline KP solitons. In our setting, the data $(\Gamma,P_0,\mathcal D)$, where $\Gamma$ is a reducible curve with a marked point $P_0$, and $\mathcal D\subset\Gamma$ is a divisor, correspond to the soliton data $(\mathcal K, [A])$ if
\begin{enumerate}
\item $\Gamma$ contains $\Gamma_0$ as a rational component and $\DS$ coincides with the restriction of $\mathcal D$ to $\Gamma_0$ and different rational components of $\Gamma$ are connected at double points;  
\item The data $(\Gamma,P_0,\mathcal D)$ uniquely define the wave function $\hat\psi$ as a meromorphic function on $\Gamma\backslash P_0$ with divisor $\mathcal D$, having an essential singularity at $P_0$. Moreover, at double points the values of the wave function coincide on both components for all times.
\end{enumerate}
In degenerate cases, the construction of the components of the curve and of the divisor is obviously not unique and,
as pointed out by S.P. Novikov, an untrivial question is whether \textbf{real regular} soliton solutions can be obtained as rational degenerations of \textbf{real regular} finite-gap solutions. In the case of the real regular KP multisolitons this imposes the following additional requirements:
\begin{enumerate}
\item $\Gamma$ is the rational degeneration of an $\mathtt M$--curve;
\item The divisor is contained in the union of the ovals, each ``finite'' oval contains exactly one divisor point.
\end{enumerate}

In \cite{AG1} we provided an optimal answer to the above problem for the real regular soliton data in the totally positive part of the Grassmannian, $\GTP$. We proved that $\Gamma_0$ is a component of a reducible curve $\Gamma(\xi)$, with $\xi\gg 1$ a parameter which parametrizes the position of the nodes, arising as a rational degeneration of some smooth $\mathtt M$--curve of genus equal to the dimension of the positive Grassmannian, $k(n-k)$, we used classical total positivity for the algebraic part of the construction and computed explicitly the divisor positions in the ovals at leading order in $\xi$. In \cite{AG3} we started the search of an explicit relation between the combinatorial structure of $\GTNN$ and the spectral problem using the Le--graphs introduced by Postnikov \cite{Pos}. We obtained a system of relations which we solved recursively to explicitly construct both the KP wave function and the divisor.
Again our approach was constructive and in \cite{AG2} we applied it to obtain real regular finite gap solutions parametrized by real regular non special divisors on a genus 4 $\mathtt M$--curve obtained from the desingularization of spectral problem for the soliton solutions in $Gr^{\mbox{\tiny TP}}(2,4)$.

\smallskip

\paragraph{\textbf{Main results}}

The paper is divided into 3 parts.

In \textbf{Part 1. Construction of the spectral curve and the divisor} we prove that any graph representing the KP soliton data is dual to a spectral curve which is the rational degeneration of a smooth $\mathtt M$--curve, we explicitly construct real regular spectral data (wave function and divisor) in agreement with the prediction in \cite{DN}, and we prove the invariance of the KP divisor with respect to the many gauge freedoms of the associated network.

In \textbf{Part 2. Transformation properties of the divisor} we explain the transformation of such algebraic geometric data both with respect to Postnikov moves and reductions, and with respect to amalgamations of totally non-negative Grassmannians, thus establishing the first step to a purely cluster algebraic approach to the problem. 

In \textbf{Part 3. Singularities of divisors} we formulate some open problems and, in particular, we discuss the possible singularities of the divisors.

 \textbf{Part 1:} Let the soliton data $(\mathcal K, [A])$ be fixed, with $[A]\in \S \subset Gr^{\mbox{\tiny TNN}}(k,n)$, where $\S$ is an irreducible positroid cell of dimension $|D|$, and let $\mathcal G$ be a connected planar bicolored directed trivalent perfect graph in the disk representing $\S$ (Definition \ref{def:graph}) in Postnikov equivalence class \cite{Pos}. The graph $\mathcal G$ has $g+1$ faces where $g=|D|$ if the graph is reduced, otherwise $g>|D|$. 

The construction of the curve $\Gamma$ (Section \ref{sec:gamma}) generalizes that in \cite{AG3} where we treated the case of Le--graphs. $\mathcal G$ models the real part of the spectral curve and is the dual graph of a reducible curve $\Gamma$ which is the connected union of rational components. The boundary of the disk and all internal vertices of $\mathcal G$ are copies of $\mathbb{CP}^1$, the edges represent the double points where two such components are glued to each other and the faces are the ovals of the real part of the resulting reducible curve. 
We identify the boundary of the disk with the Sato component $\Gamma_0$, and the $n$ boundary vertices $b_1,\dots, b_n$ correspond to the ordered marked points, $\mathcal K = \{ \kappa_1 < \kappa_2 < \cdots < \kappa_n \}$. It is easy to check that $\Gamma$ is a rational reduction of a smooth genus $g$ $\mathtt M$--curve.
 
To extend the wave function from $\Gamma_0$ to $\Gamma$, we use a two-step procedure.
\begin{enumerate}
\item We define the wave function for all times $\vec t$ consistently at the double points of the curve;
\item We extend the wave function to all other components of $\Gamma$ assuming that they are meromorphic of degree either 1 or 0. The degree is 1 if the component corresponds to a trivalent white vertex and is 0 in all other cases.
\end{enumerate}

In Section \ref{sec:anycurve} we perform the first step using the full--rank geometric system of linear relations on the network introduced in \cite{AG4}, after assigning the not-normalized Sato wave function at the marked points $\kappa_j$ as boundary conditions on $\mathcal G$. Both the geometric and the weight gauges act on the wave function by multiplication by non-zero constants, therefore they do not affect the normalized wave function. Moreover, it is possible to explicitly solve the system of relations in terms of flows \cite{Tal2,AG7} and to compute explicitly the KP divisor in the coordinates associated to the chosen orientation. 

In Section \ref{sec:inv} we perform the second step. To each component $\Gamma_V$  corresponding to a white trivalent vertex $V$ we assign a divisor point. The KP divisor on $\Gamma$ is the union of the Sato divisor $\DS$ and of all divisor points associated to such internal vertices: $\DKP =\DS \cup \{ P_V \, , \, V \mbox{ white trivalent vertex in } \mathcal G \, \}$. Since $\DS$ has degree $k$, and the number of trivalent white vertices of a plabic graph is $g-k$, $\DKP$ has degree $g$. Since the wave function is real for real times at the double points,  $\DKP$ is contained in the union of the ovals of $\Gamma$. The normalized wave function $\hat \psi(P,\vec t)$ is then the unique meromorphic function on $\Gamma\backslash \{ P_0\}$ such that $(\hat \psi(\cdot,\vec t)) + \DKP \ge 0$, for all $\vec t$, therefore $\hat \psi$ is the KP wave function on $\Gamma$ for the soliton data $(\mathcal K, [A])$ (Theorem \ref{lemma:KPeffvac}).

In Theorem \ref{theo:inv} we prove that $\DKP$ is invariant with respect to changes of orientation and of the choice of gauge ray, weight and vertex gauges. In Lemma \ref{lem:pos_div} we detect the oval to which each divisor point $P_V$ belongs to, whereas in Section \ref{sec:comb} we prove that each finite oval contains exactly one divisor point, {\sl i.e.} $\DKP$ satisfies the reality and regularity conditions established in \cite{DN}.
As a consequence, we obtain a direct
relation between the total non--negativity property encoded in the geometrical setting of \cite{AG4} and the reality and regularity condition of the divisor studied in this paper. 

\textbf{Part 2:} In Section \ref{sec:moves_reduc} we give the explicit transformation rules of the curve, the wave function and the divisor with respect to Postnikov moves and reductions. We also present some examples; in Section \ref{sec:example} we apply our construction to soliton data in ${\mathcal S}_{34}^{\mbox{\tiny TNN}}$, the 3--dimensional positroid cell in $Gr^{\mbox{\tiny TNN}} (2,4)$ corresponding to the matroid ${\mathcal M} = \{ \ 12 \ , \ 13 \ , \ 14 \ ,\ 23 \ , \ 24\ \}$.
We construct both the reducible rational curve and its desingularization to a genus $3$ $\mathtt M$--curve and the KP divisor for generic soliton data $\mathcal K =\{ \kappa_1<\kappa_2<\kappa_3<\kappa_4\}$ and $[A]\in {\mathcal S}_{34}^{\mbox{\tiny TNN}}$. We then apply a parallel edge unreduction and a flip move and compute the divisor on the transformed curve.
We also show the effect of the square move on the divisor for soliton data $(\mathcal K,[A])$ with $[A] \in Gr^{\mbox{\tiny TP}} (2,4)$ in Section \ref{sec:ex_Gr24top}. In Section~\ref{sec:amalg} we study the effect of amalgamation on the signatures and divisor data.

\textbf{Part 3:}  The parametrization of a given positroid cell $\S$ via KP-II divisors constructed in this paper is local in the following sense: for each point in $\S$ and a collection of phases $\mathcal K$, we choose a fixed time $\vec t_0$ such that near this point the parametrization is locally regular. Globally 3 possible cases may occur: 
\begin{enumerate}
\item There exists a time $\vec t_0$ such that a pair of divisor point are on the same node, but for generic $\vec t$ the divisor is generic. In this case it is necessary to apply an appropriate blow-up procedure to resolve the singularity. In the case of reduced graphs this is the only degeneration which may occur. We plan to study this problem in a future paper.  Here in Section~\ref{sec:global} we solve this problem in the simplest non-trivial case  $Gr^{TP}(1,3)$. 
\item There exists a collection of positive weights such that for any time $\vec t$ a pair of divisor point are on the same node, but for generic collection of weights and generic $\vec t$ the divisor is generic. This situation may occur for the reducible graphs studied in this paper. We briefly discuss this case in Section~\ref{sec:constr_null}. 
\item For a given graph, a pair of divisor point are on the same node for any collection of positive weights and any time $\vec t$ . This situation may occur only if we release the condition that for any edge there exists a path from boundary to boundary containing it. We present an example in Section~\ref{sec:counterexample}. 
\end{enumerate}

\paragraph{\textbf{Remarks and open questions}}

\smallskip

Below we list few more open questions.

Our construction may be considered as a tropicalization of the spectral problem (smooth $\mathtt M$--curves and divisors) associated to real regular finite--gap KP solutions (potentials) in the rational degeneration of such curves. An interesting open question is whether all smooth $\mathtt M$--curves may be obtained starting from our construction. At this aim it should be relevant to investigate the relation of our construction with the KP tropical limit studied in \cite{AFMS}. 

The tropical limit studied in \cite{KW2} (see also \cite{DMH} for a special case) has a different nature: reconstruct the soliton data from the asymptotic contour plots. In our setting, that would be equivalent to tropicalize the reducible rational spectral problem connecting the asymptotic behavior of the potential (KP solution) to the asymptotic behavior in $\vec t$ of the zero divisor of the KP wave function (see \cite{A2} for some preliminary results concerning soliton data in $Gr^{\mbox{\tiny TP}}(2,4)$).

Relations between integrability and cluster algebras were demonstrated in \cite{FG,KG}, and the cluster algebras were essentially motivated by total positivity \cite{FZ1,FZ2}. In \cite{KW2} cluster algebras have appeared in connection with KP solitons asymptotic behavior. We expect that they should also appear in our construction in connection with the tropicalization of the zero divisor. Moreover a deep relation of (degenerate) KP solutions with cluster algebras is also suggested by the fact that the geometric systems of relations which encode the position of the divisor have a natural interpretation as amalgamation of small positive Grassmannians respecting total non--negativity \cite{FG1,Lam2, AG4}.

For a fixed reducible curve the Jacobian may contain more than one connected component associated to real regular solutions. Therefore, in contrast with the smooth case, different connected components may correspond to different Grassmannians. Some of these components may correspond not to full positroid cells, but to special subvarieties. For generic curves the problem of describing these subvarieties is completely open. For a rational degeneration of genus $(n-1)$ hyperelliptic $\mathtt M$-curves this problem was studied in \cite{A1} and it was shown that the corresponding soliton data in $Gr^{\mbox{\tiny TP}}(k,n)$ formed $(n-1)$--dimensional varieties known in literature \cite{BK} to be related to the finite open Toda system. The same KP soliton family has been recently re-obtained in \cite{Nak} in the framework of the Sato Grassmannian, whereas the spectral data for the finite Toda was studied earlier in \cite{KV}. 

Moreover, all results valid in the KP hierarchy also go through for its reductions such as the KdV and the Boussinesq hierarchies.  Therefore a natural question is to classify the subvarieties in $Gr^{\mbox{\tiny TNN}}(k,n)$ associated to such relevant reductions. We remark that a similar problem was addressed in \cite{KX} for complex KP soliton solutions.

In \cite{AG2} we studied in detail the transition from multiline soliton solutions to finite-gap solutions associated to almost degenerate $\mathtt M$-curves in the first non-trivial case. We expect that the coordinates on the moduli space, compatible with $\mathtt M$-structure, introduced in \cite{Kr5}, may be useful in this study. 

It is an open problem whether all real and regular divisor positions in the ovals are realizable as the soliton data vary in $\S$ for a given normalizing time $\vec t_0$. The latter problem is naturally connected to the classification of realizable asymptotic soliton graphs studied in \cite{KW2}. 

Finally, similar gluing problems of little Grassmannians expressed as compatibility of linear systems at vertices respecting the total positivity property appear also in several different problems, such as the momentum--helicity conservation relations in the on--shell amplituhedron problem for the $N=$ SYM in \cite{AGP1,AGP2, Lam2} and the geometry of polyhedral subdivisions \cite{Pos2,PSW}. It is unclear to us whether and how our approach for KP may be related to these questions.

\smallskip

{\bf Notations:} We use the following notations throughout the paper:
\begin{enumerate}
\item $k$ and $n$ are positive integers such that $k<n$;
\item  For $s\in {\mathbb N}$ let $[s] =\{ 1,2,\dots, s\}$; if $s,j \in {\mathbb N}$, $s<j$, then
$[s,j] =\{ s, s+1, s+2,\dots, j-1,j\}$;
\item  ${\vec t} = (t_1,t_2,t_3,\dots)$ is the infinite vector of real KP times where $t_1=x$, $t_2=y$, $t_3=t$, and we assume that only a finite number of components are different from zero;
\item We denote $\theta(\zeta,\vec t)= \sum\limits_{s=1}^{\infty} \zeta^s t_s$, due to the previous remark  $\theta(\zeta,\vec t)$ is well-defined for any complex $\zeta$;
\item We denote the real KP phases 
$\kappa_1< \kappa_2 < \cdots < \kappa_n$ and
$\theta_j \equiv \theta (\kappa_j, \vec t)$.
\end{enumerate}

\part{Construction of the spectral curve and the divisor}

In this Part we associate to a plabic network the corresponding algebro-geometircal KP data, namely a reducible rational $\mathtt M$-curve and a divisor on it, satisfying the reality and regularity conditions. 

\section{Systems of relations on plabic networks and totally non-negative Grassmannians}\label{sec:vectors}

Following \cite{Pos} we parametrize totally non-negative Grassmannians in terms of planar bicolored directed trivalent perfect networks (plabic networks). In our construction these graphs are dual to the spectral curves on which the KP wave functions are defined. These spectral curves associated to the real regular multisoliton spectral data are reducible $\mathtt M$--curves. The KP wave function on such curve may be uniquely reconstructed from its values at the double points of the curve, which correspond to the edges of the graph. These values are defined through linear relations at the vertices of the network with boundary conditions fixed by the soliton data. Such relations were introduced in \cite{Lam2}, where the problem of characterizing relations respecting the total non-negativity was posed. In \cite{AG4} we provided a geometric solutions of this problem. Another approach to the same problem was developed in \cite{AGPR} using Kasteleyn theorem. Relations between these two approaches were discussed in \cite{A3}. In this Section we briefly recall the geometric approach, and in Section~\ref{sec:3} we apply it to construct the wave function and the divisor.

\subsection{Totally non--negative Grassmannians and plabic networks in the disk}\label{sec:plabic_graphs}

\begin{definition}\textbf{Totally non-negative Grassmannian \cite{Pos}.}
Let $Mat^{\mbox{\tiny TNN}}_{k,n}$ denote the set of real $k\times n$ matrices of maximal rank $k$ with non--negative maximal minors $\Delta_I (A)$. Let $GL_k^+$ be the group of $k\times k$ matrices with positive determinants. Then the totally non-negative Grassmannian $\GTNN$ is 
\[
\GTNN = GL_k^+ \backslash Mat^{\mbox{\tiny TNN}}_{k,n}.
\]
\end{definition}

In the theory of totally non-negative Grassmannians an important role is played by the positroid stratification. Each cell in this stratification is defined as the intersection of a Gelfand-Serganova stratum \cite{GS,GGMS} with the totally non-negative part of the Grassmannian. More precisely:
\begin{definition}\textbf{Positroid stratification \cite{Pos}.} Let $\mathcal M$ be a matroid i.e. a collection of $k$-element ordered subsets $I$ in $[n]$, satisfying the exchange axiom (see, for example \cite{GS,GGMS}). Then the positroid cell $\S$ is defined as
$$
\S=\{[A]\in \GTNN\ | \ \Delta_{I}(A) >0 \ \mbox{if}\ I\in{\mathcal M} \ \mbox{and} \  \Delta_{I}(A) = 0 \ \mbox{if} \ I\not\in{\mathcal M}  \}.
$$
A positroid cell is irreducible if, for any $j\in [n]$, there exist $I, J\in \mathcal M$ such that $j\in I$ and $j\not\in J$. 
\end{definition}
The combinatorial classification of all non-empty positroid cells and their rational parametrizations were obtained in \cite{Pos}, \cite{Tal2}. In our construction we use the classification of positroid cells via directed planar networks in the disk in \cite{Pos}. Irreducible positroid cells play the relevant role in the applications to real regular KP soliton solutions, since they provide the minimal realization of such solutions in totally non--negative Grassmannians \cite{CK, KW2} (see also Section \ref{sec:soliton_theory}).

In the following we restrict ourselves to \textbf{irreducible positroid cells} $\S$ and we consider planar bicolored directed trivalent perfect graphs in the disk (plabic graphs) representing $\S$. Irreducibility of the positroid cell implies that the graph has neither isolated boundary sinks nor isolated boundary sources. We assume that the graphs satisfy the following additional assumption (Item \ref{cond:sign} in the definition below): for any internal edge $e$ there exists a directed path containing $e$ and starting at a boundary source and ending at a boundary sink. Such assumption is essential to guarantee that amalgamation of the little positive Grassmannians $Gr^{\mbox{\tiny TP}}(1,3)$ and $Gr^{\mbox{\tiny TP}}(2,3)$ preserving the total non--negativity is fully controlled by geometric signatures and that all signatures compatible with the total non--negative property are of geometric type. Let us remark that the graph can be reducible.

More precisely, we consider the following class of graphs ${\mathcal G}$:
\begin{definition}\label{def:graph} \textbf{Planar bicolored directed trivalent perfect graphs in a convex topological disk (plabic graphs).} A graph ${\mathcal G}$ is called plabic if it has the following properties:
\begin{enumerate}
\item  ${\mathcal G}$ is planar, directed and lies inside a disk. Moreover ${\mathcal G}$ is connected in the sense that it does not possess components isolated from the boundary;
\item It has finitely many vertices and edges;
\item It has $n$ boundary vertices on the boundary of the disk labeled $b_1,\cdots,b_n$ clockwise. Each boundary vertex has degree 1. $b_i$ is a source (respectively sink) if its edge is outgoing (respectively incoming). The remaining vertices are called internal and are located strictly inside the disk. It is not restrictive to assume that they have valency either 2 or 3; 
\item ${\mathcal G}$ is a perfect graph, that is each internal vertex in  ${\mathcal G}$ is incident to exactly one incoming edge or to one outgoing edge;\label{perf:5}
\item Each vertex is colored black or white. If a trivalent vertex has only one incoming edge, it is colored white, otherwise, it is colored black. Bivalent vertices are assigned either white or black color; \label{perf:6}
\item For any internal edge $e$ there exists a directed path containing $e$ and starting at a boundary source and ending at a boundary sink. In particular the graph has no internal sources or sinks. \label{cond:sign} 
\end{enumerate}
Moreover, to simplify the overall construction we further assume that the boundary vertices $b_j$, $j\in [n]$, lie on a common straight interval in the boundary of the disk. 
\end{definition}
\begin{definition}
An orientation on given plabic graph ${\mathcal G}$ is called \textbf{perfect} if it respects properties (\ref{perf:5}), (\ref{perf:6}), i.e. it respects perfectness of  ${\mathcal G}$. 
\end{definition}
We consider in the following different orientations on the same graph, but we assume that all of them are perfect with respect to the fixed coloring.
\begin{remark}
It is easy to check that if Property~(\ref{cond:sign}) is satisfied for some perfect orientation, then it is satisfied for all  perfect orientations of ${\mathcal G}$.
\end{remark}

\begin{remark}
In our construction only the bivalent vertices joined by edges to two boundary vertices have to be kept. All other bivalent vertices may be eliminated using Postnikov moves and are irrelevant in the construction of the KP divisor.
\end{remark}

In Figure \ref{fig:Rules0} we present an example of a plabic graph satisfying Definition~\ref{def:graph} and representing a 10-dimensional positroid cell in $Gr^{\mbox{\tiny TNN}}(4,9)$. 

The graph is of type $(k,n)$ if it has $n$ boundary vertices and $k$ of them are boundary sources. Any choice of perfect orientation preserves the type of ${\mathcal G}$. To any perfect orientation $\mathcal O$ of ${\mathcal G}$ one assigns the base $I_{\mathcal O}\subset [n]$ of the $k$-element source set for $\mathcal O$. Following \cite{Pos} the matroid of ${\mathcal G}$ is the set of $k$-subsets  $I_{\mathcal O}$ for all perfect orientations:
$$
\mathcal M_{\mathcal G}:=\{I_{\mathcal O}|{\mathcal O}\ \mbox{is a perfect orientation of}\ \mathcal G \}.
$$

\begin{definition}\textbf{Plabic network.}
 A plabic network is a perfectly oriented plabic graph $({\mathcal G},{\mathcal O}(I))$ together with an assignment of non-zero weights $w_e$ to the edges.  
\end{definition}

\begin{definition}\textbf{Postnikov boundary measurement map \cite{Pos}.}

For any given oriented planar network in the disk the formal boundary measurement map is defined by:
\begin{equation}\label{eq:bmm}
M_{ij} := \sum\limits_{P:b_i\mapsto b_j} (-1)^{\mbox{\scriptsize Wind}(P)} w(P),
\end{equation}
where the sum is over all directed walks from the source $b_i$ to the sink $b_j$, $w(P)$ is the product of the edge weights of $P$ and $\mbox{Wind}(P)$ is its topological winding index defined in \cite{Pos}.
\end{definition}

These formal power series sum up to subtraction free rational expressions in the weights \cite{Pos} and explicit expressions in function of flows and conservative flows in the network is provided in \cite{Tal2}. Let $I$ be the base inducing the orientation of $\mathcal N$ used in the computation of the boundary measurement map. Then the point $Meas(\mathcal N)\in Gr(k,n)$ is represented by the boundary measurement matrix $A$ such that:
\begin{itemize}
\item The submatrix $A_I$ in the column set $I$ is the identity matrix;
\item The remaining entries $A^r_j = (-1)^{\sigma(i_r,j)} M_{ij}$, $r\in [k]$, $j\in \bar I$, where $\sigma(i_r,j)$ is the number of elements of $I$ strictly between $i_r$ and $j$.
\end{itemize}

In \cite{Pos} it is proven that $\mathcal M_{\mathcal G}$ is a positroid: for any choice of real positive weights on the edges of $\mathcal G$ the image of Postnikov boundary measurement map on the directed network $\mathcal N$ of graph $\mathcal G$ is a point in $\mathcal S^{\mbox{\tiny TNN}}_{\mathcal M_{\mathcal G}}$, i.e. the matrix $A$ is totally non-negative. Vice versa, for any point $[A]\in \mathcal S^{\mbox{\tiny TNN}}_{\mathcal M_{\mathcal G}}$ there is a choice of positive weights so that the resulting network of graph $\mathcal G$ represents $[A]$. Indeed, if one changes the perfect orientation of the network and simultaneously changes the weights to the reciprocal on the edges changing the direction, then the Grassmannian point remains invariant. In addition, the Grassmannian point is preserved by the following weight gauge transformation.

\begin{remark}\label{rem:gauge_weight}\textbf{The weight gauge freedom \cite{Pos}.} Given a point $[A]\in\S$
and a planar directed graph ${\mathcal G}$ in the disk representing $\S$, then $[A]$
 is represented by infinitely many gauge equivalent systems of weights $w_e$ on the edges $e$ of ${\mathcal G}$. Indeed, if a positive number $t_V$ is assigned to each internal vertex $V$, whereas $t_{b_i}=1$ for each boundary vertex $b_i$, then the transformation on each directed edge $e=(U,V)$
\begin{equation}
\label{eq:gauge}
w_e\rightarrow w_e t_U \left(t_V\right)^{-1},
\end{equation}
transforms the given directed network into an equivalent one representing $[A]$.
\end{remark}

In~\cite{Pos} it was shown that the boundary measurement map is one-to-one and onto map from the set of equivalence classes of positive weights to the corresponding positroid cell if the graph is reduced. 

Plabic graphs $\mathcal G$ representing the same positroid cell $\S$ are equivalent via a finite sequence of moves and reductions \cite{Pos}. More precisely, a plabic graph $\mathcal G$ can be transformed into a plabic graph $\mathcal G'$ via a finite sequence of Postnikov moves and reductions if and only if $\mathcal M_{\mathcal G}=\mathcal M_{\mathcal G'}$.

A graph  $\mathcal G$ is reduced if there is no other graph in its move reduction equivalence class which can be obtained from $\mathcal G$ applying a sequence of transformations containing at least one reduction. Each positroid cell $\S$ is represented by at least one reduced graph, the so called Le--graph, associated to the Le--diagram representing $\S$ and it is possible to assign weights to such graphs in order to obtain a global parametrization of $\S$ \cite{Pos}.

Each Le-graph is reduced, and, if $\S$ is irreducible, it is a reduced plabic graph. 
If $\mathcal G$ is a reduced  plabic graph, then the dimension of  $\mathcal S^{\mbox{\tiny TNN}}_{\mathcal M_{\mathcal G}}$ is equal to the number of faces of $\mathcal G$ minus 1. The plabic graph in Figure \ref{fig:Rules0} is reduced.

\subsection{Systems of relations on plabic networks: definition and explicit representation}\label{sec:def_edge_vectors}

In \cite{Lam2} Lam proposed to parametrize the boundary measurement map with complex weights using spaces of relations on bicolored graphs introducing formal half--edge variables $z_{U,e}$ which satisfy the following system on the graph:
\begin{definition}\textbf{Lam relations \cite{Lam2}}\label{def:lam}
Let $\mathcal G$  be a plabic graph, with either real or complex weights $w_{U,V}$ assigned to the oriented edges $e=(U,V)$, let $\varepsilon_{U,V}$ be a function on the  directed edges taking values in $\{0,1\}$, and let $\mathcal W$ be some vector space. Consider a system of vectors  $z_{U,e}\in{\mathcal W}$ assigned to half-edges, i.e. pairs $(U,e)$, where $U$ is a vertex and $e$ is an edge incident to $U$.  The system $z_{U,e}$ satisfy Lam relations if 
  
\begin{enumerate}
\item For any edge $e=(U,V)$, $z_{U,e} = (-1)^{\varepsilon_{U,V}}  w_{U,V} z_{V,e}$. If we reverse the orientation,  $z_{V,e} = (-1)^{\varepsilon_{U,V}}  w_{U,V}^{-1} z_{U,e}$;
\item If $e_i$, $i\in [m]$, are the edges at an $m$-valent white vertex $V$, then $\sum_{i=1}^m z_{V,e_i} =0$;
\item If $e_i$, $i\in [m]$, are the edges at an $m$-valent black vertex $V$, then $z_{V,e_i} =z_{V,e_j}$ for all $i,j\in[m]$.
\end{enumerate}
\end{definition}

In \cite{Lam2} it is conjectured that there exist simple rules to assign signatures $\varepsilon_{U,V}$ so that the above system has full rank for any choice of positive weights, and the image of this weighted space of relations is the positroid cell $\S\subset\GTNN$ corresponding to the graph.
If such a signature exists, it is not unique because the system has big gauge freedom.

\begin{definition}\textbf{Equivalence between edge signatures} \label{def:admiss_sign_gauge}
Let $\epsilon^{(1)}_{U,V}$ and $\epsilon^{(2)}_{U,V}$ be two signatures on all the edges $e=(U,V)$ of the plabic graph $\mathcal G$ with the same perfect orientation. Two signatures are equivalent if there exists an index $\eta(U) \in \{ 0,1\}$ at each internal vertex $U$ such that modulo 2
\begin{equation}\label{eq:equiv_sign}
\epsilon^{(2)}_{U,V} = \left\{\begin{array}{ll}
\epsilon^{(1)}_{U,V} +\eta(U)+\eta(V), & \mbox{ if } e=(U,V) \mbox{ is an internal edge},\\
        \epsilon^{(1)}_{U,b} +\eta(U), & \mbox{ if } e=(U,b) \mbox{ is the edge at the boundary sink } b, \\ 
        \epsilon^{(1)}_{b,V} +\eta(V), & \mbox{ if } e=(b,V) \mbox{ is the edge at the boundary source } b.                                                       
\end{array}\right.
\end{equation} 
\end{definition}

\begin{lemma}\label{lem:gauge}
If $\epsilon^{(1)}_{U,V}$ and $\epsilon^{(2)}_{U,V}$ are two equivalent signatures on a plabic graph, then for a given collection of weights Lam system of relations has full rank with respect to  $\epsilon^{(1)}_{U,V}$ if and only if it has full rank  with respect to  $\epsilon^{(2)}_{U,V}$. Moreover, for the same boundary conditions at the boundary sinks they define the same solution at the boundary sources.

If $w_e^{(1)}$ and $w_e^{(2)}$ are two gauge equivalent collections of weights  on a plabic graph $\mathcal G$, then Lam system of relations has full rank with respect to $w_e^{(1)}$  if and only if it has full rank  with respect to  $w_e^{(2)}$. Moreover, for the same boundary conditions at the boundary sinks they define the same solution at the boundary sources.
\end{lemma}

In our paper \cite{AG4} we introduced geometric signatures, we proved the invariance of face signatures with respect to the gauge freedoms of network, and we explained the relation between geometric signatures and non-negativity of the boundary measurement map. Below we recall all these results.

To define the geometric signature $\varepsilon_{U,V}$, we fix a gauge by choosing a gauge ray direction and introducing local winding and intersection numbers.
\begin{figure}
  \centering{\includegraphics[width=0.5\textwidth]{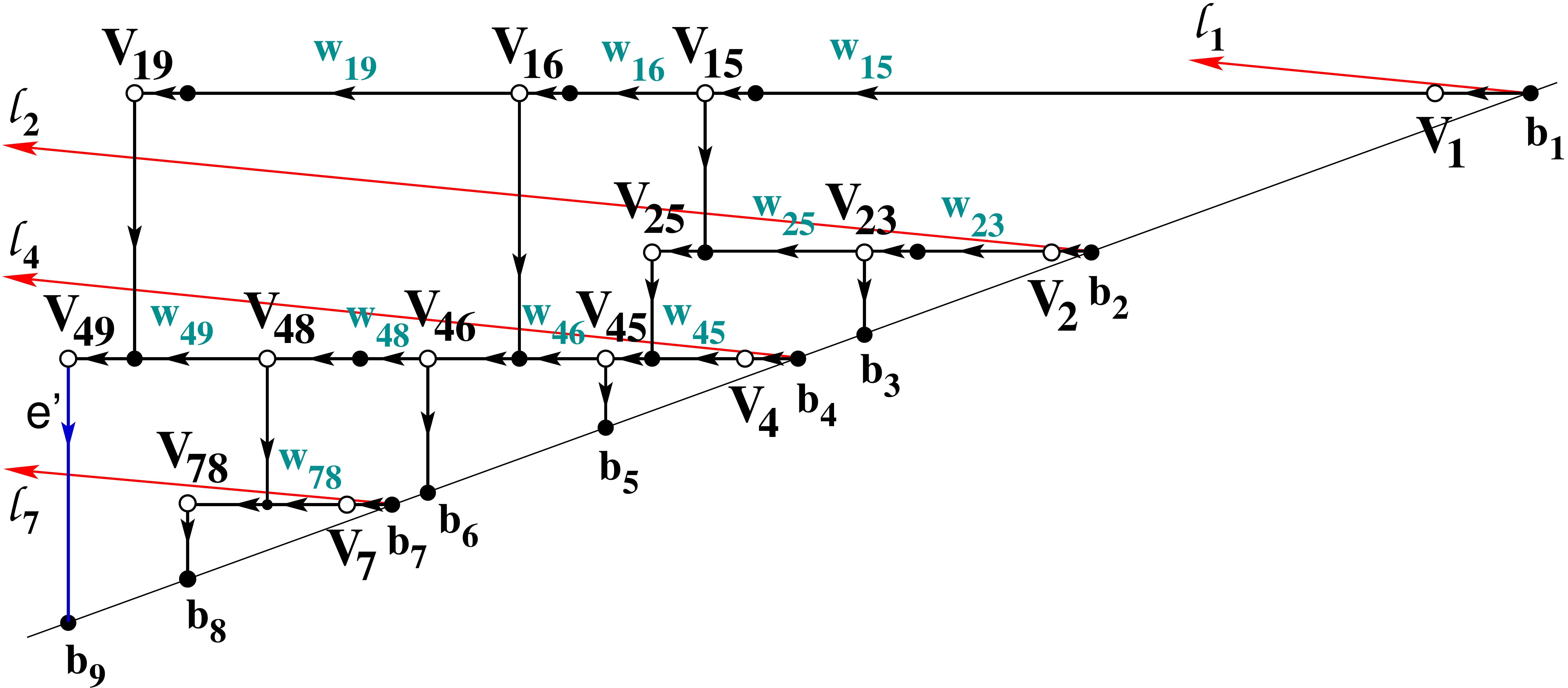}
      \caption{\small{\sl The rays starting at the boundary sources for a given orientation of the 
          network uniquely fix the system of relations.
					}}\label{fig:Rules0}}
\end{figure}
\begin{definition}\label{def:gauge_direction}\textbf{The gauge ray direction $\mathfrak{l}$.}
A gauge ray direction is an oriented direction ${\mathfrak l}$ with the following properties:
\begin{enumerate}
\item The ray with the direction ${\mathfrak l}$ starting at a boundary vertex points inside the disk; 
\item No internal edge is parallel to this direction;
\item All rays starting at boundary vertices do not contain internal vertices.
\end{enumerate}

\end{definition}
\begin{definition}\label{def:gauge_ray}\textbf{The gauge rays at the boundary sources.}
  Given a perfect orientation ${\mathcal O}(I)$ on the graph and a gauge ray direction as in Definition~\ref{def:gauge_direction}, we assign the gauge ray ${\mathfrak l}_s$ starting at source $b_s$ in the direction  ${\mathfrak l}$ to the boundary source $b_s$.  (see Fig.~\ref{fig:Rules0}).
\end{definition}

\begin{remark}
Gauge ray directions were used in \cite{GSV} to measure the local winding number. In \cite{AG4} we introduced the intersections of gauge rays with a given path as an analog of the index $\sigma(i_r,j)$ for internal edges.  
\end{remark}

The local winding number between a pair of consecutive edges $e_k,e_{k+1}$ is then defined as follows. 

\begin{definition}\label{def:winding_pair}\textbf{The local winding number at an ordered pair of oriented edges}
For an ordered pair $(e_k,e_{k+1})$ of oriented edges, define
\begin{equation}\label{eq:def_s}
s(e_k,e_{k+1}) = \left\{
\begin{array}{ll}
+1 & \mbox{ if the ordered pair is positively oriented }  \\
0  & \mbox{ if } e_k \mbox{ and } e_{k+1} \mbox{ are parallel }\\
-1 & \mbox{ if the ordered pair is negatively oriented }
\end{array}
\right.
\end{equation}
Then the winding number of the ordered pair $(e_k,e_{k+1})$ with respect to the gauge ray direction $\mathfrak{l}$ is
\begin{equation}\label{eq:def_wind}
\mbox{wind}(e_k,e_{k+1}) = \left\{
\begin{array}{ll}
+1 & \mbox{ if } s(e_k,e_{k+1}) = s(e_k,\mathfrak{l}) = s(\mathfrak{l},e_{k+1}) = 1\\
-1 & \mbox{ if } s(e_k,e_{k+1}) = s(e_k,\mathfrak{l}) = s(\mathfrak{l},e_{k+1}) = -1\\
0  & \mbox{otherwise}.
\end{array}
\right.
\end{equation}
\end{definition}

\begin{definition}\label{def:intersect_number}\textbf{The intersection number at an oriented edge}. 
  Given a perfect orientation ${\mathcal O}(I)$ on the graph and a gauge ray direction, the intersection number
  $\mbox{int}(e)$ for an edge $e$ is the number of intersections of the gauge rays  ${\mathfrak l}_s$ starting at the boundary sources $b_s$ with $e$. To the intersection of $l_s$ with $e$ we assign $+1$ if the pair $({\mathfrak l},e)$ is positively oriented, and $-1$ otherwise. 
\end{definition}
\begin{figure}
  \centering{\includegraphics[width=0.4\textwidth]{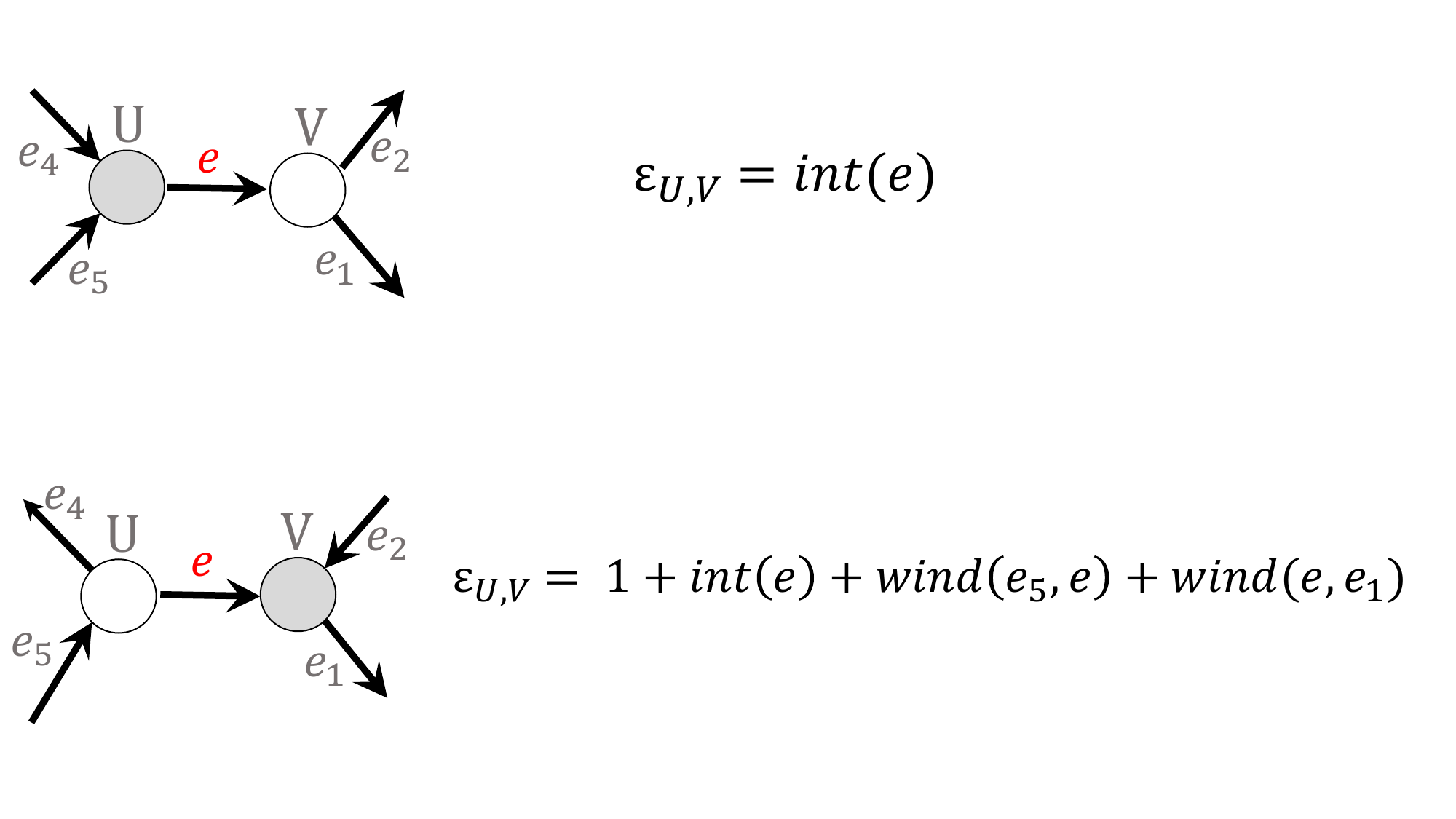}\hfill\includegraphics[width=0.4\textwidth]{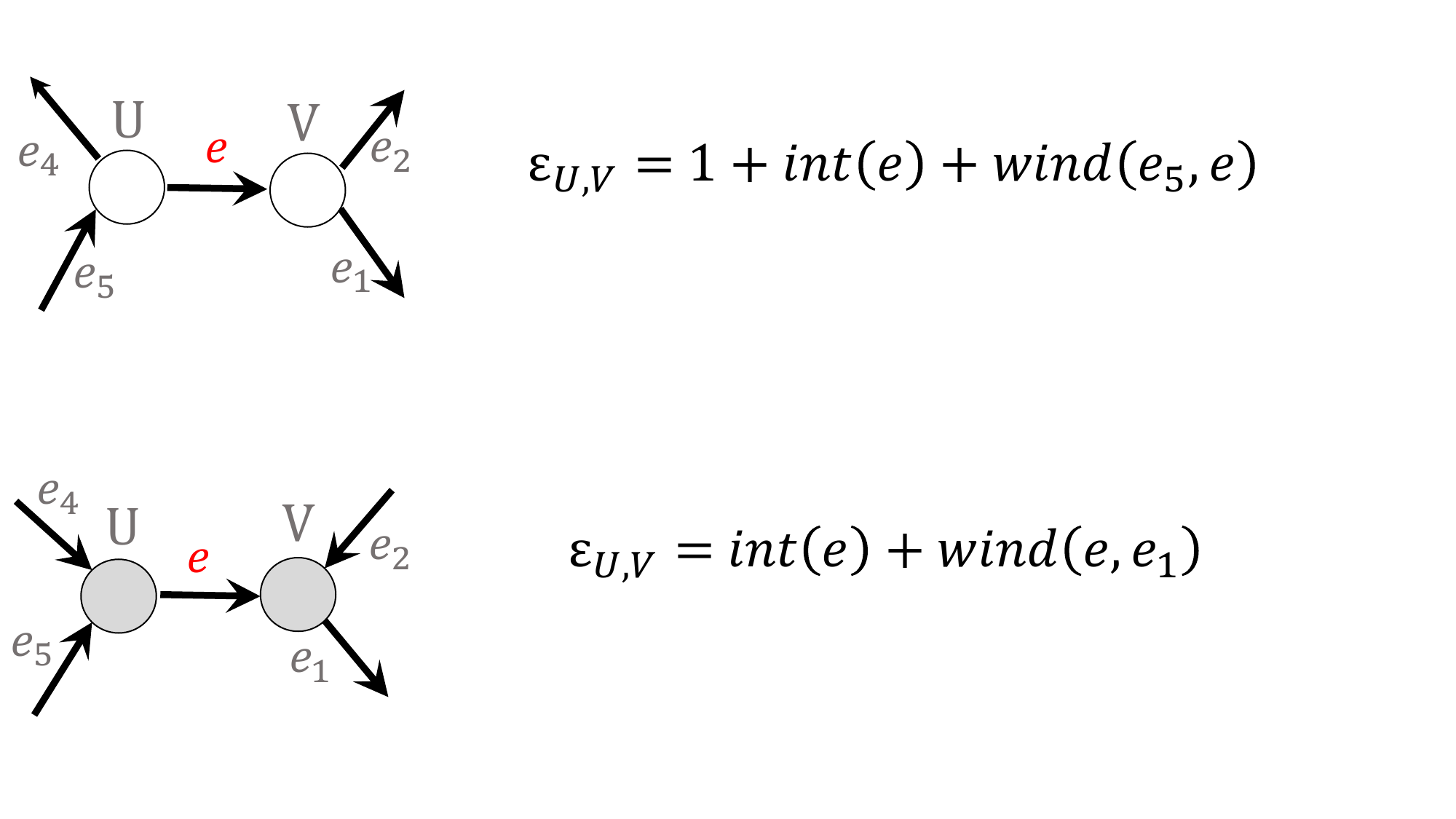}
    \includegraphics[width=0.4\textwidth]{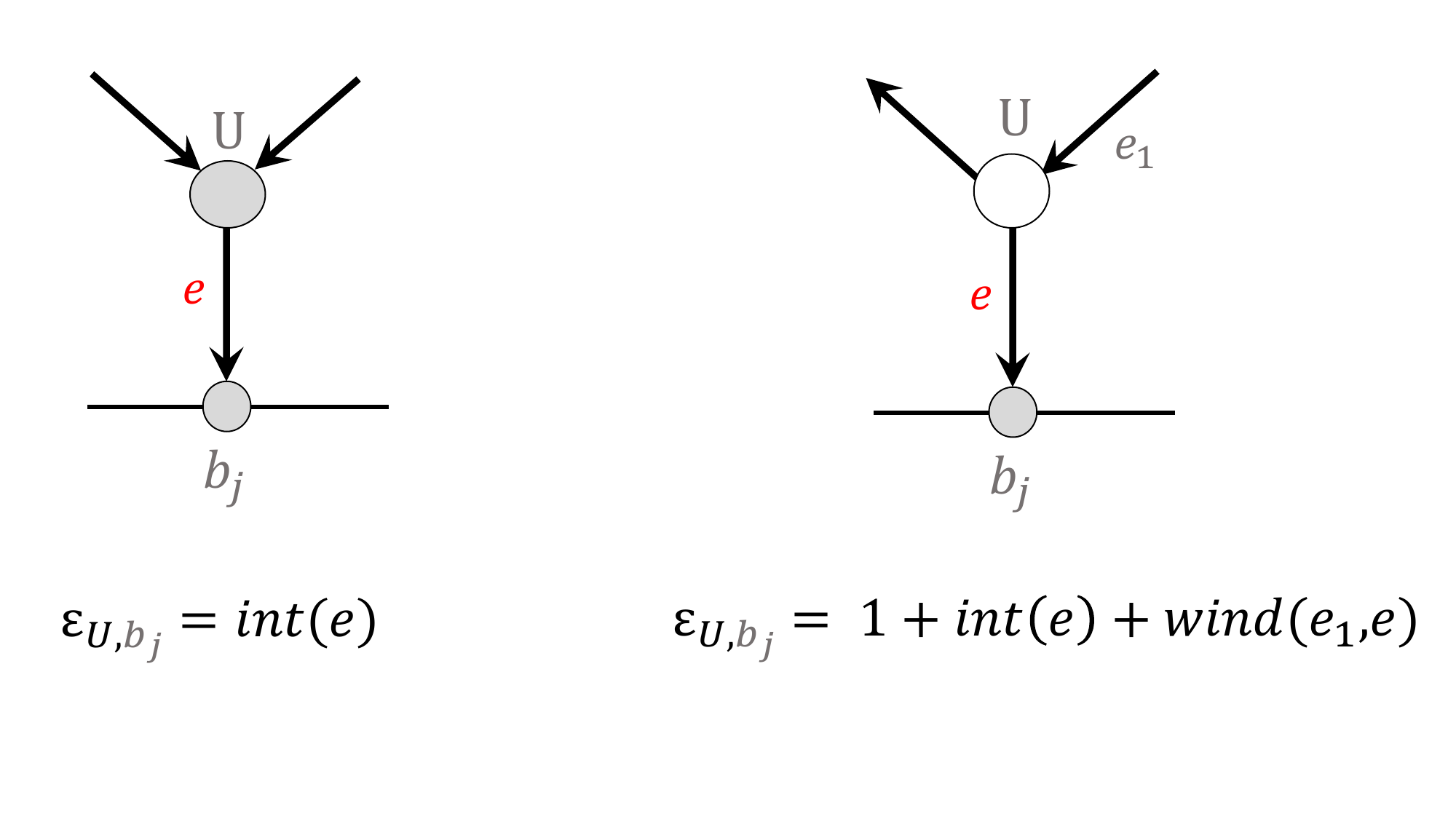}\hfill\includegraphics[width=0.4\textwidth]{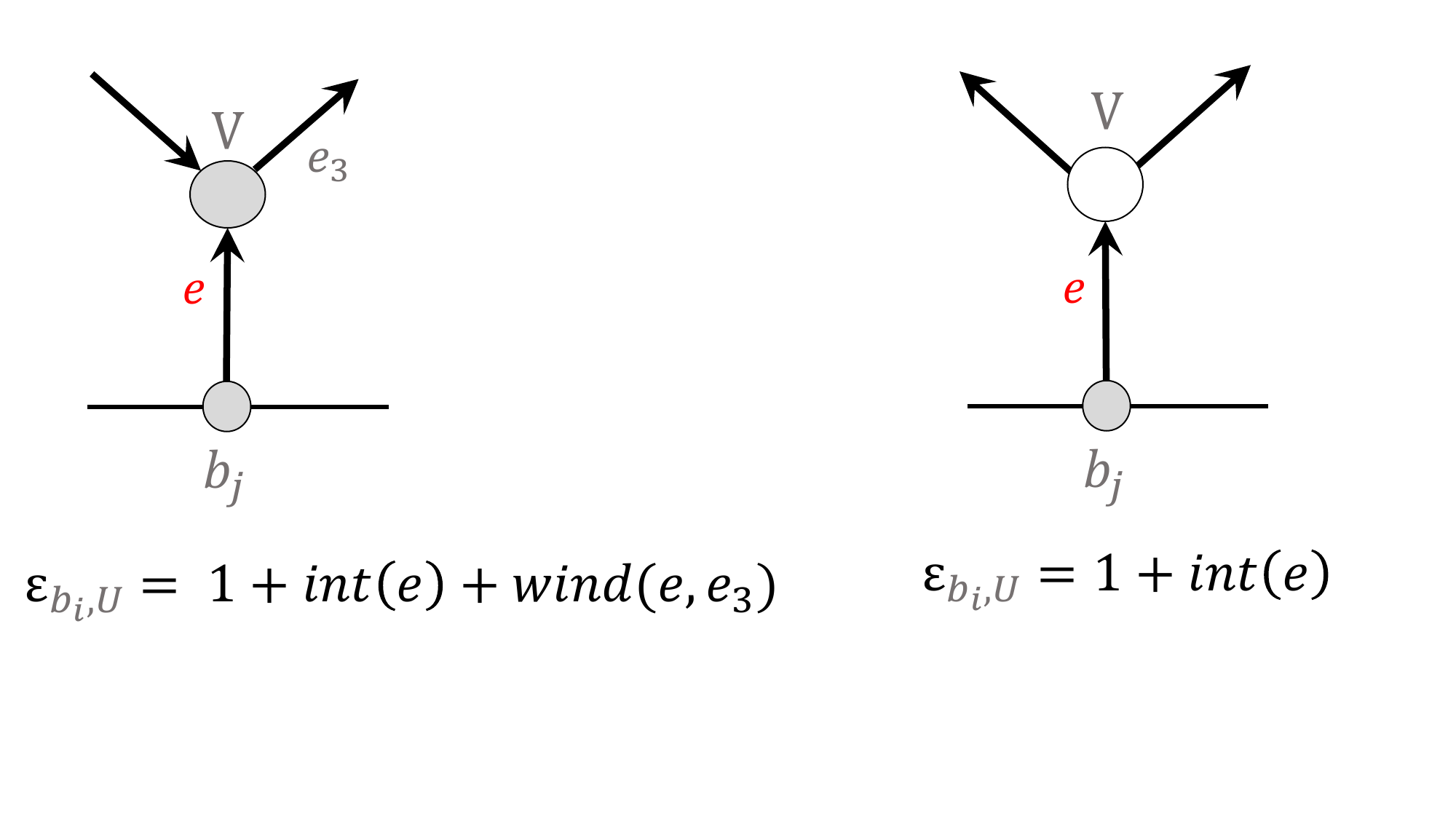}
		\vspace{-1 truecm}
      \caption{\small{\sl Computation of the geometric signature in Definition~\ref{def:geo_sign}.
					}}\label{fig:geo_sig}}
\end{figure}
\begin{definition}\textbf{The geometric signature on $(\mathcal G, \mathcal O, \mathfrak l)$.}\label{def:geo_sign} 
Let $(\mathcal G, \mathcal O, \mathfrak l)$ be a plabic graph representing a $|D|$--dimensional positroid cell $\S \subset \GTNN$ with perfect orientation $\mathcal O$ associated to the base $I$ and gauge ray direction $\mathfrak l$. 

We call geometric signature on  $(\mathcal G, \mathcal O, \mathfrak l)$ the following function $\epsilon_e $ on the directed edges of the graph taking values in $\{0,1\}$ (all formulas in this Definition are modulus 2):

\begin{enumerate}
\item  If $e=(U,b_j)$ is the edge at the boundary sink $b_j$, $j\in \bar{I}$, then
\begin{equation}
\label{eq:lin_lam_1.0.5}
\epsilon_{U,b_j}=\left\{\begin{array}{ll} \mbox{int}(e), & \mbox{ if } U \mbox{ is black,} \\ 
1+ \mbox{int}(e) + \mbox{wind}(e_1,e), & \mbox{ if } U  \mbox{ is white and } e_1 \mbox{ is incoming at } U;
                        \end{array} \right.          
                    \end{equation}
                    
\item If $e=(b_i,V)$ is the edge at the boundary source $b_i$, $i\in I$, then
\begin{equation}
\label{eq:lin_lam_1.1}
\epsilon_{b_i,V}=\left\{\begin{array}{ll} 1+ \mbox{int}(e)+\mbox{wind}(e,e_3), & \mbox{ if } V \mbox{ is black and }  e_3 \mbox{ is outgoing at } V,\\ 
1 + \mbox{int}(e), & \mbox{ if } V  \mbox{ is white};
      \end{array} \right.
\end{equation}
\item If $e=(U,V)$ is an internal edge, then
\begin{equation}\label{eq:lam_corr_edge}
\resizebox{\textwidth}{!}{$ 
\epsilon_{U,V}= \left\{ \begin{array}{ll}
\mbox{int}(e), & \mbox{ if } U \mbox{ black and } V { white};\\
1+\mbox{int}(e)+\mbox{wind}(e_5,e), & \mbox{ if } U, V \mbox{ white and } e_5 \mbox{ incoming at } U;\\
1+\mbox{int}(e)+\mbox{wind}(e_5,e)+\mbox{wind}(e,e_1), & \mbox{ if } e_5 \mbox{ incoming at } U \mbox{ white and } e_1 \mbox{ outgoing at } V\mbox{ black;}\\
\mbox{int}(e) +\mbox{wind}(e,e_1), & \mbox{ if } U, V \mbox{ black and } e_1 \mbox{ outgoing at } V;\\
\end{array} \right.
$}
\end{equation}
\end{enumerate}
We also call geometric all signatures gauge equivalent to the one defined above. 
\end{definition}

We illustrate Definition~\ref{def:geo_sign} in Figure~\ref{fig:geo_sig}.

\begin{remark}
The definition of geometric signatures was motivated by the formulas expressing the vectors at internal edges as the sums over all paths connecting the given edge to the boundary sinks with signs defined in terms of winding and intersection numbers \cite{AG7}. 
\end{remark}

In \cite{AG4} we prove that changes of perfect orientations, of gauge ray direction or internal vertex position change (which can affect local windings and intersections) act on geometric signatures as gauge transformations. Therefore the total geometric signature on each face is independent on the choice of perfect orientation and gauge ray direction. In \cite{AG4}, we also show that its value just depends on the number of white vertices bounding the face. This result will be used in Section \ref{sec:comb} to prove the reality and regularity properties of KP divisors.

\begin{theorem}\textbf{Effect of elementary transformations on signatures and the total signature at faces}\label{theo:sign_face} \cite{AG4}
Let $({\mathcal G},{\mathcal O}(I))$ be a plabic graph in the disk representing a positroid cell $\S\subset \GTNN$, and let $\epsilon_{U,V}$ be its geometric signature. or a face  $\Omega$  of the graph $\mathcal G$ let $\epsilon(\Omega)$ be the total contribution of the geometric signature at the edges $e=(U,V)$ bounding the face $\Omega$:
\begin{equation}\label{eq:eps_tot}
\epsilon(\Omega) = \sum_{e\in\partial\Omega} \epsilon_{e}.
\end{equation}
Then $\epsilon(\Omega)\mod(2)$ is invariant with respect to changes of orientation, gauge ray direction and internal vertices position change, and
\begin{equation}\label{eq:sign_face}
\epsilon(\Omega) \; = \;  
\left\{ \begin{array}{ll}
n_{\mbox{\scriptsize{white}}}(\Omega) \; + \; 1 \quad \mod 2, & \quad \mbox{if } \Omega \mbox{ is a finite face}; \\
\\
n_{\mbox{\scriptsize{white}}}(\Omega) \; + \; k \quad \mod 2, & \quad \mbox{if } \Omega \mbox{ is the infinite face}.
\end{array}
\right.
\end{equation}
\end{theorem}

Finally, we characterize the solutions to the system of relations for the geometric signature.

\begin{theorem}\textbf{Lam system of relations for the geometric signature and Postnikov boundary measurement matrix}\label{thm:lam_rel} \cite{AG4,AG7}
Let $\mathcal G$ be a plabic graph representing a $|D|$--dimensional positroid cell $\S \subset \GTNN$, and let $\epsilon_{U,V}$ be a geometric signature in the equivalence class of $\mathcal G$. Let ${|\mathcal E}|$ be the number of edges of the graph. Then Lam system of relations has the following properties:
\begin{enumerate}
\item It has $2|{\mathcal E}|$ variables and   $2|{\mathcal E}|+k-n$ equations;  
\item If all weights $w_{e}$ are positive, then Lam system of relations has full rank;   
\item Given the matrix representing Lam system, then for any $I\in{\mathcal M}$ the maximal minor associated to the variables at the internal half-edges and at the sources $b_j$, $j\in I$, is different from zero;
\item\label{thm:lam_rel:i5} The half-edge solutions of Lam system are untrivial rational functions of the edge weights with subtraction-free denominators. These components are explicitly given in terms of the generalized Talaska flows (see \cite{AG7} for the explicit formulas). In particular, if the network possesses an acyclic orientation, these numerators are subtraction-free rational functions, and this property is preserved for any other perfect orientation and choice of the ray direction, of vertex gauge and of weight gauge, because of the gauge invariance of the signature;
\item Let the half-edge variables be vectors in $\R^n$, and  $I\in{\mathcal M}$ be fixed. If one assigns the $j$--th vector of the canonical basis $E_j$ to the variable $z_{b_j}$,
  \begin{equation}
  \label{eq:bc1}  
z_{b_j} =  E_j, \ \  j\in \bar I
\end{equation}
then at the boundary sources the solution of the system is
\begin{equation}
  \label{eq:bc2}  
z_{b_{i_r}, e_{i_r}} =E_{i_r}- A[r],
\end{equation}
where $i_r\in I$, and $A[r]$ is the $r$-th row of the reduced row echelon matrix with respect to the base $I$ representing the network $\mathcal N$. Therefore Lam relations for the geometric signature reproduce Postnikov measurement matrix.
\end{enumerate}
\end{theorem}

\begin{remark}\textbf{Zero numerators on reducible networks}\label{rem:zero_vectors}
Since the numerators in Item~\ref{thm:lam_rel:i5} are not necessary subtraction-free, on reducible plabic networks it may happen that for some positive weights they may vanish even if there exist directed paths starting at $e$ and ending at $b_j$ (see Section~\ref{sec:constr_null}).
\end{remark}

\section{KP multi-line solitons in the Sato Grassmannian and in finite-gap theory}\label{sec:soliton_theory}

Kadomtsev-Petviashvili-II (KP) equation \cite{KP}
\begin{equation}\label{eq:KP}
(-4u_t+6uu_x+u_{xxx})_x+3u_{yy}=0,
\end{equation}
is one of the most famous integrable equations, and it is a member of an integrable hierarchy (see \cite{D,DKN,H,MJD,S} for more details).  

The multiline soliton solutions are a special class of solutions  
realized starting from the soliton data $({\mathcal K}, [A])$, where 
${\mathcal K}$ is a set of real ordered phases $\kappa_1<\cdots<\kappa_n$, $[A]$ denotes
a point in the finite dimensional real Grassmannian $Gr (k,n)$ represented by a $k\times n$ real matrix  $A =( A^i_j )$ ($i\in [k], j\in [n]$), of maximal rank $k$.
Following \cite{Mat}, see also \cite{FN}, to such data we associate $k$ linearly independent solutions
$f^{(i)}(\vec t) = \sum_{j=1}^n A^i_j e^{\theta_j}$, $i\in [k]$, to the heat hierarchy
$\partial_{t_l} f = \partial_x^l f$, $l=2,3,\dots$.
Then
\begin{equation}\label{eq:KPsol}
u( {\vec t} ) = 2\partial_{x}^2 \log(\tau ( {\vec t}))
\end{equation}
is a multiline soliton solution to (\ref{eq:KP}) with
\[
\tau (\vec t) = Wr_{t_1} (f^{(1)},\dots, f^{(k)})= \sum\limits_{I} \Delta_I (A)\prod_{\mycom{i_1<i_2}{ i_1,i_2 \in I}} (\kappa_{i_2}-\kappa_{i_1} ) \, e^{ \sum\limits_{i\in I} \theta_i },
\]
where the sum is over all $k$--element ordered subsets $I$ in $[n]$, {\it i.e.} $I=\{ 1\le i_1<i_2 < \cdots < i_k \le n\}$ and $\Delta_I (A)$ are the maximal minors of the matrix $A$ with respect to the columns $I$, {\it i.e.} the Pl\"ucker coordinates for the corresponding point in the finite dimensional Grassmannian $Gr (k,n)$. Since linear recombinations of the rows of $A$ preserve the KP multisoliton solution $u({\vec t})$ in (\ref{eq:KPsol}), the soliton data is the corresponding point $[A]\in Gr(k,n)$.

$u( {\vec t} ) = 2\partial_{x}^2 \log(\tau ( {\vec t}))$ 
is a real regular multi--line soliton solution to the KP equation (\ref{eq:KP}) bounded for all real $x,y,t$ if and only if $\Delta_I (A) \ge 0$, for all $I$, that is if and only if $[A]\in Gr^{\mbox{\tiny TNN}} (k,n)$ \cite{KW2}. We remark that the weaker statement that the solution of the KP hierarchy is bounded for all real times if and only if all Pl\"ucker coordinates are non-negative was earlier proven in \cite{Mal}.

Any given soliton solution is associated to an infinite set of soliton data $({\mathcal K}, [A])$. However there exists a unique \textbf{minimal} pair $(k,n)$ such that the soliton solution can be realized with $n$ phases $\kappa_1<\cdots<\kappa_n$, $[A]\in \GTNN$ but not with $n-1$ phases and $[A^\prime]\in Gr^{\mbox{\tiny TNN}} (k^{\prime},n^{\prime})$ and either $(k^{\prime}, n^{\prime}) =(k, n-1)$ or $(k^{\prime}, n^{\prime}) =(k-1, n-1)$.
In the following, to avoid excessive technicalities we consider only regular and irreducible soliton data.

\begin{definition}\label{def:regsol}{\bf Regular and irreducible soliton data} \cite{CK}
$({\mathcal K}, [A])$ are regular soliton data if ${\mathcal K} = \{ \kappa_1 < \cdots < \kappa_n \}$ and $[A]\in \GTNN$, that is if the KP soliton solution as in (\ref{eq:KPsol}) is real regular and bounded for all $(x,y,t)\in \mathbb{R}^3$.

Moreover  the regular soliton data $({\mathcal K}, [A])$ are irreducible if $[A]$ is a point in the irreducible part of the real Grassmannian, {\sl i.e.} if the reduced row echelon matrix $A$ has the following properties:
\begin{enumerate}
\item\label{it:col} Each column of $A$ contains at least a non--zero element;
\item\label{it:row} Each row of $A$ contains at least one nonzero element in addition to the pivot.
\end{enumerate}
\end{definition}
The class of solutions associated to irreducible regular soliton data has remarkable asymptotic properties both in the $(x,y)$ plane at fixed time $t$ and in the tropical limit ($t\to \pm \infty)$, which have been related to the combinatorial classification of the irreducible part $\GTNN$ for generic choice of the phases ${\mathcal K}$ in a series of papers (see \cite{BPPP,CK,DMH,KW1,KW2} and references therein).

According to Sato theory \cite{S}, the wave function associated to regular soliton data $({\mathcal K},[A])$, can be obtained from the dressing (inverse gauge) transformation of the vacuum (zero--potential) wave function $\displaystyle \phi^{(0)} (\zeta, \vec t) =\exp ( \theta(\zeta, {\vec t}))$, which solves
$\partial_x \phi^{(0)} (\zeta, \vec t)=\zeta \phi^{(0)} (\zeta, \vec t)$, 
$\partial_{t_l}\phi^{(0)} (\zeta, \vec t) = \zeta^l \phi^{(0)} (\zeta, \vec t)$, $l\ge 2$,
via the dressing ({\it i.e.} gauge) operator $W = 1 -{\mathfrak w}_1({\vec t})\partial_x^{-1} -\cdots - {\mathfrak w}_k({\vec t})\partial_x^{-k}$,
where ${\mathfrak w}_1({\vec t}),\dots,{\mathfrak w}_k({\vec t})$ are the
solutions to the following linear system of equations
$\partial_x^k f^{(i)} = {\mathfrak w}_1 \partial_x^{k-1} f^{(i)}+\cdots + {\mathfrak w}_k f^{(i)}$, $i\in [k]$. Therefore we may define the Darboux operator
\begin{equation}\label{eq:D}
{\mathfrak D}^{(k)} = W \partial_x^{k} =  \partial_x^k - {\mathfrak w}_1 (\vec t)\partial_x^{k-1} -\cdots - {\mathfrak w}_k(\vec t),
\end{equation}
such that 
\begin{equation}\label{eq:Df}
{\mathfrak D}^{(k)}  f^{(i)} \equiv 0, \ \ i\in[k].
 \end{equation} 
Then
\[
L= W \partial_x W^{-1} = \partial_x + \frac{u(\vec t)}{2}\partial_x^{-1} +\cdots,  \qquad
u(\vec t) = 2\partial_x {\mathfrak w}_1 (\vec t),
\]
\begin{equation}\label{eq:psi0}
  \psi^{(0)} (\zeta; \vec t) = W\phi^{(0)} (\zeta; \vec t) =\frac{1}{\zeta^k}\frac{ Wr_{t_1} (f^{(1)},\dots, f^{(k)}, \phi^{(0)})}{ Wr_{t_1} (f^{(1)},\dots, f^{(k)})},
\end{equation}
respectively are the KP-Lax operator, the KP--potential (KP solution) and the KP wave function, {\sl i.e.}
\begin{equation}\label{eq:dress_hier}
L \psi^{(0)} (\zeta; \vec t) =\zeta \psi^{(0)} (\zeta; \vec t), 	\quad\quad
\partial_{t_l}\psi^{(0)} (\zeta; \vec t)= B_l \psi^{(0)} (\zeta; \vec t), \ \ l\ge 2,
\end{equation}
where $B_l = (W \partial_x^l W^{-1} )_+ =(L^l)_+ $ (here and in the following the symbol $(\cdot )_+$ denotes the differential part of the operator).

In the following we also use a different normalization of the Sato wave function.
\begin{remark}\label{def:unnorm}{\bf Not-normalized Sato wave function}
The KP wave function associated to this class of solutions may be equivalently expressed as
\begin{equation}\label{eq:Satowf} 
{\mathfrak D}^{(k)}\phi^{(0)} (\zeta; \vec t)  =  \zeta^k \psi^{(0)} (\zeta; \vec t),
\end{equation}
We call ${\mathfrak D}^{(k)}\phi^{(0)} (\zeta; \vec t)$ \textbf{not-normalized Sato wave function.}
\end{remark}

\begin{definition}\label{def:Satodiv}{\bf Sato divisor coordinates}

Let the regular soliton data be $({\mathcal K}, [A])$, ${\mathcal K} = \{ \kappa_1 < \cdots < \kappa_n \}$, $[A]\in \GTNN$. We call \textbf{Sato divisor coordinates at time} $\vec t$, the set of roots $\zeta_j (\vec t)$, $j\in [k]$, of the characteristic equation associated to the Dressing transformation
\begin{equation}\label{eq:Dressing_roots}
\zeta_j^k(\vec t) - {\mathfrak w}_1 (\vec t)\zeta_j^{k-1}(\vec t)-\cdots  - {\mathfrak w}_{k-1} (\vec t)\zeta_j(\vec t)- {\mathfrak w}_k(\vec t) = 0. 
\end{equation}
\end{definition}

In \cite{Mal} it is proven the following proposition
\begin{proposition}\textbf{Reality and simplicity of the KP soliton divisor}\label{prop:malanyuk} \cite{Mal}.
Let the regular soliton data be $({\mathcal K}, [A])$, ${\mathcal K} = \{ \kappa_1 < \cdots < \kappa_n \}$, $[A]\in \GTNN$. Then for all $\vec t$, $\zeta_j^k(\vec t)$ are real and satisfy $\zeta_j(\vec t)\in [\kappa_1,\kappa_n]$, $j\in [k]$. 
Moreover for almost every $\vec t$ the roots of (\ref{eq:Dressing_roots}) are simple.
\end{proposition}

The following definition is then fully justified.

\begin{definition}\label{def:Sato_data}\textbf{Sato algebraic--geometric data} Let $({\mathcal K}, [A])$ be given regular soliton data with $[A]$ belonging to a $|D|$ dimensional positroid cell in $\GTNN$.  Let $\Gamma_0$ be a copy of $\mathbb{CP}^1$ with marked points $P_0$, local coordinate $\zeta$ such that $\zeta^{-1} (P_0)=0$ and $\zeta(\kappa_1)<\zeta(\kappa_2)<\cdots<\zeta(\kappa_n)$. Let $\vec t_0$ be real and such that the real roots $\zeta_j (\vec t_0)$ in (\ref{eq:Dressing_roots}) are simple.

Then to the data $({\mathcal K}, [A], \Gamma_0\backslash \{ P_0	\} ,\vec t_0)$ we associate the \textbf{Sato divisor} $\DS=\DS(\vec t_0)$
\begin{equation}\label{eq:Satodiv}
\DS=\{ \gamma_{S,j} \in \Gamma_0 \, : \, \zeta(\gamma_{S,j})= \zeta_j (\vec t_0), \quad j\in [k]\}.
\end{equation}
\end{definition}

\begin{definition}\label{def:Sato_function}\textbf{Normalized Sato wave function} With the same nonations as in Definition~\ref{def:Sato_data} we define the normalized Sato wave function for such data by:
\begin{equation}\label{eq:SatoDN}
{\hat \psi } (P, \vec t) = \frac{{\mathfrak D}\phi^{(0)} (P; \vec t)}{{\mathfrak D}\phi^{(0)} (P; \vec t_0)} = \frac{\psi^{(0)} (P; \vec t)}{\psi^{(0)} (P; \vec t_0)}, \quad\quad \forall P\in \Gamma_0\backslash \{ P_0\},
\end{equation}
with ${\mathfrak D}\phi^{(0)} (\zeta; \vec t)$ as in (\ref{eq:Satowf}). 
By definition $({\hat \psi}_0 (P,\vec t)) + \DS \ge 0$, for all $\vec t$.
\end{definition}

In the following, we use the same symbol for the points in $\Gamma_0$ and their local coordinates to simplify notations. In particular, we use the symbol $\gamma_{S,j}$ both for the Sato divisor points and Sato divisor coordinates.

\begin{remark}\label{rem:fundam}\textbf{Incompleteness of Sato algebraic--geometric data} 
Let $1\le k<n$ and let $\vec t_0$ be fixed. Given the phases $\kappa_1<\cdots <\kappa_n$ and the spectral data $( \Gamma_0\backslash \{ P_0	\} , \DS) $, where $\DS=\DS(\vec t_0) $ is a $k$ point divisor satisfying Proposition \ref{prop:malanyuk}, it is, in general, impossible to identify uniquely the point $[A]\in \GTNN$ corresponding to such spectral data. Indeed, assume that $[A]$ belongs to an irreducible positroid cell of dimension $|D|$. Then the degree of $\DS$ equals $k$, but $\max\{k, n-k\} \le |D| \le k(n-k)$.
\end{remark}

Our purpose is a completion of the Sato algebraic--geometric data based on singular finite--gap theory on reducible algebraic curves \cite{Kr3, AG1,AG3}, where
\begin{enumerate}
\item  $\Gamma_0$ is one of the connected components of a reducible rational spectral curve $\Gamma$;
\item The restriction of the full divisor on $\Gamma$ to  $\Gamma_0$ coincides with the Sato divisor;  
\item On every component except $\Gamma_0$ the wave function is rational.
\end{enumerate}  

Indeed, soliton KP solutions can be obtained from regular finite--gap solutions of (\ref{eq:KP}) by proper degenerations of the spectral curve \cite{Kr0,Kr2,DKN}.  The spectral data for KP finite--gap solutions are introduced and described  in \cite{Kr1,Kr2}: they are a finite genus $g$ compact Riemann surface $\Gamma$ with a marked point $P_0$, a local parameter $1/\zeta$ near $P_0$ and a non-special divisor $\mathcal D=\gamma_1+\ldots+\gamma_g$ of degree $g$ in $\Gamma$.

The Baker-Akhiezer function $\hat\psi (P, \vec t)$, $P\in\Gamma$, is defined by the following analytic properties:
\begin{enumerate}
\item For any fixed $\vec t$ the function $\hat\psi (P, \vec t)$ is meromorphic in $P$ on $\Gamma\backslash P_0$;
\item On  $\Gamma\backslash P_0$ the function $\hat\psi (P, \vec t)$ is regular outside the divisor points $\gamma_j$ and has at most first order poles 
at the divisor points. Equivalently, if we consider the line bundle $\mathcal L(\mathcal D)$  associated to $\mathcal D$, then
for each fixed $\vec t$ the function $\hat\psi (P, \vec t)$ is a holomorphic section of $\mathcal L(\mathcal D)$ outside $P_0$.
\item $\hat\psi (P, \vec t)$ has an essential singularity at the point $P_0$ with the following asymptotics:
\[
{\hat \psi} (\zeta, \vec t) = e^{ \zeta x +\zeta^2 y +\zeta^3 t +\cdots} \left( 1 - \chi_1({\vec t})\zeta^{-1} - \cdots
-\chi_k({\vec t})\zeta^{-k}  - \cdots\right). 
\]
\end{enumerate}
For generic data these properties define a unique function, which is a common eigenfunction to all KP hierarchy auxiliary linear operators 
$-\partial_{t_j} + B_j$, where $B_j =(L^j)_+$, and the Lax operator $L=\partial_x+\frac{u(\vec t)}{2}\partial_x^{-1}+ u_2(\vec t)\partial_x^{-2}+\ldots.$
Therefore all these operators commute and the potential $u(\vec t)$ satisfies the KP hierarchy. In particular, the KP equation arises in the Zakharov-Shabat-Druma commutation representation \cite{ZS}, \cite{Druma} as the compatibility for the second and the third operator:
$[ -\partial_y + B_2, -\partial_t +B_3] =0$, with $B_2 \equiv (L^2)_+ = \partial_x^2 + u$, $B_3 = (L^3)_+ = \partial_x^3 +\frac{3}{4} (u\partial_x +\partial_x u) + \tilde u$
and $\partial_x\tilde u =\frac{3}{4} \partial_y u$.
The Its-Matveev formula represents the KP hierarchy solution $u(\vec t)$ in terms of the Riemann theta-functions associated with $\Gamma$ (see, for example, \cite{Dub}). 

In \cite{DN} there were established the necessary and sufficient conditions on spectral data to generate real regular KP hierarchy solutions for all real $\vec t$, under the assumption that $\Gamma$ is smooth and has genus $g$: 
\begin{enumerate}
\item $\Gamma$ possesses an antiholomorphic involution ${\sigma}:\Gamma\rightarrow\Gamma$, ${\sigma}^2=\mbox{id}$, which has the maximal possible number of fixed components (real ovals), $g+1$, therefore $(\Gamma,\sigma)$ is an $\mathtt M$-curve \cite{Har,Nat,Vi}. This involution is automatically ``separating'', i.e. the set of real ovals divides $\Gamma$ into two connected components. Each of these components is homeomorphic to a sphere with $g+1$ holes;  
\item $P_0$ lies in one of the ovals, and each other oval contains exactly one divisor point. The oval containing $P_0$ is called ``infinite'' and all 
  other ovals are called ``finite''.
\end{enumerate}

For double periodic in $x$, $y$ potentials $u(x,y)$ the direct scattering transform was developed in \cite{Kr4}, where it was shown that for real regular potentials the spectral curve is always an $\mathtt M$-curve (generically of infinite genus) with correctly located divisor points. 

The sufficient condition of the Theorem in \cite{DN} still holds true if the spectral curve $\Gamma$ degenerates in such a way that the divisor remains in the finite ovals at a finite distance from the essential singularity \cite{DN}. Of course, this condition is not necessary for degenerate curves. Moreover, the algebraic-geometric data for a given soliton data $({\mathcal K},[A])$ are not unique since we can construct infinitely many reducible curves generating the same soliton solutions. 

As it was pointed by S.Novikov, it is natural to construct real regular  multiline solutions by degenerating real regular finite-gap solutions. As it was shown in \cite{AG1,AG3}, any real regular multiline soliton solution can be obtained by such degeneration. The construction of \cite{AG3} is based on the representation of totally non--negative Grassmannians via directed planar networks \cite{Pos}. In the following Sections, we complete the project started in \cite{AG3} constructing a map from the networks in Postnikov equivalence class to real and regular spectral data for KP--II soliton solutions. 

\section{Algebraic-geometric approach for irreducible KP soliton data in $\GTNN$}\label{sec:3}

In this Section we state the main results of our paper. Following \cite{AG3}, we define the desired properties of Baker-Akhiezer functions on reducible curves associated with a given soliton data. 

\begin{definition}
\label{def:rrss}
\textbf{Real regular algebraic-geometric data associated with a given soliton solution.}
Let the soliton data $({\mathcal K},[A])$ be fixed, where ${\mathcal K}$ is a 
collection of real phases $\kappa_1<\kappa_2<\ldots <\kappa_n$, $[A]\in \GTNN$. Let $|D|$ be the dimension of the irreducible positroid cell to which $[A]$ belongs. Let $(\Gamma_0, P_0, \DS)$ be the Sato algebraic--geometric data for $({\mathcal K},[A])$ as in Definition \ref{def:Sato_data} for a given $\vec t_0$.
Let the spectral curve $\Gamma$ be a reducible connected curve with a marked point $P_0$, a local parameter $1/\zeta$ near $P_0$ such that 
\begin{enumerate}
\item $\Gamma_0$ is the irreducible component of $\Gamma$ containing $P_0$;
\item $\Gamma$  may be obtained from a rational degeneration of a smooth ${\mathtt M}$-curve of genus $g$, with $g\ge |D|$ and the antiholomorphic involution preserves the maximum number of the ovals in the limit, so that $\Gamma$ possesses $g+1$ real ovals. 
\end{enumerate}

Assume that ${\mathcal D}$ is a degree $g$ non-special divisor on $\Gamma\backslash P_0$, and that $\hat\psi$ is the normalized Baker-Ahkiezer function associated to such data, i.e. for any $\vec t$ its pole divisor is contained in ${\mathcal D}$: $({\hat \psi} (P , \vec t))+{\mathcal D} \ge 0$ on $\Gamma\backslash P_0$, where $(f)$ denotes the divisor of $f$.

We say that  \textbf{the algebraic-geometric data $(\Gamma, P_0,{\mathcal D})$ are associated to the soliton data $({\mathcal K},[A])$,} if
\begin{enumerate}
\item The restriction of ${\mathcal D}$ to $\Gamma_0$ coincides with the Sato divisor $\DS$;
\item The restriction of $\hat\psi$ to $\Gamma_0$ coincides with the Sato normalized dressed wave function for the soliton data $({\mathcal K},[A])$. 
\end{enumerate}

We say that the \textbf{divisor ${\mathcal D}$ satisfies the reality and regularity conditions} if, moreover, $P_0$ belongs to one of the fixed ovals and each other oval contains exactly one divisor point. 
\end{definition}

From now on we fix the regular irreducible soliton data $({\mathcal K}, [A])$, and for a given plabic network representing $[A]$ we present a \textbf{direct} construction of the algebraic geometric data associated to points in irreducible positroid cells of $\GTNN$. $\Gamma_0$ is the rational curve associated to Sato dressing and is equipped with a finite number of marked points: the ordered real phases ${\mathcal K} = \{ \kappa_1<\cdots<\kappa_n\}$, the essential singularity 
$P_0$ of the wave function and the Sato divisor $\DS$ as in Definition \ref{def:Sato_data}. The normalized wave function $\hat \psi$ on $\Gamma_0 \backslash \{ P_0\}$ is the normalized Sato wave function (\ref{eq:SatoDN}).
In the present paper, we do the following:

\textbf{Main construction} {\sl Assume we are given a real regular bounded multiline KP soliton solution generated by the following soliton data:
\begin{enumerate}
\item A set of $n$ real ordered phases ${\mathcal K} =\{ \kappa_1<\kappa_2<\dots<\kappa_n\}$;
\item A point $[A]\in \S \subset \GTNN$, where $\S$ is an irreducible positroid cell of dimension $|D|$. 
\end{enumerate}
Let ${\mathcal N}$ be a plabic network in the disk in Postnikov equivalence class representing $[A]$ and let ${\mathcal G}$ be the graph of ${\mathcal N}$. Using Lam relations we extend the non-normalized Sato wave function to all internal half-edges. Two cases may occur:
\begin{enumerate}
\item There exists a time $\vec t_0$ such that at all half-edges the non-normalized wave function is different from $0$. This situation is generic. The last statement follows from the explicit formulas for half-edge vectors proven in \cite{AG4}, providing a rational expression for these vectors with non-trivial numerators and subtraction-free denominators.
\item For special choices of positive edge weights there exists an edge such that the extended wave function is identically zero for all $\vec t$ at the corresponding double points.
\end{enumerate}

If the network is reduced, the second case cannot occur, see \cite{AG3}, \cite{A3}. In Section~\ref{sec:constr_null} we briefly sketch the construction of divisor in the second case, and we plan to study it thoroughly in future.

If the first case occurs, we associate the following algebraic-geometric objects to each set $({\mathcal K}, [A]; {\mathcal N},\vec t_0)$:
\begin{enumerate}
\item A reducible curve $\Gamma=\Gamma(\mathcal G)$ which is a rational degeneration of a smooth $\mathtt M$--curve of genus $g\ge |D|$, where $g+1$ is the number of faces of $\mathcal G$. In our approach, the curve $\Gamma_0$ is the irreducible component of $\Gamma$ corresponding to the boundary of the disk. The marked point $P_0$ belongs to the intersection of $\Gamma_0$ with the infinite oval, associated to the infinite face of the graph; 
\item A unique real and regular degree $g$ non--special KP divisor $\DKP({\mathcal K}, [A])$ such that any finite oval (associated to a finite face of the graph) contains exactly one divisor point and $\DKP ({\mathcal K}, [A])\cap \Gamma_0$ coincides with the Sato divisor at time $\vec t_0$;
\item A unique KP wave--function $\hat \psi$ as in Definition \ref{def:rrss} such that
\begin{enumerate}
\item Its restriction to $\Gamma_0\backslash \{P_0\}$ coincides with the normalized Sato wave function (\ref{eq:SatoDN});
\item At the double points of the curve it coincides with the normalization of the solution to Lam system of relations; 
\item Its pole divisor has degree $\mathfrak d \le g$ and is contained in $\DKP ({\mathcal K}, [A])$.
\end{enumerate}
\end{enumerate}
}

In particular, if $\mathcal G= \mathcal G_T$ is the trivalent bicolored Le--graph \cite{Pos}, then $\Gamma(\mathcal G_T)$ is a rational degeneration of on $\mathtt M$--curve of minimal genus $|D|$, it has exactly $|D|+1$ ovals, and ${\mathfrak d}=g=|D|$ \cite{AG3}.

\subsection{Main construction -- Part I. The reducible rational curve $\Gamma=\Gamma(\mathcal G)$}
\label{sec:gamma}

The construction of $\Gamma(\mathcal G)$ is a straightforward modification of a special case in the classical construction of nodal curves by dual graphs \cite{ACG}.

\begin{figure}
  \centering
  {\includegraphics[height=0.23\textwidth]{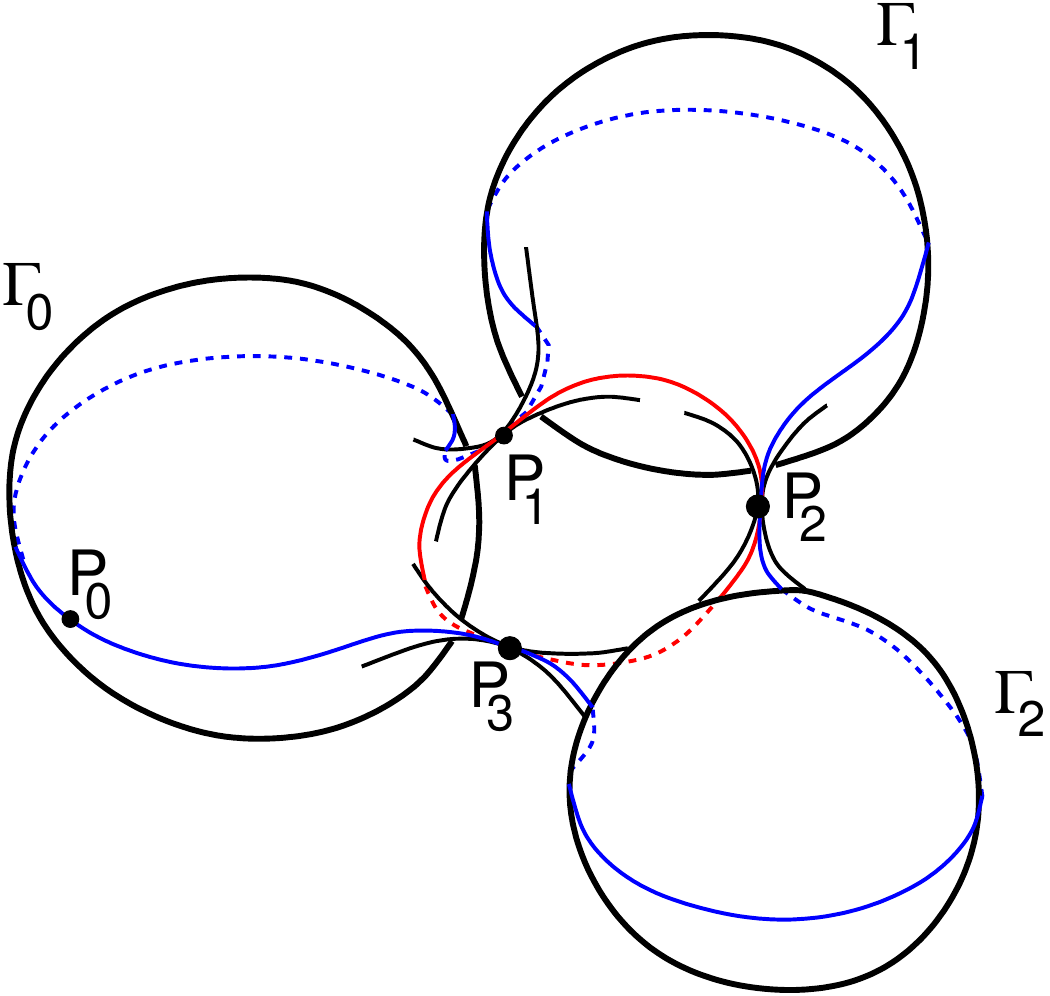}}\hspace{1cm}
  {\includegraphics[height=0.23\textwidth]{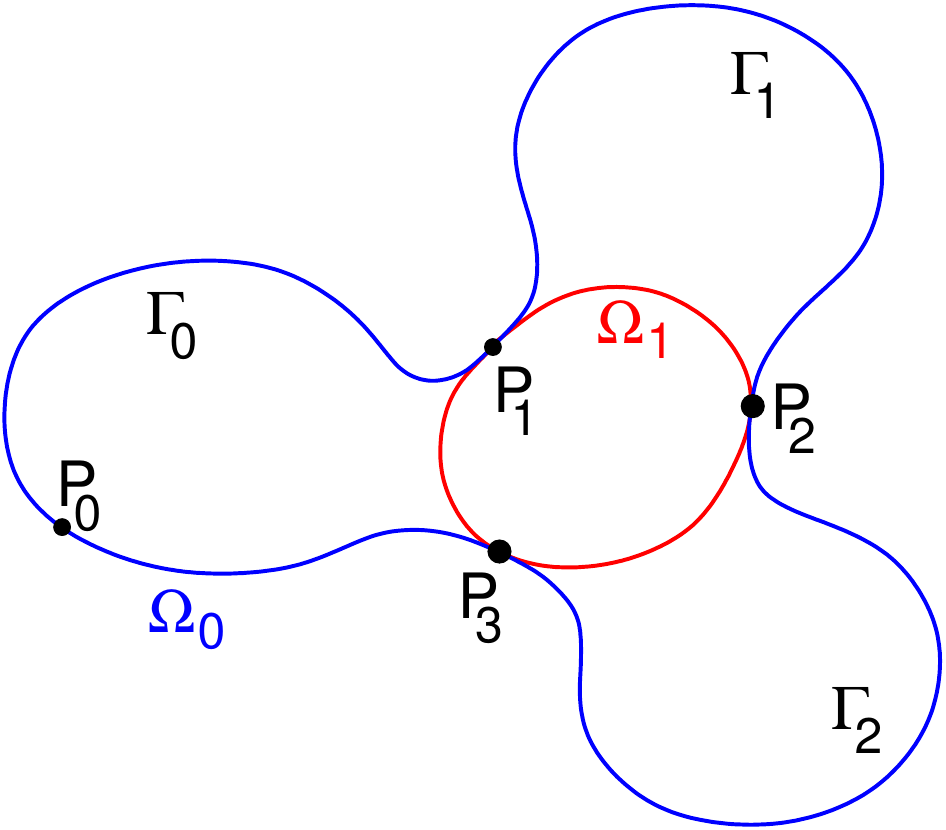}}\hspace{1cm}
  {\includegraphics[height=0.23\textwidth]{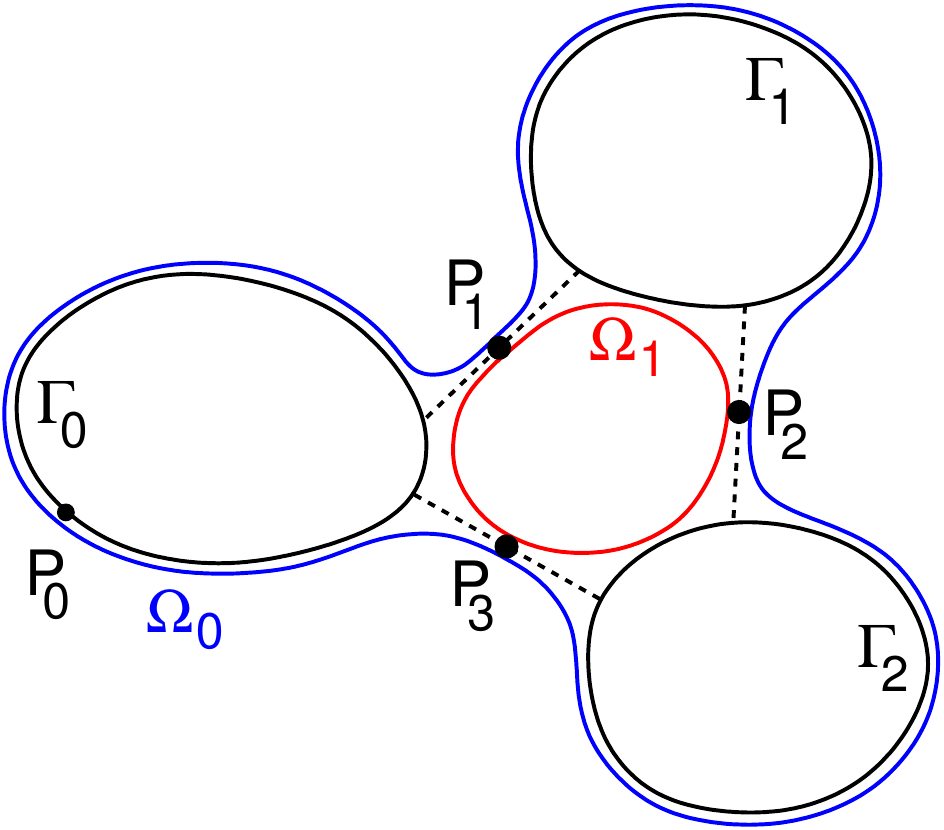}}
  \caption{\label{fig:curve_model}\small{\sl The model of reducible rational curve. On the left we have three Riemann spheres $\Gamma_0$, $\Gamma_1$ and $\Gamma_2$ glued at the points $P_1$, $P_2$, $P_3$. The real part of the curve is drawn in blue and red. In the middle we draw just the real part of the curve and evidence the two real ovals $\Omega_0$ (blue) and $\Omega_1$ (red). On the right we provide the representation of the real topological model of the curve used throughout the paper: The real parts of each rational component is a circle and the double points are represented by dashed segments.}}
\end{figure}

\begin{remark}\label{rem:labedges}\textbf{Labeling of edges at vertices}
Let ${\mathcal G}$ be a plabic graph as in Definition \ref{def:graph}. We number the edges at an internal vertex $V$ anticlockwise in increasing order with the following rule: the unique edge starting at a black vertex is labeled 1 and the unique edge ending at a white vertex is labeled 3 (see also Figure \ref{fig:markedpoints}). 
\end{remark}

We construct the curve $\Gamma=\Gamma({\mathcal G})$ gluing a finite number of copies of $\mathbb{CP}^1$, each corresponding to an internal vertex in ${\mathcal G}$, and one copy of $\mathbb{CP}^1=\Gamma_0$, corresponding to the boundary of the disk, at pairs of points corresponding to the edges of ${\mathcal G}$. 
On each component, we fix a local affine coordinate $\zeta$ (see Definition~\ref{def:loccoor}) so that  the coordinates at each pair of glued points are real. The points with real $\zeta$ form the real part of the given component. We represent the real part of $\Gamma$ as the union of the ovals (circles) corresponding to the faces of ${\mathcal G}$.
For the case in which $\mathcal G$ is the Le--network see \cite{AG3}.

We use the following representation for real rational curves (see Figure~\ref{fig:curve_model} and \cite{AG1, AG3}): we only draw the real part of the curve, i.e. we replace each  $\Gamma_j=\mathbb{CP}^1$ by a circle. Then we schematically represent the real part of the curve by drawing these circles separately and connecting the glued points by dashed segments. The planarity of the graph implies that $\Gamma$ is a reducible rational $\mathtt M$--curve. 

\begin{figure}
  \centering
  \includegraphics[width=0.7\textwidth]{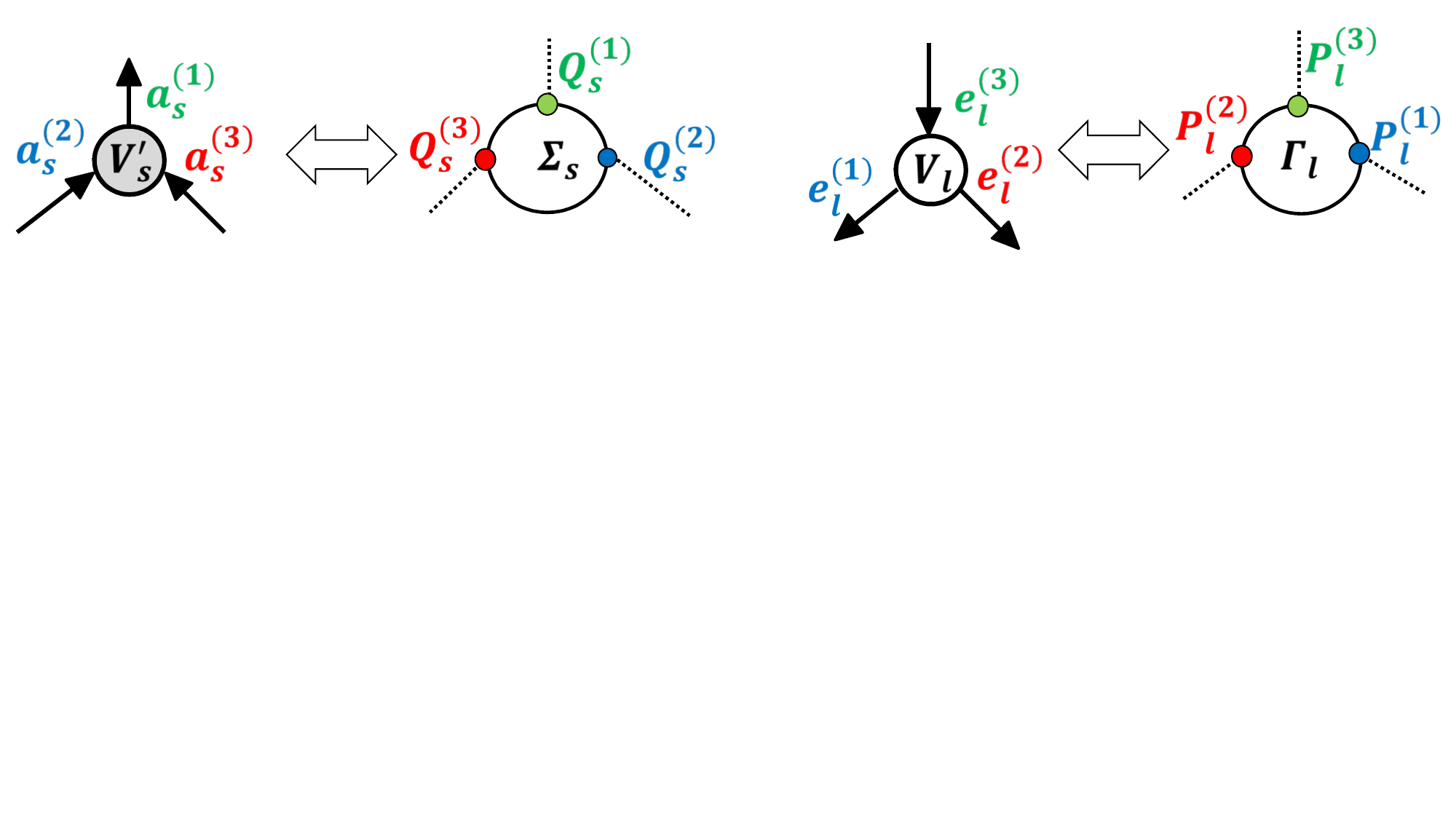}
  \vspace{-4 truecm}
  \caption{\small{\sl The rule for the marked points on the copies $\Sigma_{j}$ and $\Gamma_{l}$ corresponding to the edges of trivalent black and white vertices.}}
	\label{fig:markedpoints}
\end{figure}

\begin{main}\label{def:gamma}\textbf{The curve $\Gamma(\mathcal G)$.}
Let ${\mathcal K} =\{\kappa_1 < \cdots < \kappa_n\}$ and let ${\mathcal S}^{\mbox {\tiny TNN}}_{{\mathcal M}}\subset \GTNN$ be a fixed irreducible positroid cell of dimension $|D|$. Let ${\mathcal G}$ be a plabic graph representing $\S$ with $g+1$ faces, $g\ge|D|$. Then the curve $\Gamma = \Gamma ({\mathcal G})$ is associated to ${\mathcal G}$ using the correspondence in Table \ref{table:LeG}, after reflecting the graph w.r.t. a line orthogonal to the one containing the boundary vertices (we reflect the graph to have the natural increasing order of the marked points $\kappa_j$ on $\Gamma_0\subset\Gamma(\mathcal G)$).
\begin{table}
\caption{The graph ${\mathcal G}$  vs the reducible rational curve $\Gamma$} 
\centering
\begin{tabular}{|c|c|}
\hline\hline
$\mathcal G$ & $\Gamma$ \\[0.5ex]
\hline
Boundary of disk & Copy of $\mathbb{CP}^1$ denoted $\Gamma_0$ \\
Boundary vertex $b_l$             & Marked point $\kappa_l$ on  $\Gamma_0$\\
Black vertex   $V^{\prime}_{s}$   & Copy of $\mathbb{CP}^1$ denoted $\Sigma_{s}$\\
White vertex   $V_{l}$            & Copy of $\mathbb{CP}^1$ denoted $\Gamma_{l}$\\
Internal Edge                     & Double point\\
Face                              & Oval\\
Infinite face                     & Infinite oval $\Omega_0$ \\ [1ex]
\hline
\end{tabular}
\label{table:LeG}
\end{table}
More precisely:
\begin{enumerate}
\item We denote $\Gamma_0$ the copy of $\mathbb{CP}^1$ corresponding to the boundary of the disk and mark on it the points $\kappa_1<\cdots <\kappa_n$ corresponding to the boundary vertices $b_1,\dots, b_n$ on $\mathcal G$. We assume that $P_0=\infty$; 
\item A copy of $\mathbb{CP}^1$ corresponds to any internal vertex of $\mathcal G$. We use the symbol $\Gamma_l$ (respectively $\Sigma_s$) for the copy of $\mathbb{CP}^1$ corresponding to the white vertex $V_l$ (respectively the black vertex $V^{\prime}_s$);
\item On each copy of $\mathbb{CP}^1$ corresponding to an internal vertex $V$, we mark as many points as edges at $V$. In Remark~\ref{rem:labedges} we number the edges at $V$ anticlockwise in increasing order, so that, on the corresponding copy of $\mathbb{CP}^1$, the marked points are numbered clockwise because of the mirror rule (see Figure \ref{fig:markedpoints});
\item Gluing rules between copies of $\mathbb{CP}^1$ are ruled by edges: we glue copies of $\mathbb{CP}^1$ in pairs at the marked points corresponding to the end points of the edge;
\item The faces of $\mathcal G$ correspond to the ovals of $\Gamma$.  
\end{enumerate}
\end{main}

In Figure \ref{fig:net_curve} we present an example of curve corresponding to a network representing an irreducible positroid cell
in $Gr^{\mbox{\tiny TNN}} (4,9)$. Other examples are studied  in Sections \ref{sec:example} and \ref{sec:ex_Gr24top}.

\begin{figure}
  \centering
  {\includegraphics[width=0.47\textwidth]{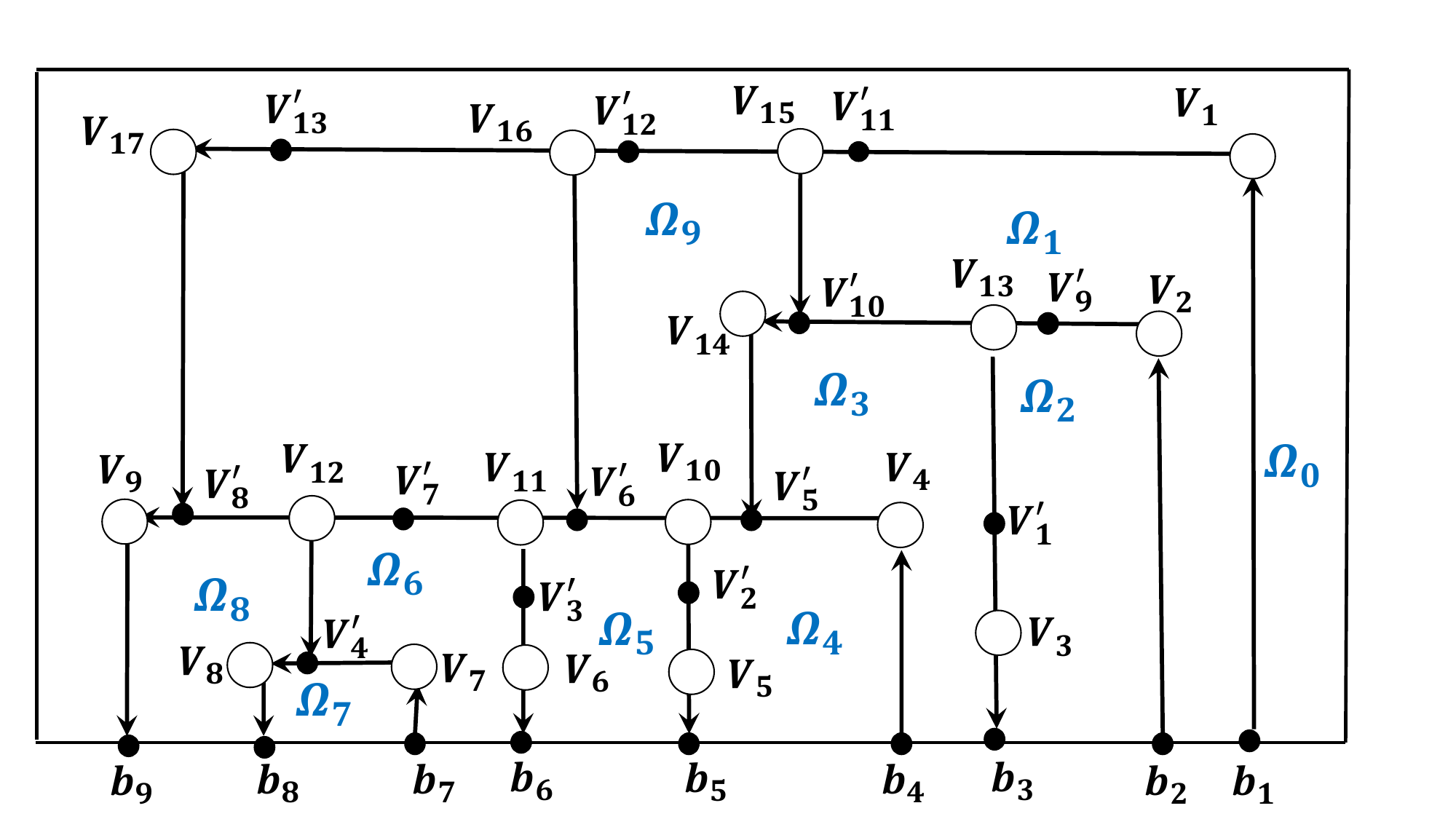}}
  \hfill
	{\includegraphics[width=0.47\textwidth]{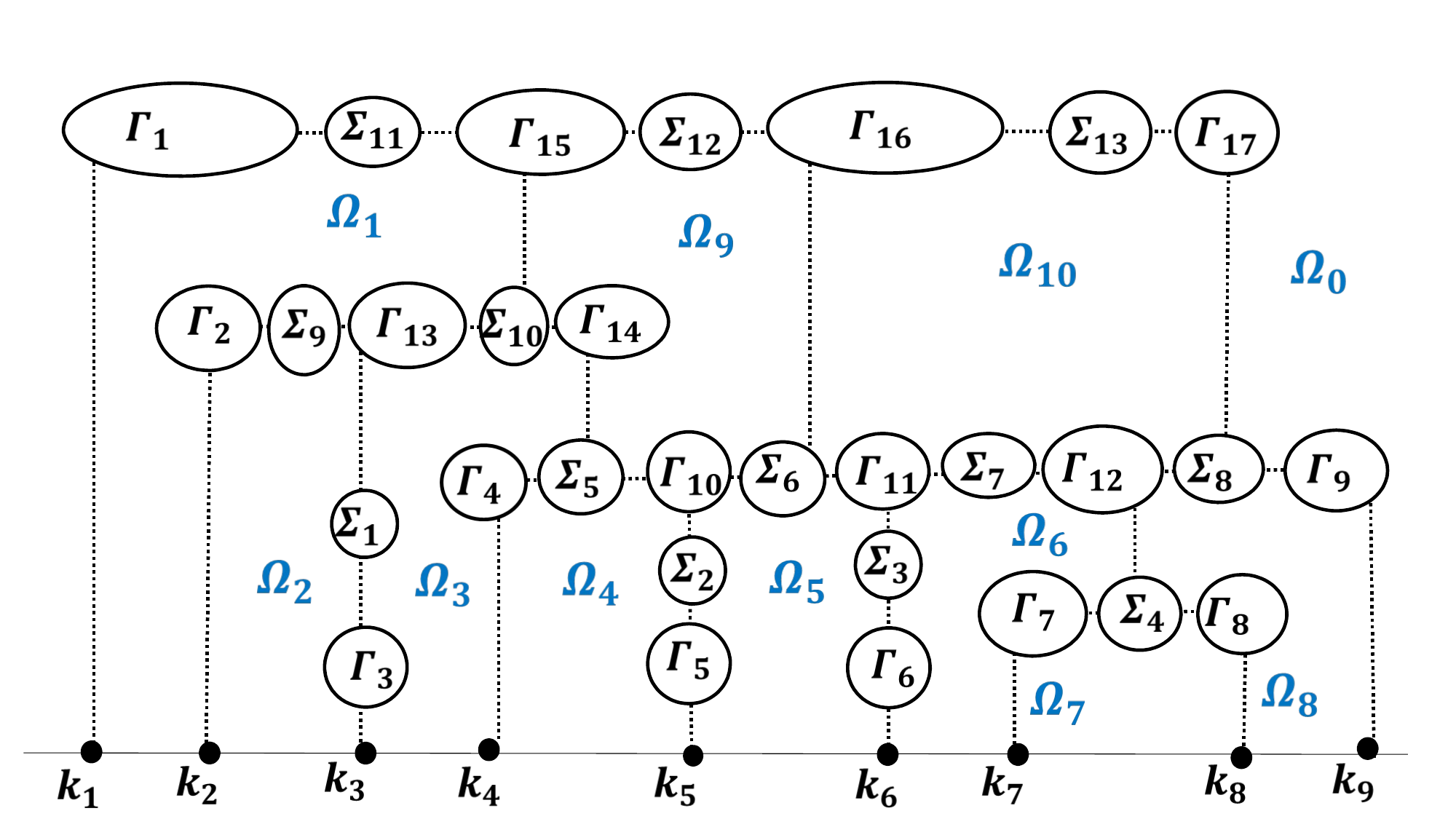}}
  \caption{\small{\sl The correspondence between the graph $\mathcal G$ and the curve $\Gamma (\mathcal G)$ for an irreducible positroid cell in $Gr^{\mbox{\tiny TNN}} (4,9)$. Components $\Gamma_l$ (respectively $\Sigma_j$) correspond to white vertices $V_l$ (respectively black vertices $V^{\prime}_j$).}}
	\label{fig:net_curve}
\end{figure}

\begin{remark} \textbf{Universality of the reducible rational curve $\Gamma(\mathcal G)$.}
If $\mathcal G$ is a trivalent graph representing $\mathcal S$, the construction of $\Gamma=\Gamma(\mathcal G)$ does \textbf{not} require the introduction of parameters. Therefore it provides a \textbf{universal} curve 
$\Gamma=\Gamma({\mathcal S};\mathcal G)$ for the whole positroid cell ${\mathcal S}$. On the contrary, if the graph has vertices of valency greater than 3, the moduli space of corresponding curves has dimension greater than zero. In Section \ref{sec:anycurve}, we show that the positroid cell ${\mathcal S}$ is (locally) parameterized by the divisor positions at the finite ovals. To construct global parametrization, we require a proper of resolution of singularities in the variety of divisors. In the simplest nontrivial case of $Gr^{TP}(1,3)$ this resolution of singularities in discussed in Section~\ref{sec:global}. We plan to study this question in the general case in a forthcoming paper.

The number of copies of $\mathbb{CP}^1$ used to construct $\Gamma(\mathcal G)$ is \textbf{excessive} in the sense that the number of ovals and the KP divisor is invariant if we eliminate 
all copies of $\mathbb{CP}^1$ corresponding to bivalent vertices (see Section \ref{sec:moves_reduc} and \cite{AG3}). 
\end{remark}

The curve $\Gamma(\mathcal G)$ is a partial normalization \cite{ACG} of a connected reducible nodal plane curve with $g+1$ ovals 
and is a rational degeneration of a genus $g$ smooth $\mathtt M$--curve. 

\begin{proposition}\label{prop:rational_curve}\textbf{$\Gamma(\mathcal G)$  is the rational degeneration of a smooth $\mathtt M$-curve of genus $g$.}
Let ${\mathcal K} = \{ \kappa_1 < \cdots < \kappa_n\}$ and ${\mathcal S}^{\mbox {\tiny TNN}}_{{\mathcal M}}$ be an irreducible positroid cell in $\GTNN$ corresponding to the matroid ${\mathcal M}$. Let $\Gamma=\Gamma(\mathcal G)$ be as in Construction \ref{def:gamma}. Then
\begin{enumerate}
\item $\Gamma$ possesses $g+1$ ovals which we label $\Omega_s$, $s\in [0,g]$; 
\item $\Gamma$ is the rational degeneration of a regular $\mathtt M$--curve of genus $g$.
\end{enumerate}
\end{proposition}

\begin{proof}
The proof follows along similar lines as in \cite{AG3}, where we prove the analogous statement in the case of the Le--graph.
Let $t_W, t_B, d_W$ and $d_B$ respectively be the number of trivalent white, trivalent black, bivalent white and bivalent black internal vertices of ${\mathcal G}$. Let $n_I$ be the number of internal edges ({\sl i.e.} edges not connected to a boundary vertex) of ${\mathcal G}$. By Euler formula we have $g = n_I +n -(t_W + t_B+d_W+d_B)$.
Moreover, there hold the following relations between the number of edges and that of vertices
$3(t_W+t_B)+2(d_W+d_B) = 2n_I +n$, $2t_B+ t_W+d_W+d_B =n_I+k$.
Therefore 
\begin{equation}\label{eq:vertex_type}
t_W = g-k, \qquad t_B = g-n+k, \qquad d_W+d_B= n_I+ 2n - 3g.
\end{equation}
By definition, $\Gamma$ is represented by $d= 1+t_W + t_B+d_W+d_B =n_I+n-g+1$ complex projective lines which intersect generically in $d(d-1)/2$ double points. The regular curve is obtained keeping  the number of ovals fixed while perturbing the double points corresponding to the edges in $\mathcal G$ creating regular gaps (see \cite{AG3} for explicit formulas for the perturbations). The total number of desingularized double points, $N_d$ equals the total number of edges in $\mathcal G$: $N_d =n_I+n$. 
Then the genus of the smooth curve is given by the following formula $N_d - d+1 = g $.
\end{proof}

In the construction in Proposition~\ref{prop:rational_curve} generically each line intersects more than 3 other lines, but only 2 or 3 of these intersection correspond to double point the rational ${\mathtt M}$-curve. Therefore we use partial normalization  of the nodal plane curve (see \cite{ACG})  to remove unnecessary intersections.

When we pass from the rational ${\mathtt M}$-curve to the smooth one, we have two types on constraints on the perturbation of the equation:
\begin{enumerate}
\item The ``extra'' intersections remain double points so that they can be resolved using normalization; 
\item The ovals of the rational curve are preserved under the perturbation. This is controlled by the signs of the perturbation at the remaining double points.
\end{enumerate}  
The construction of smooth ${\mathtt M}$-curves fulfilling both requirements is explicit and follows along similar lines as in \cite{AG3}, see also \cite{Kr4}. In Figure \ref{fig:mcurveex} we show such continuous deformation for the example in Figure \ref{fig:net_curve} after the elimination of all copies of $\mathbb{CP}^1$ corresponding to bivalent vertices in the network. In this case $d=14$ ($\Gamma$ is representable by the union of two quadrics and 10 lines). Under genericity assumptions, the reducible rational curve $\Gamma$ has $\Delta=89$ singular points before partial normalization. The perturbed regular curve is obtained by perturbing $N_d=21$ ordinary intersection points (for each of them $\delta=1$), corresponding to the double points in the topological representation in Figure~\ref{fig:mcurveex}[left]. $\Delta-N_d = 68$ points remain intersections after this deformation and are resolved during normalization. Finally, the normalized perturbed regular curve has then genus $g=\frac{(d-1)(d-2)}{2}-\Delta+N_d=10$.

\begin{figure}
  \centering
  {\includegraphics[width=0.53\textwidth]{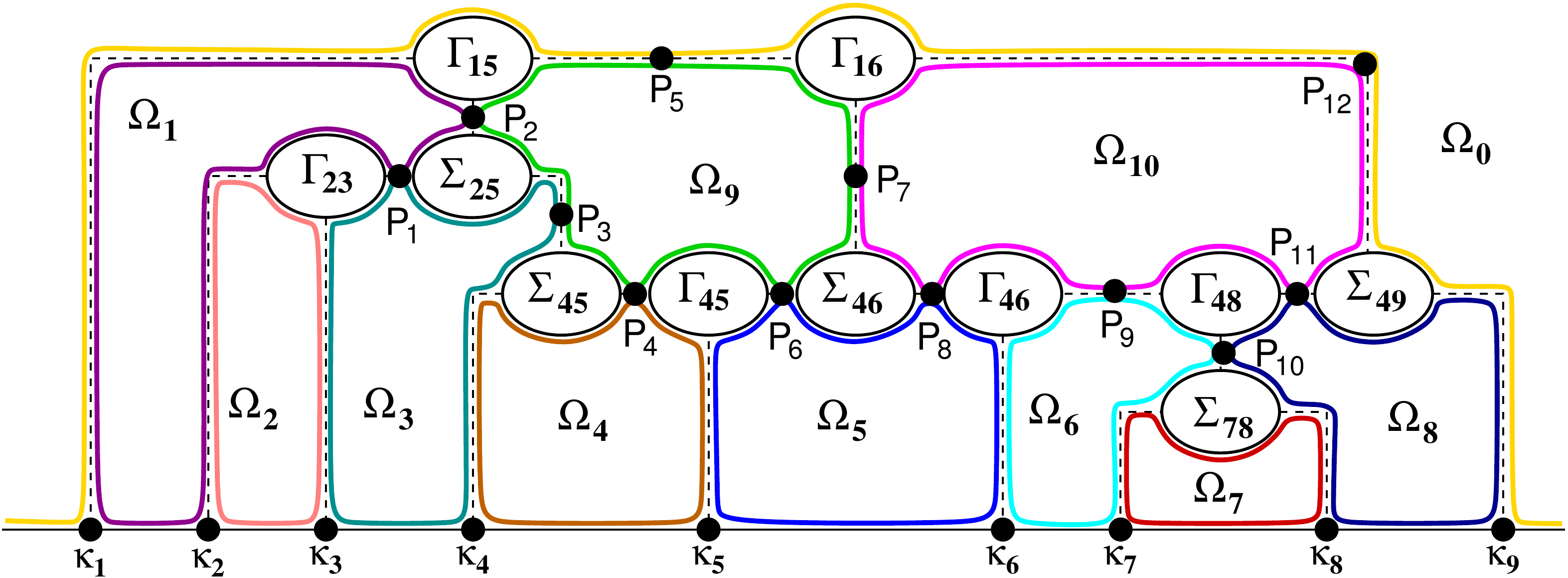}}
  \hfill
  {\includegraphics[width=0.46\textwidth]{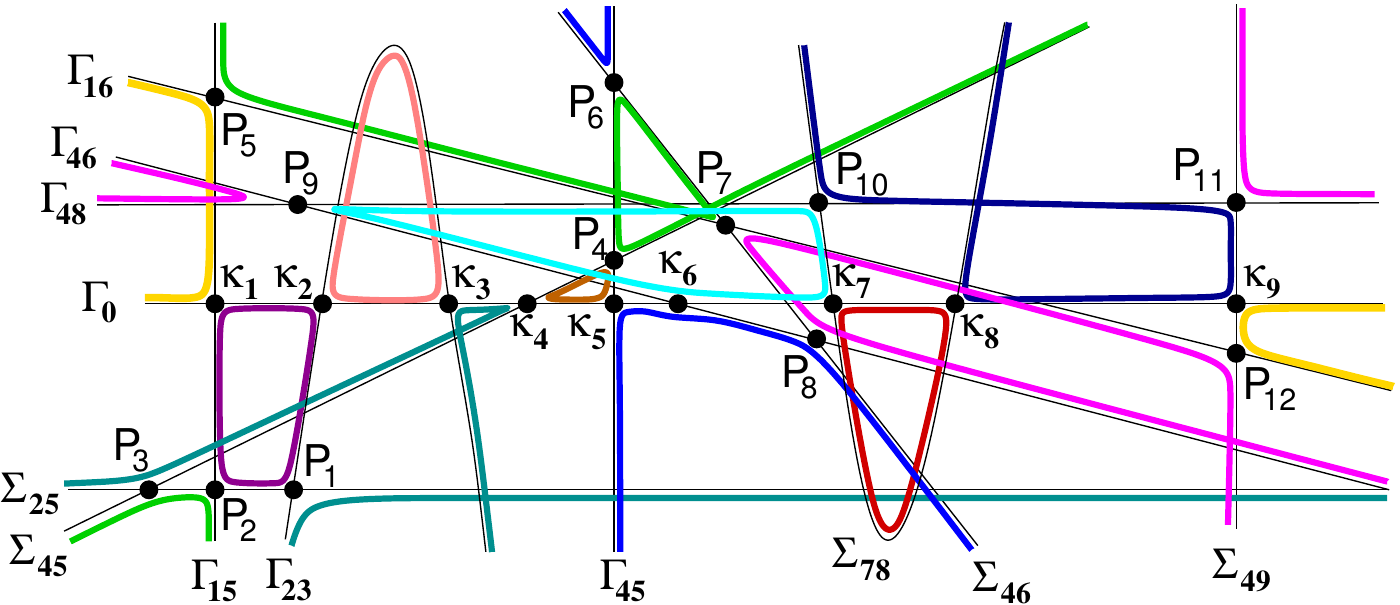}}
  \caption{\small{\sl The topological model of the oval structure of $\Gamma$ [left] is the partial normalization of the reducible plane 
nodal curve [right].  The nodal points, $P_s$, $s\in [12]$, and $\kappa_j$, $j\in [9]$, surviving the partial normalization are marked on the nodal curve [right] and are represented by dashed segments as usual on topological model [left].}\label{fig:mcurveex} }
\end{figure}

\subsection{The KP divisor on $\Gamma(\mathcal G)$ for the soliton data $({\mathcal K}, [A])$}
\label{sec:KPdiv}

Throughout this Section we fix both the soliton stratum $({\mathcal K}, \S)$, with ${\mathcal K} =\{ \kappa_1 < \cdots <\kappa_n\}$ and $\S \subset \GTNN$ an irreducible positroid cell of dimension $|D|$, and the plabic graph ${\mathcal G}$ in the disk representing $\S$ in Postnikov equivalence class. Let $g+1$ be the number of faces of $\mathcal G$ with $g\ge |D|$, and let $\Gamma=\Gamma(\mathcal G)$ be the curve corresponding to $\mathcal G$ as in Main~construction~\ref{def:gamma}. 
In this Section we state the main results of our paper: 
\begin{enumerate}
\item \textbf{Construction of the KP divisor $\DKP$ on a given curve $\Gamma$ for soliton data $(\mathcal K, [A])$:}
For given data $({\mathcal K},[A]\in\S;{\mathcal G})$ we assume that the network $\mathcal N$ of graph ${\mathcal G}$ representing $[A] \in \S$ is generic, i.e. there does not exist an edge such that the edge KP wave function vanishes identically on it. Then we prove that the curve $\Gamma=\Gamma(\mathcal G)$ can be used as the spectral curve for the soliton data $(\mathcal K, [A])$. Indeed, there exists a time $t_0$ such that the wave function at the double points of the curve (i.e. edges of the graph) is different from zero, and one can construct a unique degree $g$ non-special effective real and regular KP divisor $\DKP$ and the unique real and regular normalized KP wave function $\hat \psi(P,\vec t)$ such that on $\Gamma\backslash P_0$ we have $(\hat \psi(P,\vec t))\ge \DKP$, where $(f)$ denotes the divisor 
of $f$.

In Section \ref{sec:constr_null} we briefly illustrate the modification to our construction for an example of non--generic network, {\sl i.e.} a network of reducible graph ${\mathcal G}$ representing $[A] \in \S$ such that there exists an edge at which the KP wave function vanishes identically in KP times. We plan to discuss thoroughly the problem of networks admitting zero vectors in a future publication.
\item \textbf{Invariance of $\DKP$} The construction of $\DKP$ is carried out using a 
directed network ${\mathcal N}$ representing $[A]$ of graph ${\mathcal G}$ and fixing a gauge ray direction. We prove that $\DKP$ is independent on the weight gauge, the vertex gauge, the gauge ray direction and the orientation of the network. In particular, if $\mathcal G$ is a reduced graph move--equivalent to the Le--graph, we get a local parametrization of $\S$ via KP divisors on $\Gamma(\mathcal G)$.
\item \textbf{Transformation laws between curves and divisors induced by Postnikov moves and parallel edge reductions on networks:} Postnikov \cite{Pos} introduces moves and reductions which transform networks preserving the boundary measurement, thus classifying networks representing a given point $[A] \in \S$. In our construction, all Postnikov moves and and parallel edge reductions preserve the class of graphs, and, in Section \ref{sec:moves_reduc}, we provide the explicit transformation of the divisor under the action of such moves and reductions.
\end{enumerate} 

Throughout the paper, we assign affine coordinates to each component of $\Gamma({\mathcal G})$ using the orientation of the graph $\mathcal G$, and we use the same symbol $\zeta$  for any such affine coordinate to simplify notations.

\begin{figure}
  \centering
  {\includegraphics[width=0.60\textwidth]{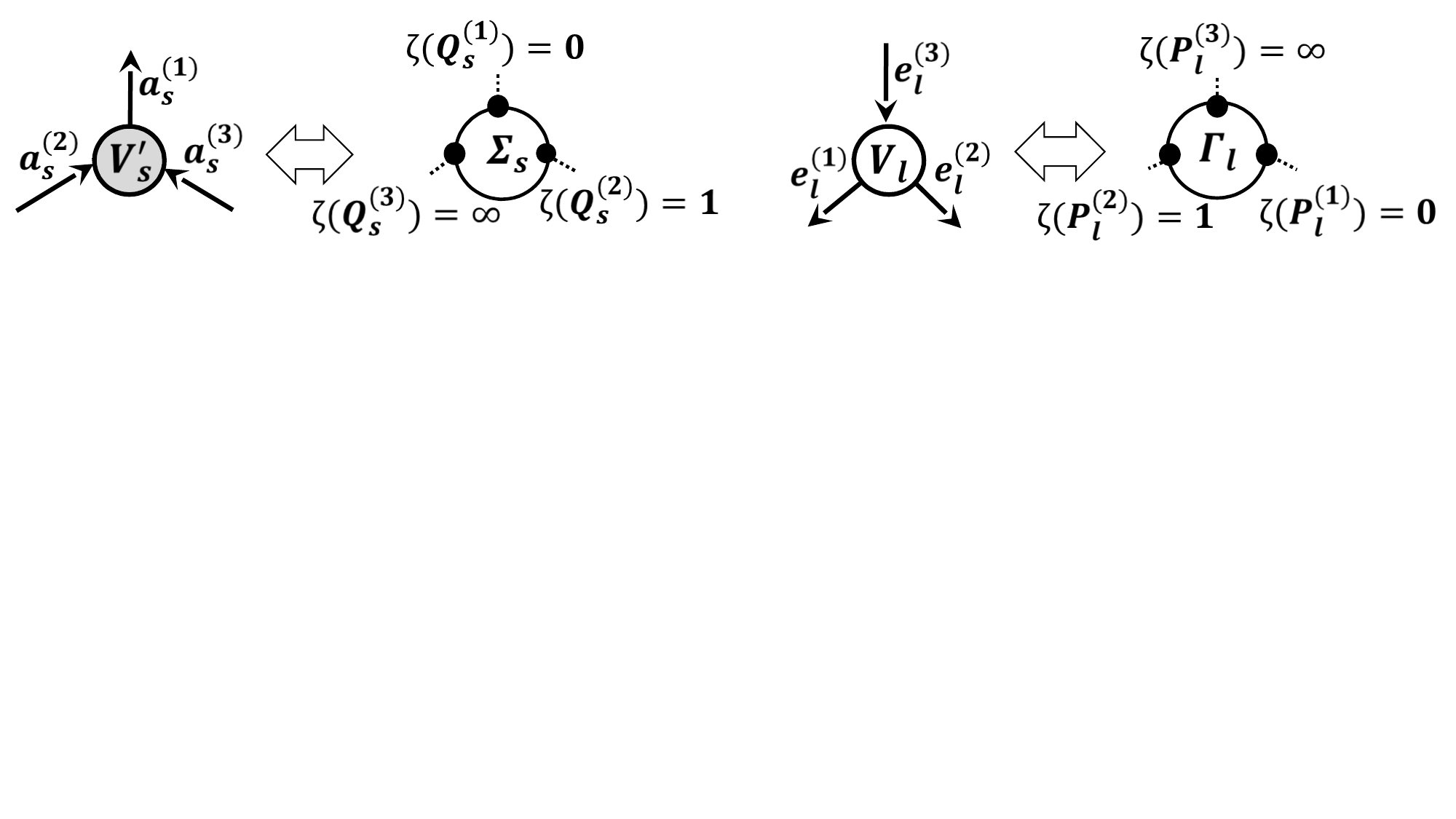}}
	\vspace{-3.8 truecm}
  \caption{\small{\sl Local coordinates on the components $\Gamma_{l}$ and $\Sigma_{s}$ (reflection w.r.t. the vertical line): the marked point $P^{(m)}_{l}\in \Gamma_{l}$ corresponds to the edge $e^{(m)}_{l}$ at the white vertex $V_{l}$ and the marked point $Q^{(m)}_{s}\in \Sigma_{r}$ corresponds to the edge $a^{(m)}_{s}$ at the black vertex $V_{s}^{\prime}$.}}
	\label{fig:lcoord}
\end{figure}

\begin{definition}\label{def:loccoor}{\bf Affine coordinates on $\Gamma(\mathcal G)$}
On each copy of $\mathbb{ CP}^{1}$ the local coordinate $\zeta$ is uniquely identified by the following properties: 
\begin{enumerate}
\item On $\Gamma_0$, $\zeta^{-1} (P_0)=0$ and $\zeta(\kappa_1)< \cdots < \zeta(\kappa_n)$. To abridge notations, we identify the $\zeta$--coordinate with the marked points $\kappa_j=\zeta(\kappa_j)$, $j\in [n]$;
\item On the component $\Gamma_{l}$ corresponding to the internal white vertex $V_{l}$:
\[
\zeta (P^{(1)}_{l}) =0, \;\;\zeta (P^{(2)}_{l}) =1, \;\;\zeta (P^{(3)}_{l}) =\infty,
\]
whereas on the component $\Sigma_{s}$ corresponding to the internal black vertex $V^{\prime}_{s}$:
\[
\zeta (Q^{(1)}_{s}) =0, \;\;\zeta (Q^{(2)}_s) =1, \;\;\zeta (Q^{(3)}_{s}) =\infty.
\] 
\end{enumerate}
\end{definition}

We illustrate Definition \ref{def:loccoor} in Figure~\ref{fig:lcoord} (see also Figure~\ref{fig:markedpoints}).

In view of Definition \ref{def:rrss}, the desired properties of the KP divisor and of the KP wave function on $\Gamma(\mathcal G)$ for given soliton data $(\mathcal K, [A])$ are the following.

\begin{definition}\label{def:real_KP_div}\textbf{Real regular KP divisor compatible with $[A] \in {\mathcal S}_{\mathcal M}^{\mbox{\tiny TNN}}$.}  
Let $\Omega_0$ be the infinite oval containing the marked point $P_0\in\Gamma_0$ and let $\Omega_s$, $s\in[g]$ be the finite ovals of $\Gamma$. Let $\DS=\DS (\vec t_0)$ be the Sato divisor for the soliton data $(\mathcal K, [A])$.
We call a degree $g$ divisor $\DKP\in\Gamma\backslash \{ P_0\}$ a real and 
regular KP divisor compatible with $(\mathcal K, [A])$ if:
\begin{enumerate}
\item $\DKP\cap \Gamma_0 = \DS$; 
\item There is exactly one divisor point on each component of $\Gamma$ corresponding to a trivalent white vertex; 
\item\label{item:defodd_KP} In any finite oval $\Omega_s$, $s\in [g]$, there is exactly one divisor point;
\item\label{item:defeven_KP} In the infinite oval $\Omega_0$, there is no divisor point.
\end{enumerate}
\end{definition}

\begin{definition}\label{def:KPwave}{\bf A real regular KP wave function on $\Gamma$ corresponding to $\DKP$:}  
Let $\DKP$ be a degree $g$ real regular divisor on $\Gamma$ satisfying Definition~\ref{def:real_KP_div}.
A function ${\hat \psi} (P, \vec t)$, where $P\in\Gamma\backslash \{ P_0\}$ and $\vec t$ are the KP times, is called a real and regular KP wave function on $\Gamma$ corresponding to $\DKP$ if:
\begin{enumerate}
\item $\hat \psi (P, \vec t_0)=1$ at all points $P\in \Gamma\backslash \{P_0\} $;
\item The restriction of $\hat \psi$ to $\Gamma_0\backslash \{P_0\}$ coincides with the normalized Sato wave function defined in (\ref{eq:SatoDN}): ${\hat \psi } (P, \vec t)  = \frac{\psi^{(0)} (P; \vec t)}{\psi^{(0)} (P; \vec t_0)}$; 
\item For all $P\in\Gamma\backslash\Gamma_0$ the function ${\hat \psi} (P, \vec t)$ satisfies all equations (\ref{eq:dress_hier}) of the dressed hierarchy;
\item If both $\vec t$ and $\zeta(P)$ are real, then ${\hat \psi} (\zeta(P), \vec t)$ is real. Here $\zeta(P)$ is the local affine coordinate on the corresponding component of $\Gamma$ as in Definition~\ref{def:loccoor};
\item\label{it:comp} $\hat \psi$ takes equal values at pairs of glued points $P,Q\in \Gamma$, for all $\vec t$:  $\hat \psi(P, \vec t) = \hat \psi(Q, \vec t)$;
\item For each fixed $\vec t$ the function $\hat \psi(P, \vec t)$ is meromorphic of degree $\le g$ in $P$ on $\Gamma\backslash \{ P_0\}$: for any fixed $\vec t$ we have $({\hat \psi} (P, \vec t))+\DKP\ge 0$ on $\Gamma\backslash P_0$, where $(f)$ denotes the divisor 
of $f$. Equivalently, for any fixed $\vec t$ on $\Gamma\backslash \{ P_0\}$ the function $\hat \psi(\zeta, \vec t)$ is regular outside the points of 
$\DKP$ and at each of these points either it has a first order pole or is regular;
\item For each $P\in\Gamma\backslash \{ P_0\}$ outside  $\DKP$  the function $\hat \psi(P, \vec t)$ is regular in $\vec t$ for all times.
\end{enumerate}
\end{definition}

\begin{theorem}\label{theo:exist}\textbf{Existence and uniqueness of a real and regular KP divisor and KP wave function 
on $\Gamma$.}
Let the phases ${\mathcal K}$, the irreducible positroid cell $\S\subset \GTNN$, the plabic graph $\mathcal G$ representing $\S$ be fixed. Let $\Gamma=\Gamma(\mathcal G)$ with marked point $P_0\in \Gamma_0$ be as in Construction \ref{def:gamma}. 

If $\mathcal G$ is reduced and equivalent to the Le--network via a finite sequence of moves of type (M1), (M3) and flip moves (M2), there are no extra conditions and, for any $[A]\in \S$, let $\mathcal N$ of graph $\mathcal G$ be a network representing $[A]$. 
Otherwise if $\mathcal G$ is reducible, we consider only generic weights such that the edge vectors on the corresponding  network $\mathcal N$ of graph $\mathcal G$ are all non-zero.

Then, there exists a reference time $\vec t_0$ such that, to the following data $({\mathcal K}, [A]; \Gamma, P_0; \mathcal N; \vec t_0 )$, we associate
\begin{enumerate}
\item A \textbf{unique} real regular degree $g$ KP divisor $\DKP$ as in Definition \ref{def:real_KP_div};
\item A \textbf{unique} real regular KP wave function $\hat \psi(P, \vec t)$ corresponding to this divisor satisfying Definition~\ref{def:KPwave}.
\end{enumerate}
Moreover, $\DKP$ and $\hat \psi(P, \vec t)$ are both invariant with respect to changes of the geometric positions of the  vertices, of the weight gauge, of the ray direction gauge and of the orientation of the graph $\mathcal G$.
\end{theorem}

{\sl Sketch of the proof: } The construction of the KP divisor $\DKP$ and of the Baker--Akhiezer function $\hat \psi(P, \vec t)$ on $\Gamma(\mathcal G)$ is carried in several steps:
\begin{enumerate}
\item\textbf{The edge wave function} In Section \ref{sec:anycurve}, we define a non normalized dressed edge wave function $\Psi_e (\vec t)$  at the edges $e\in \mathcal N$ which solves the system of Lam relations introduced in Section \ref{sec:vectors} using the non-normalized Sato wave function as the boundary condition. We then assign a degree $g-k$ dressed network divisor $\DDN$ to $\mathcal N$ using the linear relations at the trivalent white vertices. 
The dependence of both the edge wave function and of the network divisor on the gauge ray direction, the weight gauge, the vertex gauge and the orientation of the network is ruled by the corresponding transformation properties of the geometric signature discussed in Section \ref{sec:vectors}, see Lemma~\ref{lem:gauge}, Theorem~\ref{theo:sign_face}; 
\item\textbf{The KP divisor on $\Gamma$} In Section \ref{sec:KP_div}, we define the KP divisor on $\Gamma$. We rule the value of the KP wave function at the double points using the system of edge wave functions defined in the previous step.
Therefore each dressed network divisor number is the local coordinate of a KP divisor point $\Pdr_l \in \Gamma_l$ and the position of $\Pdr_l$ is independent on the orientation of the network, on the gauge ray direction, on the weight gauge and on the vertex gauge. This set of divisor points has degree $g-k$ equal to the number of trivalent white vertices in the plabic graph $\mathcal G$. The KP divisor $\DKP$ is defined as the sum of this divisor with the degree $k$ Sato divisor $\DS$;
\item\textbf{Counting the divisor points in the ovals} In Section \ref{sec:comb} we prove that there is one divisor point in each finite oval and no divisor point in the infinite oval;
\item\textbf{The KP wave function on $\Gamma$} The normalized KP wave function $\hat \psi (P, \vec t)$ on $\Gamma$ is constructed imposing that
\begin{itemize}
\item It coincides with the Sato wave function on $\Gamma_0$;
\item It coincides for all times at each double point $P$ with the value of the normalized dressed edge wave function defined in Step~2;
\item The normalized wave function is then meromorphically extended to all other components of $\Gamma$ so that, for any fixed $\vec t$, we have $({\hat \psi} (P, \vec t))+\DKP\ge 0$ on $\Gamma\backslash P_0$.
\end{itemize}
\end{enumerate}

\smallskip

As a consequence of Theorem \ref{theo:exist} and of the absence of identically zero edge wave functions on reduced networks, we get the first statement in the following Corollary. The second part follows from the explicit characterization of the effect of moves and reductions on the wave function and the divisor carried in Section \ref{sec:moves_reduc}.

\begin{corollary}\label{cor:param}
Under the hypotheses of Theorem \ref{theo:exist}:
\begin{enumerate}
\item \textbf{Local parametrization of $\S$ via KP divisors:}
If the plabic graph $\mathcal G$ representing $\S$ is reduced and equivalent to the Le--graph via a finite sequence of moves (M1), (M3) and flip moves (M2), then for any fixed
$\vec t_0$, there is a local one-to-one correspondence between KP divisors on $\Gamma (\mathcal G)$ and points $[A]\in \S$. 
\item \textbf{Discrete transformation between curves and divisors induced by moves and reductions:}
Let ${\mathcal G}$ and ${\mathcal G}^{\prime}$ be two plabic graphs equivalent by a finite sequence of moves and reductions for which Theorem \ref{theo:exist} holds true for the same $\vec t_0$, then there is an explicit transformation of the KP divisor on
$\Gamma(\mathcal G)$ to the KP divisor on $\Gamma({\mathcal G}^{\prime})$. 
\end{enumerate}
\end{corollary}

\begin{remark}\textbf{Global parametrization of positroid cells via divisors} We claim that given a curve $\Gamma$ associated to a reduced plabic graph $\mathcal G$ representing a positroid cell $\S$, real and regular KP divisors $\DKP$ such that $\# ( \DKP \cap \Gamma_0) = k$ provide a global minimal parametrization of $\S$ after applying some blow-ups in all cases where some of the divisor points occur at double points. We plan to discuss thoroughly this issue in a future publication and we just discuss the question for soliton data associated to $Gr^{\mbox{\tiny TP}} (1,3)$ in Section \ref{sec:global}.
\end{remark}

\section{Construction of the KP wave function at the double points of the spectral curve}\label{sec:anycurve}

In this Section we construct dressed edge wave functions on networks and introduce effective dressed network divisors. 

Throughout this Section, we fix the phases ${\mathcal K}= \{ \kappa_1 <\cdots< \kappa_n\}$ and the plabic graph in the disk $\mathcal G$ representing the irreducible positroid cell $\S \subset \GTNN$ as in Definition \ref{def:graph}. $\Gamma =\Gamma(\mathcal G)$ is the curve as in Construction
\ref{def:gamma}. We denote by $g+1$ the number of faces (ovals) of $\mathcal G$ ($\Gamma$).
$\mathcal O$ is the orientation of $\mathcal G$ and $\mathfrak l$ a fixed gauge direction. Finally $I=\{ 1\le i_1<\cdots < i_k\le n\}$ is the base in $\mathcal M$ associated to $\mathcal O$, whereas
$\bar I= [n] \backslash I$.

For any $[A]\in \S$, $({\mathcal N}, \mathcal O)$ is a directed network of graph $\mathcal G$ representing such point. Using the weight gauge freedom, in the following we assume that edges at boundary vertices carry unit weights. $z_{b_e}$ is the system of half-edge vectors constructed in Theorem~\ref{thm:lam_rel}  on $({\mathcal N}, \mathcal O, \mathfrak l)$ with boundary conditions $E_j$ at the boundary vertices and we assume from now on that all half-edge vectors are not zero.  We recall that, if an half-edge vector is zero in one orientation then it is zero in any other orientation of $\mathcal N$  and that zero vectors can't appear in plabic networks possessing an acyclic orientation (see \cite{AG4} and Remark \ref{rem:zero_vectors}).

\subsection{Main construction -- part II. The dressed half-edge wave function}\label{sec:vertex_wavefn_general_case}

For the rest of the paper we denote ${\mathfrak E}_\theta (\vec t) = (e^{\theta_1(\vec t)}, \dots, e^{\theta_n(\vec t)})$, $\theta_j(\vec t) = \sum_{l\ge 1} \kappa_j^{l} t_l$, where $\vec t =(t_1=x,t_2=y,t_3=t, t_4,\dots)$ are the KP times, and 
 $\prec \cdot, \cdot\succ $ denotes the usual scalar product. Moreover we assume that only a finite number of entries of $\vec t$ are non zero. To simplify the construction we suppose that the edges at the boundary vertices are parallel to each other and each one ends at an internal bivalent vertex. 

\begin{definition}\label{def:vvw_gen}\textbf{The dressed half-edge wave function (dressed h.e.w.) on $({\mathcal N}, \mathcal O, \mathfrak l)$.}
  Let $({\mathcal N}, \mathcal O, \mathfrak l)$ be the oriented network associated to the soliton data $(\mathcal K, [A])$, where the perfect orientation $ \mathcal O$ is associated to the base $I$. Let us denote $z_{U,e}$ its system of half--edge vectors with respect to the boundary conditions
\begin{equation}
  \label{eq:bc11}  
z_{b_j} =  E_j, \ \  j\in \bar I.
\end{equation}
  Finally let $A$ be the RREF matrix of $[A]$ w.r.t. the base $I$ so that
$f^{(r)} (\vec t) = \sum_{j=1}^n A^r_j e^{\theta_j(\vec t)}$, $r\in [k]$,
are the heat hierarchy solutions generating the Darboux transformation $\mathfrak D^{(k)}$ for the soliton data $(\mathcal K, [A])$.

We define the \textbf{dressed half--edge wave function (dressed h.e.w.)} for the half--edge $(e,U)$ as 
\begin{equation}\label{eq:KP_half_ef}
\Psi_{U,e; \mathcal O,\mathfrak l} (\vec t) \equiv \prec z_{U,e}, {\mathfrak D}^{(k)}{\mathfrak E}_\theta (\vec t)\succ .
\end{equation}
\end{definition}
In particular, the dressed half-edge wave function at the boundary vertices coincides with the Sato wave fucntion
\begin{equation}\label{eq:KP_half_ef2}
\Psi_{U,e_j; \mathcal O,\mathfrak l} (\vec t) \equiv \prec z_{b_j}, {\mathfrak D}^{(k)}{\mathfrak E}_\theta (\vec t) \succ = {\mathfrak D}^{(k)}\phi^{(0)} (\kappa_j; \vec t), \ \ j\in [n],
\end{equation}
because of (\ref{eq:bc11}),
\begin{equation}
  \label{eq:bc21}  
z_{b_{i_r}, e_{i_r}} =E_{i_r}- A[r],
\end{equation}
and the fact that the heat hierarchy solutions belong to the kernel of the dressing operator.

\begin{remark}
  From now on we assume that our network is generic, i.e. all half-edge vectors are different from zero.
\end{remark}

\begin{lemma}\label{lem:ref_time}
Let $\Psi_{U,e_j; \mathcal O,\mathfrak l} (\vec t)$ be the dressed h.e.w. function. Then: 
\begin{enumerate}
\item  $\Psi_{e, \mathcal O,\mathfrak l} (\vec t) \not \equiv 0$, for all $e\in \mathcal N$;
\item  There exists a reference time $\vec t_0= (x_0,0,\dots)$ such that
for any $e\in \mathcal E$, the dressed h.e.w. function $\Psi_{e, \mathcal O,\mathfrak l} (\vec t_0) \not = 0$.
  \end{enumerate}
\end{lemma}

The full rank linear system for the half-edge vectors (see Section \ref{sec:def_edge_vectors}) induces a full rank linear system satisfied by half--edge wave function:
\begin{enumerate}
\item 
At any trivalent white vertex $U$ incident with incoming edge $e_3$ and outgoing edges $e_1$, $e_2$,
\begin{equation}\label{eq:lin_Phi1}
\displaystyle \sum_{k=1}^3 \Psi_{U,e_k; \mathcal O,\mathfrak l} (\vec t) =0;
\end{equation}
\item At any trivalent black vertex $U$ incident with incoming edges  $e_2$, $e_3$ and outgoing edge $e_1$,
\begin{equation}\label{eq:lin_Phi3}
  \Psi_{U,e_1; \mathcal O,\mathfrak l} (\vec t) = \Psi_{U,e_2; \mathcal O,\mathfrak l} (\vec t) = \Psi_{U,e_3; \mathcal O,\mathfrak l} (\vec t);
 \end{equation} 
\item At each edge $e=(U,V)$,
  \begin{equation}\label{eq:lin_Phi4}
 \Psi_{U,e; \mathcal O,\mathfrak l} (\vec t) = (-1)^{\epsilon_U,V} w_{U,V}  \Psi_{V,e; \mathcal O,\mathfrak l} (\vec t),    
\end{equation}
where $\epsilon_{U,V}$ is the geometric signature defined in Section~\ref{sec:def_edge_vectors}.
\end{enumerate}

Finally, in the following Proposition the transformation rules for the half--edge wave functions are a consequence of the transformation rules for the half-edge vectors (Lemma \ref{rem:gauge_weight} and Theorem \ref{theo:sign_face}). The last item is a key observation which will be used both to prove the invariance of the position of the divisor point in the oval with respect to the geometric transformations of the directed graph dual to the reducible $\mathtt M$--curve and to count the number of divisors points in each oval.

\begin{proposition}\label{prop:change_orient}\textbf{The dependence of the dressed e.w. on the geometric transformations}
Let $\epsilon^{(1)}_{U,V}$ and $\epsilon^{(2)}_{U,V}$ respectively be the edge signatures on 
$\mathcal N^{(1)} \equiv (\mathcal N, \mathcal O, \mathfrak l)$ and on ${\mathcal N}^{(2)} \equiv (\hat{ {\mathcal N}}, \hat{ \mathcal O},\hat{\mathfrak l})$, plabic networks representing the same point $[A]\in \S \subset \GTNN$. Denote $\Psi^{(i)}_{U,e} (\vec t)$, $i=1,2$, the respective half--edge wave functions at the the half--edge $(U,e)$. Then:
\begin{enumerate}
\item If ${\mathcal N}^{(2)}$ is obtained from ${\mathcal N}^{(1)}$ by either changing the gauge ray direction on the same directed graph or acting with a vertex gauge transformation on the graph of ${\mathcal N}^{(1)}$ keeping fixed the orientation, 
\begin{equation}\label{eq:phi_gauge} 
\Psi^{(2)}_{U,e} (\vec t) = (-1)^{\eta(U)} \Psi^{(1)}_{U,e} (\vec t),
\end{equation}
where $\eta(U)$ is the gauge equivalence transformation of the two signatures;
\item If ${\mathcal N}^{(2)}$ is obtained from ${\mathcal N}^{(1)}$ by changing the orientation of the given graph, 
then, for any edge $e$ there exists a real constant $\alpha_e\not = 0$ such that 
\begin{equation}\label{eq:phi_orient}
\Psi^{(2)}_{U,e} (\vec t) = \alpha_e \Psi^{(1)}_{U,e} (\vec t);
\end{equation}
\item Let $V$ be an internal vertex and $e,f$ be edges at $V$. If ${\mathcal N}^{(2)}$ is obtained from ${\mathcal N}^{(1)}$ by either changing the orientation of the given graph, or the gauge ray direction on the same directed graph or acting with a vertex gauge transformation or a weight gauge transformation, then, for any $\vec t$
\begin{equation}\label{eq:sign_prod}
\mbox{sign }\left( \Psi^{(2)}_{U,e} (\vec t)\Psi^{(2)}_{U,f} (\vec t) \right) = \mbox{sign } \left( \Psi^{(1)}_{U,e} (\vec t)\Psi^{(1)}_{U,f} (\vec t) \right).
\end{equation}
\end{enumerate}
\end{proposition}

The last statement may be easily checked case by case using the definition of the half-edge wave function and items (1) and (2) in the same Proposition.

\subsection{Main construction -- part III. The dressed network divisor}

Next we associate the dressed network divisor to $({\mathcal N},\mathcal O, \mathfrak l)$ as follows.

\begin{definition}\label{def:vac_div_gen}\textbf{The dressed network divisor $\DDN$.}
At each trivalent white vertex $V$ of $({\mathcal N},\mathcal O,\mathfrak l)$, 
let $\Psi_{V,e_m; \mathcal O,\mathfrak l} (\vec t)$ be the dressed h.e.w. on the edge $e_m$, $m\in [3]$, with the convention that the edges are numbered anticlockwise and $e_3$ is the incoming edge at $V$. Then, to $V$ we assign the dressed network divisor number 
\begin{equation}\label{eq:dress_pole_def}
\gamma_{\textup{\scriptsize dr},V} = \frac{\Psi_{V,e_1;\mathcal O,\mathfrak l} (\vec t_0) }{\Psi_{V,e_1;\mathcal O,\mathfrak l} (\vec t_0)+ \Psi_{V,e_2;\mathcal O,\mathfrak l} (\vec t_0)}.
\end{equation}
We call $\DDN = \{ (\gdr, V_l), \; l\in [g-k]  \}$ the dressed network divisor on ${\mathcal N}$, where $V_l$, $l\in [g-k]$, are the trivalent white vertices of the network.
\end{definition}

If $\mathcal N$ is the acyclically oriented Le--network representing $[A]$, then the dressed network divisor constructed in \cite{AG3} coincides with the above definition for the choice of the gauge ray $\mathfrak l$ as in Figure \ref{fig:Rules0}.

By definition, the dressed network divisor number at each trivalent white vertex is real and represents the local coordinate of a point on the corresponding copy of $\mathbb{CP}^1$ in $\Gamma$. In general, on $({\mathcal N}, \mathcal O,\mathfrak l)$, the value of the dressed network divisor number at a given trivalent white vertex $V$ depends on the choice of $\vec t_0$. If this is not the case, we call \textbf{trivial} the corresponding \textbf{network divisor number}. The latter case occurs at white vertices where the linear system involves proportional vectors. If ${\mathcal N}$ is the Le--network, there are no trivial network divisor numbers. 
 
\begin{lemma}\label{lemma:trivial_div}\textbf{Trivial network divisor numbers}
Let $z_{V,e_m}$, $m\in [3]$, be the half-edge vectors at a trivalent white vertex $V$ of $({\mathcal N},\mathcal O,\mathfrak l)$ and $\gamma_{\textup{\scriptsize dr},V}$ be the dressed network divisor number at $V$. If there exists a non zero constant $c_V$ such that either $z_{e_2} = c_V z_{V,e_1}$ or $z_{V,e_3} = c_V z_{V,e_1}$ or $z_{V,e_3} = c_V z_{V,e_2}$, then the dressed network divisor number $\gamma_{\textup{\scriptsize dr},V}$ is trivial.
\end{lemma}

\smallskip

All gauge transformations (change of ray direction, of the vertex gauge and of weight gauge) leave invariant the divisor number on $V$. Those associated to changes of orientation correspond to a well defined change of the local coordinate on the corresponding copy of $\mathbb{CP}^1$, and the dressed network divisor numbers 
change in agreement with such coordinate transformation (Proposition \ref{prop:indep_orient}).
We shall use such transformation properties of the divisor numbers to prove the invariance of the KP divisor on $\Gamma$ in Theorem \ref{theo:inv}.

\begin{proposition}\label{pro:indep_gauge}\textbf{Independence of the network divisor on the gauge ray direction, the weight gauge and the vertex gauge.}
Let ${\mathcal N}^{(1)}$, ${\mathcal N}^{(2)}$ be two oriented networks representing the same point in the Grassmannian and obtained from each other by changing either the gauge ray direction or the weight gauge or the vertex gauge.
Let $\gamma_{\textup{\scriptsize dr},V}^{(i)}$ respectively be the divisor numbers 
at the vertex $V$ in ${\mathcal N}^{(i)}$, $i=1,2$. Then
\[
\gamma_{\textup{\scriptsize dr},V}^{(2)} = \gamma_{\textup{\scriptsize dr},V}^{(1)}.
\]
\end{proposition}

The proof follows by computing the divisor numbers at $V$ in (\ref{eq:dress_pole_def}) using Proposition~\ref{prop:change_orient}.

\smallskip

\begin{proposition}\label{prop:indep_orient}\textbf{The dependence of the dressed network divisor on the orientation.}
Let ${\mathcal O}_1$ and ${\mathcal O}_2$ be two perfect orientations for network ${\mathcal N}$.
Then for any trivalent white vertex $V$, such that all edges at $V$ have the same versus in both orientations, $(e_1^{(2)},e_2^{(2)},e_3^{(2)})=(e_1^{(1)},e_2^{(1)},e_3^{(1)})$, the dressed network divisor number is the same in both
orientations 
\begin{equation}\label{eq:inep_gauge_1}
\gamma_{\textup{\scriptsize dr},V,\mathcal O_2} = \gamma_{\textup{\scriptsize dr},V,\mathcal O_1}.
\end{equation}
If at the vertex $V$ we change orientation of edges from $(e_1^{(1)},e_2^{(1)},e_3^{(1)})$ to $(e_2^{(2)},e_3^{(2)},e_1^{(2)})=(e_1^{(1)},-e_2^{(1)},-e_3^{(1)})$ (Figure \ref{fig:phi_orient}[left]), then the relation between the dressed network divisor numbers at $V$ in the two orientations is
\begin{equation}\label{eq:inep_gauge_2}
\gamma_{\textup{\scriptsize dr},V,\mathcal O_2} = \frac{1}{1-\gamma_{\textup{\scriptsize dr},V,\mathcal O_1}}.
\end{equation}
If at the vertex $V$ we change orientation of edges from $(e_1^{(1)},e_2^{(1)},e_3^{(1)})$ to $(e_3^{(2)},e_1^{(2)},e_2^{(2)})=(-e_1^{(1)},e_2^{(1)},-e_3^{(1)})$ (Figure \ref{fig:phi_orient}[right]), then the relation between the dressed network divisor numbers at $V$ in the two orientations is
\begin{equation}\label{eq:inep_gauge_3}
\gamma_{\textup{\scriptsize dr},V,\mathcal O_2} = \frac{\gamma_{\textup{\scriptsize dr},V,\mathcal O_1}}{\gamma_{\textup{\scriptsize dr},V,\mathcal O_1}-1}.
\end{equation}
\end{proposition}

\begin{figure}
  \centering{\includegraphics[width=0.55\textwidth]{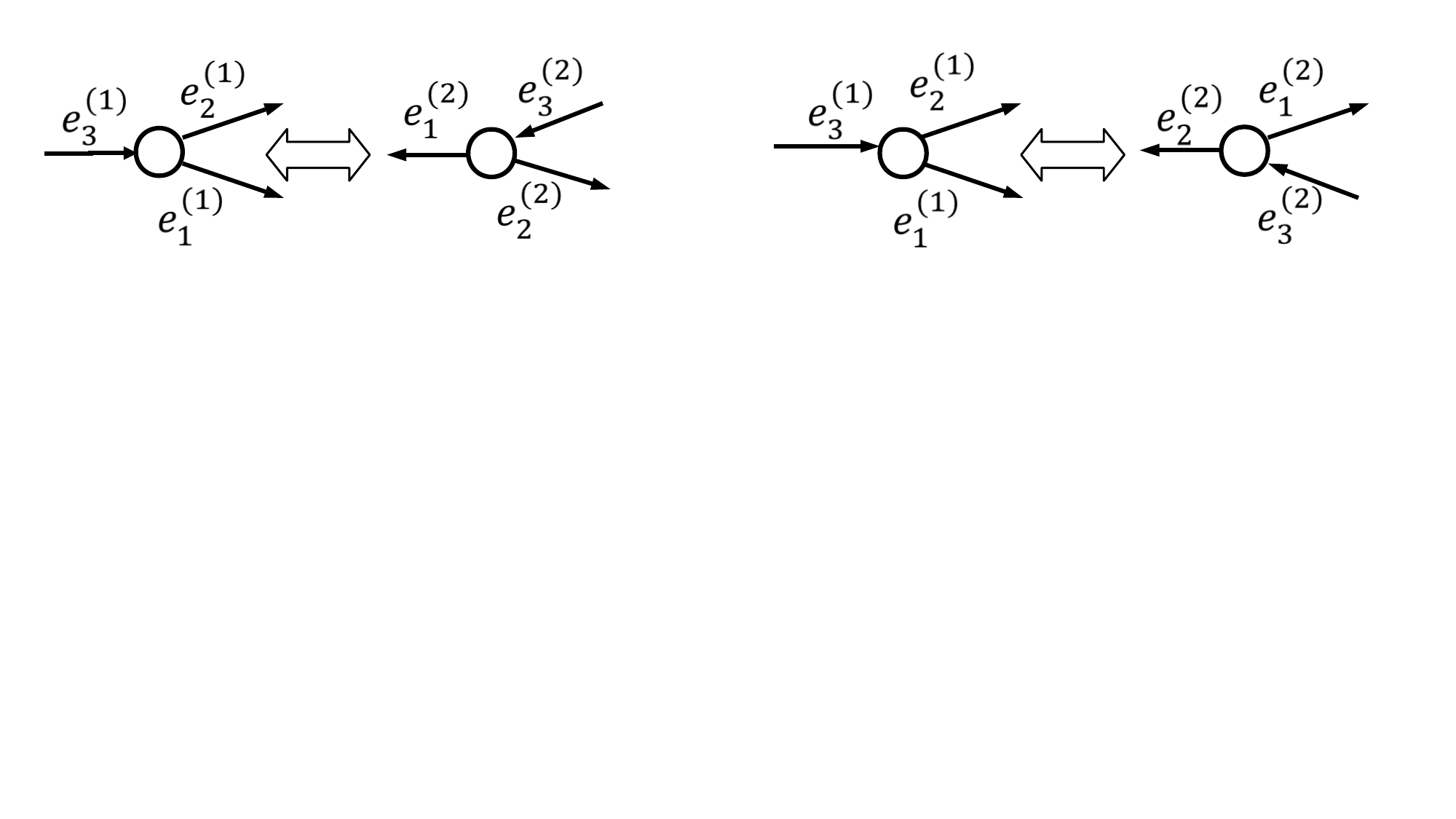}}
	\vspace{-3.7 truecm}
  \caption{\small{\sl The change of orientation at a vertex $V$.}\label{fig:phi_orient}}
\end{figure}

{\sl Proof}
The proof follows by direct computation of the dressed network divisor numbers using (\ref{eq:phi_orient}) and the linear system (\ref{eq:lin_Phi1}) at $V$ for both orientations.

\begin{remark}
In next Section we prove the invariance of the KP divisor on $\Gamma$ with respect to changes of orientation of the graph using  Proposition~\ref{prop:indep_orient}. Indeed, on $\Gamma_l$, the copy of $\mathbb{CP}^1$ corresponding to $V$, the transformation rule of the network divisor number coincides with the change of coordinates of the divisor point induced by the change of orientation of the network. 
\end{remark}

\section{Construction of the KP wave function $\hat \psi$ and characterization of the KP divisor $\DKP$ on $\Gamma$.}\label{sec:inv}

In this Section we define the KP wave function $\hat \psi (P, \vec t)$ on the curve $\Gamma=\Gamma (\mathcal G)$ and the KP divisor $\DKP$ as the union of the Sato divisor on $\Gamma_0$ and the divisor points on the components of $\Gamma$ corresponding to the white vertices whose local coordinates are the network divisor points defined in the previous Section. Then the following properties of $\DKP$ easily follow from its definition:
\begin{enumerate}
\item $\DKP$ is contained in the union of the ovals of $\Gamma$;
\item	The position of the divisor point associated to a given trivalent white vertex $V$ only depends on the relative signs of the KP half--edge wave function at $V$;
\item The KP divisor does not depend on the geometrical indices chosen to construct the network divisor numbers;
\item There is exactly one divisor point in each finite oval and no divisor point in the infinite oval containing $P_0$.
\end{enumerate}

Let us define the normalized dressed edge wave function.

\begin{definition}\label{def:norm_vac_gen}\textbf{The KP edge wave function $\hat \Psi$.}
Let $({\mathcal N}, \mathcal O, \mathfrak l)$ be the gauge--oriented network representing the soliton data $(\mathcal K, [A])$. Let $\vec t_0$ be as in the Lemma~\ref{lem:ref_time}. 
Then the \textbf{KP edge wave function (KP e.w)} on the edge $e=(U,V)$ in ${\mathcal N}$ is 
\begin{equation}\label{eq:KPvfnorN}
\hat \Psi_{e} (\vec t) = \left\{ \begin{array}{ll} \displaystyle\frac{\Psi_{U,e; \mathcal O,\mathfrak l} (\vec t) }{\Psi_{U,e; \mathcal O,\mathfrak l} (\vec t_0) }=  \displaystyle\frac{\Psi_{V,e; \mathcal O,\mathfrak l} (\vec t) }{\Psi_{V,e; \mathcal O,\mathfrak l} (\vec t_0) }, &\quad\quad \mbox{for all } e\in \mathcal E, \\
\displaystyle\frac{{\mathfrak D}^{(k)} e^{\theta_{j} (\vec t)}}{{\mathfrak D}^{(k)} e^{\theta_{j} (\vec t_0)}},&\quad\quad \mbox{if } e=e^{(D)}_j, \;\; j\in [n],
\end{array}\right. \quad\quad \forall \vec t.
\end{equation}
\end{definition}

\begin{remark}\label{rem:indep}\textbf{The KP edge wave function on $\mathcal N$.} 
The name KP wave function for $\hat \Psi_e (\vec t)$ on $\mathcal N$ is fully justified. Indeed $\hat \Psi_e (\vec t)$ just depends on the soliton data $(\mathcal K, [A])$ and on the chosen network representing $[A]$ since $\hat \Psi_e (\vec t)$ takes the same value on a given edge $e$ for any choice of orientation, gauge ray direction, weight gauge and vertex gauge. Moreover, by construction it satisfies the Sato boundary conditions at the edges at the boundary, and takes real values for real $\vec t$. This function is a common eigenfunction to all KP hierarchy auxiliary linear operators $-\partial_{t_j} + B_j$, where $B_j =(L^j)_+$, and the Lax operator $L=\partial_x+\frac{u(\vec t)}{2}\partial_x^{-1}+ u_2(\vec t)\partial_x^{-2}+\ldots$, the coefficients of these operators are the same for all edges. 
$\hat \Psi_e (\vec t)$ takes equal values at all edges $e$ at the same bivalent or black trivalent vertex, whereas, at any trivalent white vertex, it takes either the same value at all three edges for all times or distinct values at some $\vec t\not= \vec t_0$. 
\end{remark}

\subsection{Main construction -- Part IV. The KP wave function and its pole divisor}\label{sec:KP_div}
We start the Section defining both the KP wave function $\hat \psi$ and the KP divisor $\DKP$. We use the half-edge wave function to assign the value of the KP wave function at the double points of $\Gamma$ and then extend it meromorphically on each component of $\Gamma$.

\begin{construction}\label{con:dress_gen}\textbf{The KP wave function $\hat \psi$ on $\Gamma$.}
Let the soliton data $({\mathcal K}, [A])$ be given, with $\mathcal K = \{ \kappa_1 < \cdots < \kappa_n\}$ and $[A]\in \S \subset \GTNN$. Let $I\in \mathcal M$ be fixed.
Let $\mathcal G$ be a plabic graph representing $\S$ as in Definition \ref{def:graph} with $(g+1)$ faces and let the curve $\Gamma=\Gamma(\mathcal G)$ be as in Construction \ref{def:gamma}.  
Let $({\mathcal N}, \mathcal O(I), \mathfrak l)$ be a network of graph $\mathcal G$ representing $[A]$ such that at any edge  the half-edge wave function does not vanish identically.
Let $\vec t_0$ be such that the half-edge wave function defined in the previous Section in non-zero at all edges.
Let $\DDN$ and $\hat \Psi_{e} (\vec t)$ respectively be the dressed network divisor of Definition \ref{def:vac_div_gen} and the KP edge wave function of Definition \ref{def:norm_vac_gen}. Finally on each component of $\Gamma$ let the coordinate $\zeta$ be as in Definition \ref{def:loccoor} (see also Figure \ref{fig:lcoord}).

We extend the KP wave function $\hat \psi (P, \vec t)$ to $\Gamma\backslash \{P_0\}$ as follows:
\begin{enumerate}
\item The restriction of ${\hat \psi}$ to $\Gamma_{0}$ is the normalized dressed Sato wave function defined in (\ref{eq:SatoDN})
\[
{\hat \psi} (\zeta, \vec t) = \frac{{\mathfrak D}^{(k)} \phi_0 (\zeta, \vec t)}{{\mathfrak D}^{(k)} \phi_0 (\zeta, \vec t_0)};
\]
\item Let $\Sigma_l$ be the component of $\Gamma$ corresponding to the black vertex $V^{\prime}_l$. Since the KP edge wave function takes the same value $\hat \Psi_{e} (\vec t)$ at all edges $e$ at $V^{\prime}_l$, we define ${\hat \psi} (\zeta(P), \vec t)$ as a constant function with respect to the spectral parameter on $\Sigma_l$: ${\hat \psi} (\zeta(P), \vec t) \equiv \hat \Psi_{e} (\vec t)$ where $e$ is one of the edges at $V^{\prime}_l$;
\item\label{it:non_int_dress} Similarly if $\Gamma_{l}$ is the component of $\Gamma$ corresponding to a white vertex $V_l$ with KP e.w.  $\hat \Psi_{e} (\vec t)$ coinciding on all edges $e$ at $V_l$ for all $\vec t$, then, ${\hat \psi} (P, \vec t)$ is assigned the value $\hat \Psi_{e} (\vec t)$ for any $P\in \Gamma_{l}$, where $e$ is one of the  edges at $V_{l}$: ${\hat \psi} (\zeta(P), \vec t) = \hat \Psi_{e} (\vec t)$;
\item\label{it:int_dress} If $\Gamma_{l}$ is the component of $\Gamma$ corresponding to a trivalent white vertex $V_l$ with KP e.w. $\hat \Psi_{e} (\vec t)$  taking distinct values on the edges $e$ at $V_l$ for some $\vec t\not=\vec t_0$, then we define $\hat \psi$ at the marked points as
${\hat \psi} (\zeta(P^{(m)}_{l}), \vec t) = \hat \Psi_{e_m} (\vec t)$, $m\in [3]$, for all $\vec t$,
where $e_m$ are the edges at $V_{l}$. We uniquely extend $\hat \psi$ to a degree one meromorphic function on $\Gamma_{l}$ imposing that it has a simple pole at $P_{\textup{\scriptsize dr}}^{(l)}$ with real coordinate $\zeta (P_{\textup{\scriptsize dr}}^{(l)})=\gamma_{\textup{\scriptsize dr}, V_l}$, with $\gdr$ as in (\ref{eq:dress_pole_def}):
\begin{equation}\label{eq:99}
{\hat \psi} (\zeta(P), \vec t) = \frac{\hat \Psi_{e_3} (\vec t) \zeta - \gdr\hat \Psi_{e_1} (\vec t) }{ \zeta- \gdr}= \frac{ \Psi_{V_l,e_1} (\vec t) (\zeta -1) + \Psi_{V_l,e_2} (\vec t)\zeta}{\left[\Psi_{V_l,e_1} (\vec t_0)+ \Psi_{V_l,e_2} (\vec t_0)\right](\zeta - \gdr)},
\end{equation}
where $\Psi_{V_l,e_s}$ are the half--edge d.w.f. at the outgoing half--edges $(V_l,e_s)$, $s=1,2$, labeled as in Figure \ref{fig:markedpoints}.
\end{enumerate}
\end{construction}
By construction, the KP wave function has $k$ real simple poles on $\Gamma_0$. Therefore the following definition of KP divisor is fully justified.

\begin{definition}\label{def:DKP}{\textbf{The KP divisor on  $\Gamma$.}} The KP divisor $\DKP$ is the sum of the following $g$ simple poles,
\begin{enumerate}
\item the $k$ poles on $\Gamma_0$ coinciding with the Sato divisor at $\vec t= \vec t_0$;
\item the $g-k$ poles $P_{\textup{\scriptsize dr}}^{(l)}\in \Gamma_l$ uniquely identified by the condition that, in the local coordinate induced by the orientation $\mathcal O$, $\zeta(P_{\textup{\scriptsize dr}}^{(l)}) =\gamma_{\textup{\scriptsize dr}, V_l}$,
where $V_l$, $l\in [g-k]$, are the trivalent white vertices. 
\end{enumerate}
\end{definition}

Next we explain how to detect the position of the divisor point associated to a white vertex.
Let $e_m$, $P_m$, $m\in [3]$ respectively denote the edge at $V$ and the corresponding marked point on $\Gamma_V$. Let $\Omega_{ij}$, $i,j\in [3]$ be both the face of ${\mathcal N}$ bounded by the edges $e_i,e_j$ and the corresponding oval in $\Gamma(\mathcal G)$ bounded by the double points $P_i, P_j$ (see also Figure \ref{fig:corr_V_G}).
\begin{figure}
  \centering
  \includegraphics[width=0.56\textwidth]{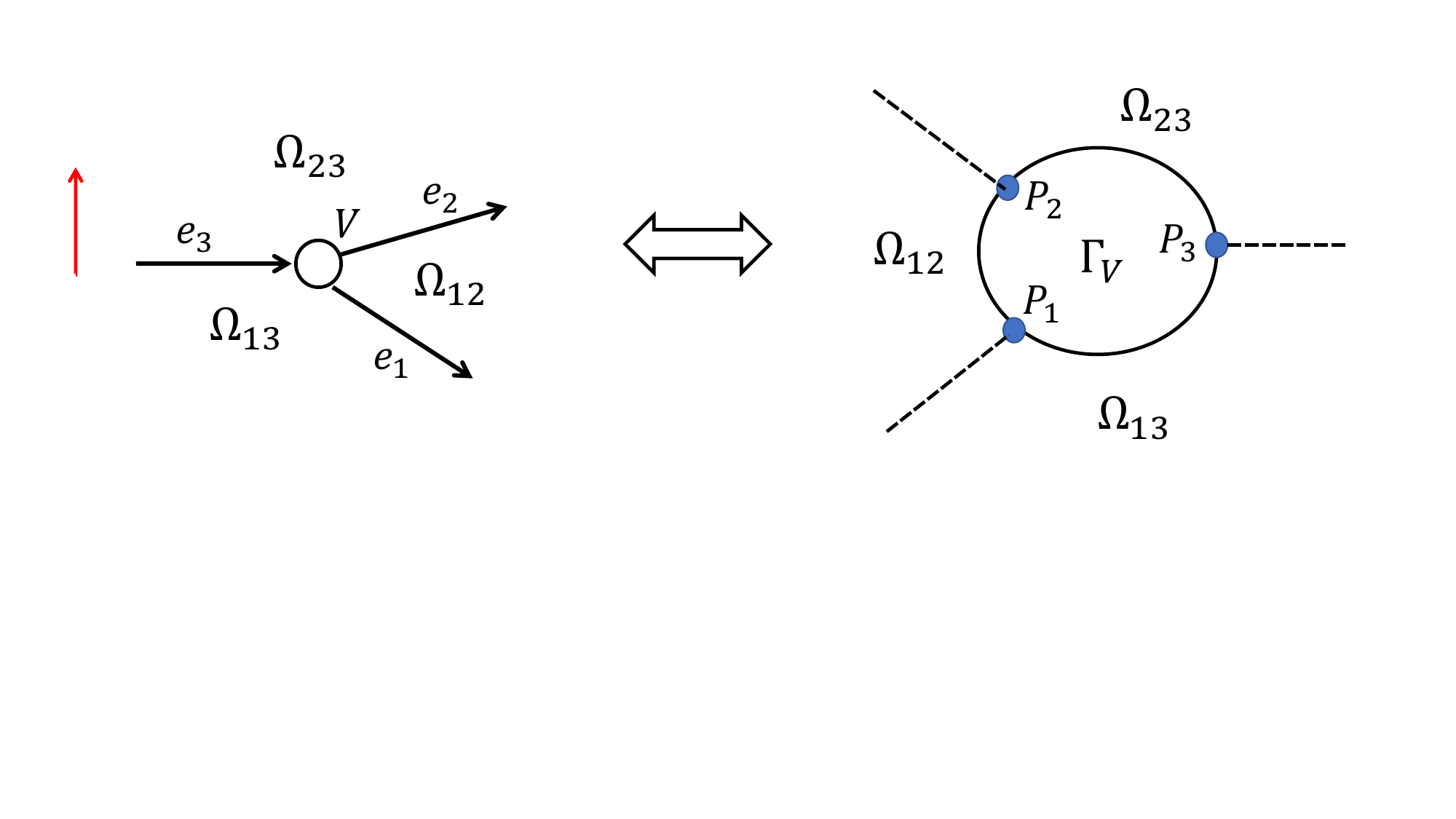}
	\vspace{-2.5 truecm}
  \caption{\small{\sl The correspondence between faces at $V$ (left) and ovals bounded by $\Gamma_V$ (right) under the assumption that the curve is constructed reflecting the graph w.r.t. a vertical ray (all boundary vertices in the original network lay on a horizontal line).}}
	\label{fig:corr_V_G}
\end{figure}
Let $\Psi_{V,m} \equiv \Psi_{V,e_m ; {\mathcal N}, \mathcal O, \mathfrak l} (\vec t_0)$ be the dressed h.e.w. at the 
half-edge $(V,e_m)$, $m\in [3]$, at a given trivalent white vertex $V$ as in Definition \ref{def:vvw_gen}. Then the network divisor number $\gamma_V$ is the local coordinate of the divisor point $P_V$ on the component $\Gamma_V$ corresponding to $V$:
\begin{equation}\label{eq:formula_div}
\zeta(P_V) \; \equiv \; \gamma_V \; = \; \frac{ \Psi_{V,a}}{\Psi_{V,a} +  \Psi_{V,b}} \; =\;  -\frac{\Psi_{V,b}}{\Psi_{V,c}}
\; = \; 1+ \frac{\Psi_{V,b}}{\Psi_{V,c}}.
\end{equation}
Then the following Lemma explains to which oval the divisor point belongs to.

\begin{lemma}\textbf{Position of the divisor point on $\Gamma_V$}\label{lem:pos_div}
In the above notations, the divisor point $P_V$ belongs to the unique oval $\Omega_{ij}$, $i,j \in [m]$, such that the half--edge wave function at $V$ satisfies 
\begin{equation}\label{eq:pos_div}
\Psi_{V,i}\Psi_{V,j} >0.
\end{equation}
\end{lemma}

The proof easily follows using $\Psi_{V,1} + \Psi_{V,2} +\Psi_{V,3}=0$ and the assumption that $\Psi_{V,s} \not =0$, for $s\in [3]$. Then (\ref{eq:pos_div}) holds comparing (\ref{eq:formula_div}) and Figure \ref{fig:corr_V_G} since, in the local coordinate $\zeta$ associated to the graph orientation at $V$, $\zeta(P_1)=0$, $\zeta(P_2)=1$ and $\zeta(P_3)=\infty$. For instance $\zeta (P_V) \equiv \gamma_V\in ]0,1[$ if and only if $\Psi_{V,1} \Psi_{V,2}>0$ and similarly in the other two cases.

\begin{theorem}\textbf{Properties of the KP divisor on $\Gamma$}\label{theo:inv}
\begin{enumerate}
\item $\DKP$ is independent on the gauge ray direction, on the weight gauge, on the vertex gauge and on the orientation of the network used to construct it;
\item $\DKP$ is contained in the union of the ovals of $\Gamma$.
\end{enumerate}
\end{theorem}

\begin{proof} The only untrivial statement is the independence of the divisor on the orientation of the network. To check it is sufficient to compare the transformation laws for the divisor numbers in Proposition~\ref{prop:indep_orient} and the changes of coordinates due to change of orientation.
\end{proof}

By construction, each double point in $\Gamma$ corresponds to an edge; therefore the KP wave function takes the same value at all double points for all times and, at the double points $\kappa_j$ corresponding to edges at boundary vertices $b_j$, $j\in [n]$, it coincides for all times with the normalized Sato wave function.   
$\hat \psi$ is meromorphic of degree $\mathfrak d_{\mbox{\tiny KP}}\le g$, and its poles are all simple and belong to $\DKP$, which is contained in the union of all the real ovals of $\Gamma$. Therefore $\hat \psi$ satisfies the properties in Definition \ref{def:KPwave}, that is it is the KP wave function for the soliton data $(\mathcal K, [A])$ and the divisor $\DKP$ on $\Gamma$.

\begin{theorem}\label{lemma:KPeffvac}\textbf{$\hat \psi$ is the unique KP wave function on $\Gamma$ for $(\mathcal K, [A])$ and the divisor $\DKP$.}
Let $\hat \psi$, ${\mathcal D}_{\textup{\scriptsize dr}, {\mathcal N}^{\prime}}$, $\DKP$ on $\Gamma$ be as in Construction \ref{con:dress_gen} and Definitions \ref{def:vac_div_gen} and \ref{def:DKP}. Then $\hat \psi$ satisfies the following properties of Definition \ref{def:KPwave} on $\Gamma\backslash\{P_0\}$:
\begin{enumerate}
\item At $\vec t=\vec t_0$ $\hat \psi (P, \vec t_0)=1$ at all points $P\in \Gamma\backslash \{P_0\} $;
\item ${\hat \psi} (\zeta(P), \vec t)$ is real for real values of the local coordinate $\zeta$ and for all real $\vec t$ on each component of $\Gamma$;
\item $\hat \psi$ takes the same value at pairs of glued points $P,Q\in \Gamma$, for all $\vec t$:  $\hat \psi(P, \vec t) = \hat \psi(Q, \vec t)$;
\item $\hat \psi(\zeta, \vec t)$ is either constant or meromorphic of degree one w.r.t. to the spectral parameter on each copy of $\mathbb{CP}^1$ corresponding on any trivalent white vertex common to $\mathcal N$ and $\mathcal N^{\prime}$. $\hat \psi(\zeta, \vec t)$ is constant w.r.t. to the spectral parameter on each other copy of $\mathbb{CP}^1$; 
\item $\DKP+(\hat\psi(P,\vec t))\ge 0$ for all $\vec t$.
\end{enumerate}
\end{theorem}

The proof of the assertions is straightforward and is omitted. We remark that, in the special case of the Le--network the divisor $\DKP$ coincides with the one constructed in \cite{AG3}.
In the next Section we complete the proof of Theorem \ref{theo:exist} by showing that
the KP divisor $\DKP$ satisfies the regularity and reality conditions (Items (\ref{item:defodd_KP}) and (\ref{item:defeven_KP}) of Definition \ref{def:real_KP_div}).

\subsection{Main construction -- Part V. Combinatorial characterization of the regularity of $\DKP$}\label{sec:comb}

In this Section we complete the proof of Theorem \ref{theo:exist}: using Theorem \ref{theo:sign_face} 
we prove that there is exactly one divisor point in each finite oval. 
In the following $\Omega$ denotes both the face in the network and the corresponding oval of the curve $\Gamma$.

\smallskip

Let $\nu_{\Omega}$ be the number of divisor points in $\Omega$ associated to the white vertices bounding $\Omega$. By definition $\nu_{\Omega}$ is equal to the number of pairs of half--edges at white vertices bounding $\Omega$ where the half--edge wave function has the same sign.
When the face $\Omega$ intersects the boundary of the disk, the total number of divisor points in $\Omega$ is the
sum of $\nu_{\Omega}$ and of the number the Sato divisor points in $\Omega\cap \Gamma_0$. 

By construction, along $\partial \Omega$ the half--edge wave function keeps the same sign at each pair of edges at a given black vertex and at each pair of edges at a trivalent white vertex associated to a divisor point in $\partial\Omega$, whereas it changes of sign at all other pair of edges at white vertices in $\partial\Omega$. Let $c_{\Omega}$ denote the total number of pair of half edges bounding $\Omega$ where the half-edge wave function changes sign. 

\smallskip

Let us start with the case in which $\Omega$ is an internal oval (face). 
Let $\epsilon_{U,V}$ be the geometric signature of the network $(\mathcal N, \mathcal O, \mathfrak l)$ and $\epsilon(\Omega)$ be the total contribution of the geometric signature at the edges $e=(U,V)$ bounding $\Omega$:
\begin{equation}
\epsilon(\Omega) = \sum_{e\, \in\, \partial\Omega} \epsilon_{e}.
\end{equation}
$n_{w,\Omega}$ denotes the total number of white vertices in $\partial \Omega$. 

In this case $\nu_{\Omega}$, the number of divisor points in $\Omega$, equals $n_{\mbox{\scriptsize{white}}}(\Omega) - c_{\Omega}$. By definition $c_{\Omega}$ has the same parity as $\epsilon(\Omega)$, the sum of the geometric signature over all edges bounding $\Omega$. Then using (\ref{eq:sign_face}) in Theorem \ref{theo:sign_face}, we immediately conclude that in each internal oval there is an odd number of divisor points since $c_{\Omega}$ and $n_{\mbox{\scriptsize{white}}}(\Omega)$ have opposite parities: \begin{equation}\label{eq:inv_edge_sign}
c_{\Omega} \; \equiv \; \epsilon(\Omega) \; = \;n_{\mbox{\scriptsize{white}}}(\Omega) -1   \quad \mod 2.
\end{equation}
We have thus proven the following Lemma:

\begin{lemma}\textbf{The number of divisor points at internal ovals}\label{lem:num_div_int_ovals}
With the above notations, at each internal oval $\Omega$ the number of KP divisor points is odd: $\nu_{\Omega} = 1 \mod 2.$ 
\end{lemma}

Next we count the number of divisor points when the face $\Omega$ intersects the boundary of the disk. 

As a first step we compute the total number of changes of sign along $\partial \Omega$. Let $\rho_{\Omega}$ be the number of pairs of consecutive boundary vertices in $\partial \Omega\cap\Gamma_0$ where the half-edge wave function changes sign.

\begin{lemma}\label{lem:div_bou}\textbf{Counting changes of sign of the half edge--wave function at the ovals intersecting the boundary}
Let $\Omega$ be an oval intersecting the boundary of the disk. Then
\begin{equation}\label{eq:ch_sign_bound}
c_{\Omega} +\rho_{\Omega} +\epsilon(\Omega) \; = \;  0 \quad \mod 2.
\end{equation}
\end{lemma}

The proof is standard.

At each finite oval $\Omega$ intersecting $\Gamma_0$ the number of Sato divisor points in $\Omega\cap \Gamma_0$ has the same parity as $\rho_{\Omega}$, the number of pairs of consecutive boundary vertices in $\partial \Omega$ where the half-edge wave function changes sign. Indeed each portion of $\Gamma_0$ bounding $\Omega_s$ is marked by two consecutive boundary vertices $b_j,b_{j+1}$.
Let $e(b_j)$ be the edge at $b_j$. Then $\Psi_{e(b_j)}(\vec t_0)\Psi_{e(b_{j+1})}(\vec t_0)<0 \,\, (>0)$ implies that there is an odd (even) number of Sato 
divisor points in $\Gamma_0\cap \Omega$ belonging to the interval $]\kappa_j, \kappa_{j+1}[$. 
The total number of divisor points is then the sum of $\rho_{\Omega}$ and of $\nu_{\Omega}$, that is the sum of Sato and non Sato divisor points in $\partial \Omega$. Then using $\nu_{\Omega} =n_{\mbox{\scriptsize{white}}}(\Omega) -c_{\Omega}$, (\ref{eq:sign_face}) and (\ref{eq:ch_sign_bound}), we easily conclude that there is an odd number of divisor points also in each finite oval intersecting $\Gamma_0$,
\begin{equation}\label{eq:div_bound_1}
\nu_{\Omega} +\rho_{\Omega} \; = \;n_{\mbox{\scriptsize{white}}}(\Omega)  +\epsilon(\Omega)   \; =\; 1 \quad \mod 2.
\end{equation}

We have thus proven the following Lemma:

\begin{lemma}\textbf{The number of divisor points at finite ovals intersecting $\Gamma_0$}\label{lem:num_div_bound_ovals}
With the above notations, at each finite oval $\Omega$ having non--empty intersection with $\Gamma_0$ the number of KP divisor points is odd: $\nu_{\Omega} +\rho_{\Omega} = 1 \mod 2.$ 
\end{lemma}

Finally, in the infinite oval $\Omega_0$, we also have $\Psi_{e(b_1)}(\vec t_0)\Psi_{e(b_{n})}(\vec t_0)<0 \,\, (>0)$ respectively when $k$ is odd (even). Therefore, the number of Sato divisor points in $\Omega_0$ has the same parity as $\rho_{\Omega_0}+k$. Then, proceeding as before, we conclude that the number of divisor points in the infinite oval is even
\begin{equation}\label{eq:div_bound_2}
\nu_{\Omega_0} +\rho_{\Omega_0} +k \; = \;n_{\mbox{\scriptsize{white}}}(\Omega_0)   +\epsilon(\Omega_0)  +k  \; =\; 0 \quad \mod 2.
\end{equation}

Since the total number of divisor points is $g$ and the ovals are $g+1$, we have thus proven the following Theorem.

\begin{theorem}\label{theo:comb}\textbf{Number of divisor points in the ovals}
There is exactly one divisor point in each finite oval $\Omega_s$, $s\in [g]$, and no divisor point in the infinite oval $\Omega_0$.
In particular, a finite oval contains a Sato divisor point if and only if $\rho_{\Omega} = 1,\, \mod 2$.
\end{theorem}

\part{Transformation properties of the divisor}

In this Part we discuss the transformations of divisors when we apply elementary transformations to the network. We consider two basic types of transformations. Moves and reductions change the networks, but they do not affect the corresponding KP solutions. In contrast, amalgamations change both the soliton solutions and the divisor.

\section{Effect of moves and reductions on curves and divisors}\label{sec:moves_reduc}

In \cite{Pos} the local transformations of planar bicolored networks in the disk which leave invariant the boundary measurement map are classified. 
There are three moves: 
\begin{enumerate}
\item[(M1)] The square move (see Figure \ref{fig:squaremove});
\item[(M2)] The unicolored edge contraction/uncontraction (see Figure \ref{fig:flipmove});
\item[(M3)] The middle vertex insertion/removal (see Figure \ref{fig:middle});
\end{enumerate}
and three reductions:
\begin{enumerate}
\item[(R1)] The parallel edge reduction (see Figure \ref{fig:parall_red_poles});
\item[(R2)] The dipole reduction;
\item[(R3)] The leaf reduction;
\end{enumerate}
such that two networks in the disk connected by a sequence of such moves and reductions represent the same point in $\GTNN$. Since we assume that both the initial and final graphs satisfy the condition that each edge belongs to some path from boundary to boundary, the dipole and leaf reductions cannot occur in this class of graphs.

In this Section we prove that moves (M1)-(M3) and reduction (R1) induce a well defined change in both the curve, the system of edge vectors and the KP divisor.
In Sections~\ref{sec:example}~and~\ref{sec:ex_Gr24top} we illustrate these transformations on some examples.

In the following, we fix both the orientation and the gauge ray direction of the plabic network. Indeed changes of orientation or of gauge direction produce effects on both the system of vectors, the edge wave function and the KP divisor which are completely under control in view of the results of the previous Sections.

We label vertices, edges, faces in the network corresponding to components, double points, ovals in the curve with the same indices. Let $({\mathcal N}, \mathcal O, \mathfrak l)$ and $({\tilde {\mathcal N}}, {\tilde{\mathcal O}}, \mathfrak l)$ respectively be the initial oriented network and  the oriented network after the move or reduction (R1), where we assume that the orientation ${\tilde {\mathcal O}}$ coincides with $\mathcal O$ at all edges except at those involved in the move or reduction where we use Postnikov rules to assign the orientation. We denote with the same symbol and a tilde any quantity referring to the transformed network of the transformed curve. For instance, $g$ and ${\tilde g}$ respectively denote the genus in the initial and transformed curves. 
To simplify notations, we use the same symbol $\gamma_l$, respectively ${\tilde \gamma}_l$ for the divisor number and the divisor point before and after the transformation.

\textbf{(M1) The square move:} If a network has a square formed by four trivalent vertices
whose colors alternate as one goes around the square, then one can switch the colors of these
four vertices and transform the weights of adjacent faces as shown in Figure \ref{fig:squaremove}[left].
The relation between the edge weights with the orientation in Figure \ref{fig:squaremove} is \cite{Pos}
\begin{equation}\label{eq:alpha}
{\tilde \alpha}_1 = \frac{\alpha_3\alpha_4}{{\tilde\alpha}_2}, 
\quad\quad {\tilde \alpha}_2 = \alpha_2 + \alpha_1\alpha_3\alpha_4, \quad\quad  {\tilde \alpha}_3 = \frac{\alpha_2\alpha_3}{{\tilde \alpha}_2}, \quad\quad \displaystyle {\tilde \alpha}_4 = \frac{\alpha_1\alpha_3}{{\tilde \alpha}_2}.
\end{equation}
The system of equations on the edges outside the square is the same before and after the move and also the boundary
conditions remain unchanged. Indeed, the number of white vertices at each face is the same before and after the move, therefore we may use the same geometric signature on all the edges of the graph. The uniqueness of the solution of the system of relations implies that the half--edge wave function is unaffected by the move outside the square. Then using the definition of divisor numbers, one immediately gets the following statement, where notations are consistent with Figure \ref{fig:squaremove}. The relative positions of the divisor points are easy to
check using $\alpha_2 < {\tilde \alpha}_2$, $\alpha_4 < ({\tilde \alpha}_4)^{-1}$, and they are shown in Table~\ref{table:2}.

\begin{figure}
  \centering{\includegraphics[width=0.4\textwidth]{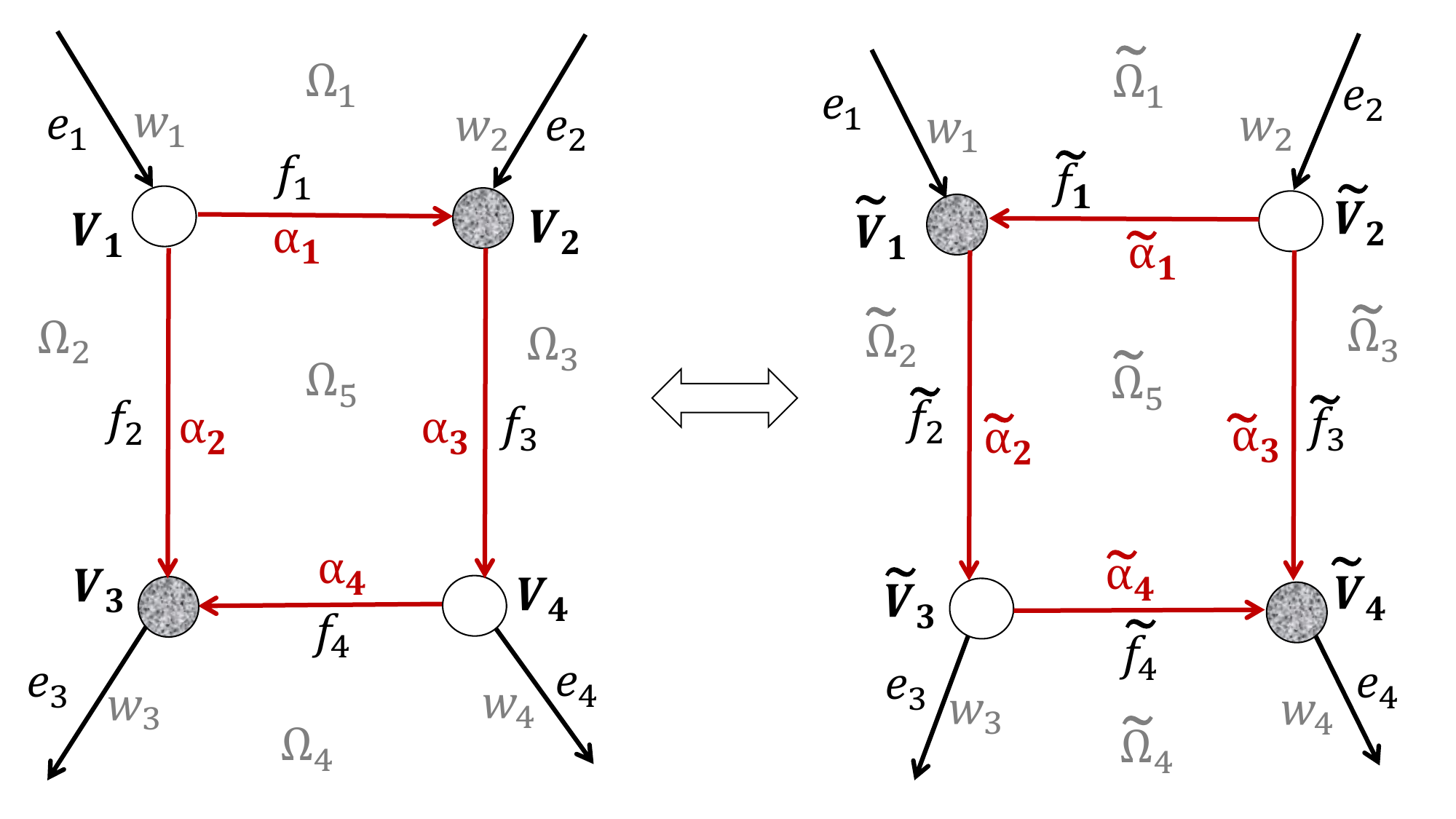}
	\hfill
	\includegraphics[width=0.4\textwidth]{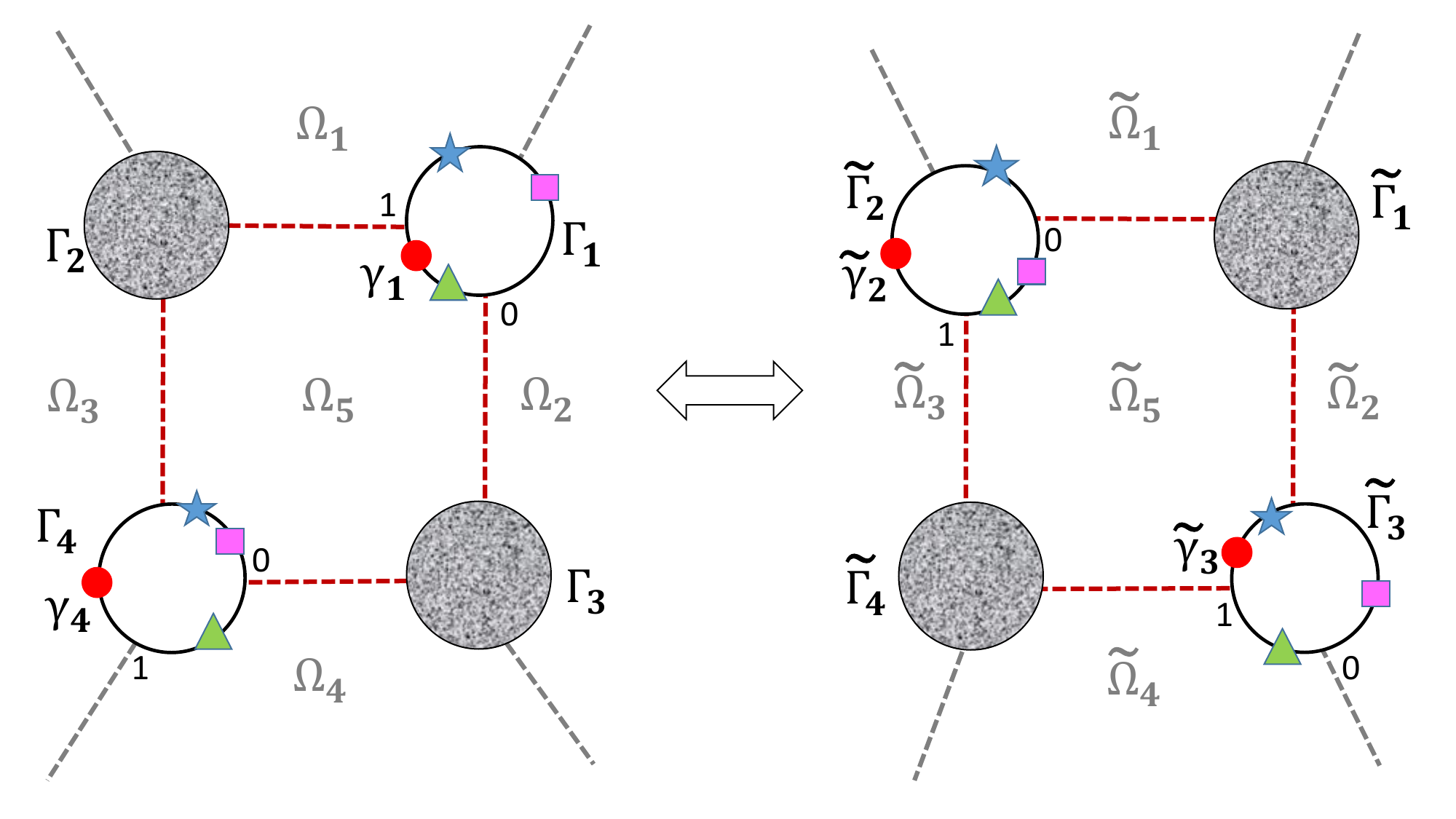}}
  \caption{\small{\sl The effect of the square move [left] on the possible configurations of dressed divisor points [right].}\label{fig:squaremove}}
\end{figure}
\begin{table}
\caption{\label{table:2}The effect of the square move on the dressed divisor} 
\centering
\begin{tabular}{|c|c|c|c|}
\hline\hline
Position of poles in $\Gamma$ & Position of poles in $\tilde\Gamma$ & Symbol for divisor point & Range of parameter\\[0.5ex]
\hline
$\gamma_1 \in \Omega_5, \;  \gamma_4 \in \Omega_4$    
& ${\tilde \gamma}_2 \in {\tilde \Omega}_5, \;  {\tilde \gamma}_3 \in {\tilde \Omega}_4$ & $\triangle$ & $\psi_0 > 0$\\
$\gamma_1 \in \Omega_5, \;  \gamma_4 \in \Omega_3$    
& ${\tilde \gamma}_2 \in {\tilde \Omega}_3, \;  {\tilde \gamma}_3 \in {\tilde \Omega}_5$  
& $\bigcircle$   & $-\alpha_4 < \psi_0 < 0$\\
$\gamma_1 \in \Omega_1, \;  \gamma_4 \in \Omega_5$    
& ${\tilde \gamma}_2 \in {\tilde \Omega}_1, \;  {\tilde \gamma}_3 \in {\tilde \Omega}_5$ 
& $\largestar$  & $- ({\tilde \alpha}_4)^{-1} < \psi_0 < -\alpha_4$\\
$\gamma_1 \in \Omega_2, \;  \gamma_4 \in \Omega_5$    
& ${\tilde \gamma}_2 \in {\tilde \Omega}_5, \;  {\tilde \gamma}_3 \in {\tilde \Omega}_2$   
& $\square$  & $\psi_0 < - ({\tilde \alpha}_4)^{-1}$\\[1ex]
\hline
\end{tabular}
\label{table:SM}
\end{table}

\begin{lemma}\textbf{The effect of the square move on the position of the divisor}
Let the local coordinates on $\Gamma_i$, ${\tilde \Gamma_i}$, $i=1,2$, be as in Figure \ref{fig:squaremove} and let
$\psi_0 = (-1)^{\epsilon_{V_4,V_3}} \frac{\Psi_{V_4,e_4} (\vec t_0)}{\Psi_{V_3,e_3} (\vec t_0)}$, where $\Psi_{V_j,e_j} (\vec t_0)$ is the value of the dressed half--edge wave function at the half edges $(V_j,e_j)$, $j=3,4$. Then 
\[
\gamma_1 = \frac{\alpha_2\tilde\alpha_2^{-1}}{1+{\tilde \alpha}_4 \psi_0}, \quad\quad \gamma_4 = \frac{\alpha_4}{\alpha_4 + \psi_0},\quad\quad {\tilde \gamma}_2 = \frac{\alpha_4 (1 + {\tilde \alpha}_4 \psi_0)}{\alpha_4 + \psi_0},\quad\quad {\tilde \gamma}_3 = \frac{1}{1+{\tilde \alpha}_4 \psi_0},
\]
and the position of the divisor points in the ovals depends on $\psi_0$ as shown in Table \ref{table:SM}. In particular,
there is exactly one dressed divisor point in $(\Gamma_1 \cup \Gamma_2)\cap\Omega_5$, $({\tilde \Gamma}_1 \cup {\tilde \Gamma}_2)\cap{\tilde \Omega}_5$.
\end{lemma}

The square move leaves the number of ovals invariant, eliminates the divisor points $\gamma_1,\gamma_2$ and creates the divisor points ${\tilde \gamma}_1$, ${\tilde \gamma}_2$. We summarize such properties in the following Lemma.

\begin{lemma}\label{lemma:poles_move1}\textbf{The effect of the square move (M1) on the curve and the divisor} 
Let ${\tilde {\mathcal N}}$ be obtained from ${\mathcal N}$ via move (M1). Let ${\mathcal D} ={\mathcal D}({\mathcal N})$,  ${\tilde {\mathcal D}} ={\mathcal D}({\tilde {\mathcal N}})$ respectively  be the dressed network divisor before and after the square move.
Then
\begin{enumerate}
\item ${\tilde g}=g$, and the number of ovals is invariant;
\item The number of dressed divisor points is invariant in every oval: ${\tilde \nu}_{l} =\nu_{l}$, $l\in [0, g]$;
\item ${\tilde {\mathcal D}} = \left( {\mathcal D} \backslash \{ \gamma_1, \gamma_2 \} \right) \cup \{ {\tilde \gamma}_1, {\tilde \gamma}_2 \}$, where
$\gamma_l$ (respectively $\tilde\gamma_l$), $l=1,2$, is the divisor point on $\Gamma_l$ (respectively  $\tilde\Gamma_l$), the component of $\mathbb{CP}^1$ associated to the white vertex $V_l$ (respectively $\tilde V_l$) involved in the square move transforming ${\mathcal N}$ into $\tilde{\mathcal N}$;
\item Either $\gamma_l$, $\tilde\gamma_l$, $l=1,2$, are all untrivial divisor points or all trivial divisor points;
\item If the divisor is generic (no divisor points coincide with the double points) in the initial configuration, then it remains generic after the move. If in the initial configuration a divisor point is located at a double point of the square, at least another divisor point outside the square is also located at the corresponding double point; in this case after the move there is at least one divisor point at a double point of the square and the number of collapsed divisor points changes by $\pm 1$.
\end{enumerate}
\end{lemma}

\begin{figure}
  \centering{\includegraphics[width=0.5\textwidth]{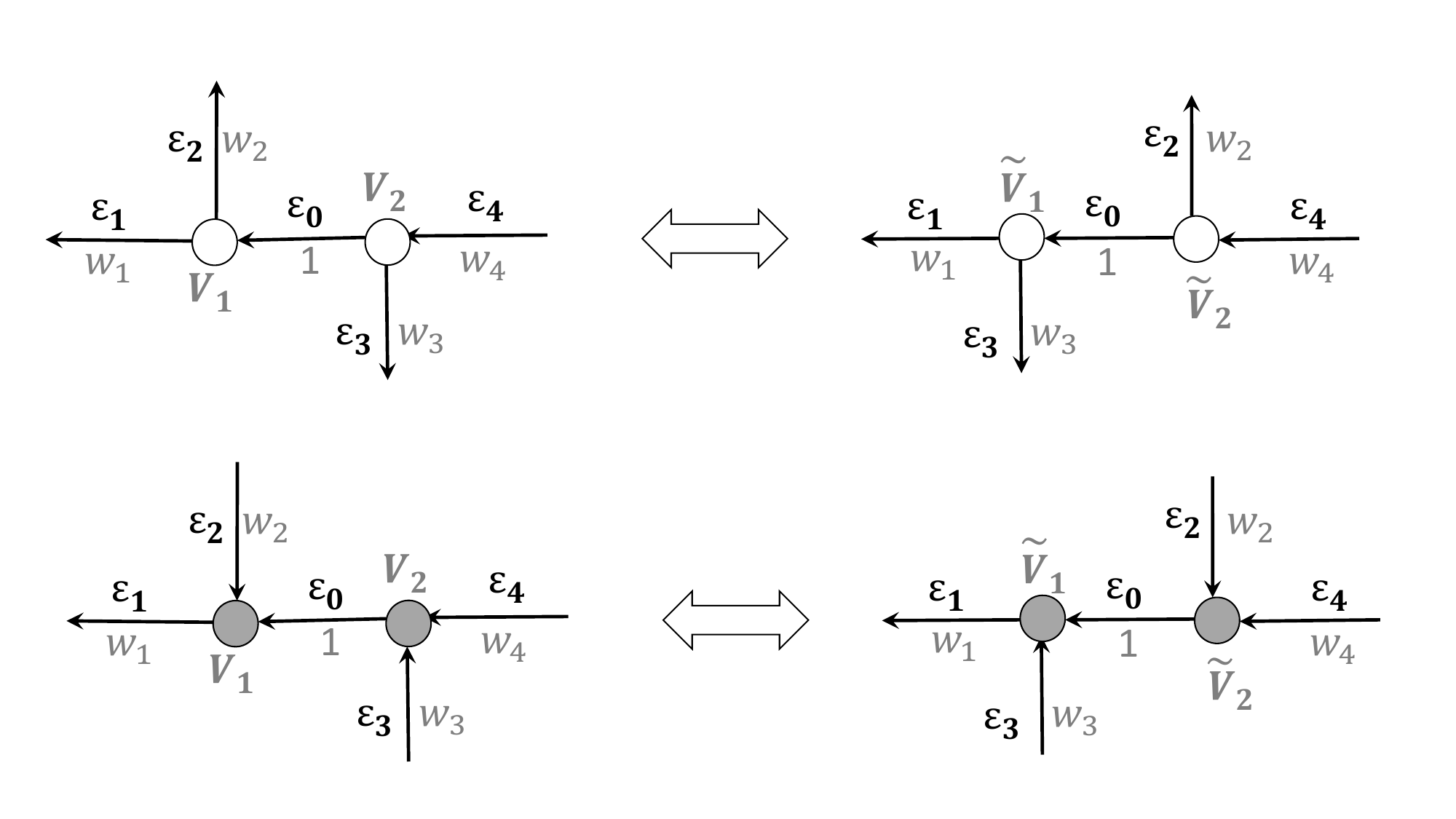}}
  \caption{\small{\sl The insertion/removal of an unicolored internal vertex is equivalent to a flip move of the unicolored vertices.}\label{fig:flipmove}}
\end{figure}

\textbf{(M2) The unicolored edge contraction/uncontraction}
The unicolored edge contraction/un\-con\-traction consists in the elimination/addition of an internal vertex of equal color and of an unit edge, and it leaves invariant the face weights and the boundary measurement map. Such move consists in a flip of unicolored vertices in the case of 
trivalent graphs (see Figure \ref{fig:flipmove}).
A generic contraction/uncontraction of unicolored internal edges can be expressed as a combination of elementary flip moves each involving a pair of consecutive unicolored vertices. 
A flip move at black vertices leaves the divisor invariant. Indeed, in such case we have two equivalent systems of relations using the same signature before and after the move in a gauge for which $\epsilon_0$, the signature at the edge $e_0=(V_1,V_2)$, is zero.
For the flip move at white vertices  it is convenient to use the same signature on both graphs in a gauge for which $\epsilon_0$, the signature at the edge $e_0=(V_1,V_2)$, is one. Indeed, in this case, the total signature of each oval involved in the flip move, changes its parity.
In the following lemma we label $\Omega_l$, $l\in [4]$, the ovals involved in the flip move as in Figure \ref{fig:flip_move_poles}. 

\begin{lemma}\label{lemma:poles_move2}\textbf{The effect of the flip move (M2) at a pair of white vertices on the divisor}
Let ${\tilde {\mathcal N}}$ be obtained from $\mathcal N$ via a flip move move (M2) at a pair of trivalent white vertices. 
Let ${\mathcal D}= {\mathcal D}({\mathcal N})$, ${\tilde {\mathcal D}}= {\mathcal D}({\tilde {\mathcal N}})$, respectively be the dressed network divisor before and after such move. Then
\begin{enumerate}
\item ${\tilde g}=g$ and the number of ovals is invariant;
\item The number of divisor points is invariant in every oval except possibly at the ovals involved in the move.
In the ovals $\Omega_l$, $l\in [4]$ the parity of the number of divisor points before and after the move is invariant:
 ${\tilde \nu}_{l} -\nu_{l} = 0 \,\,(\!\!\!\!\mod 2)$, $l\in [4]$;
\item ${\tilde {\mathcal D}} = \left( {\mathcal D}\backslash \{ \gamma_1, \gamma_2 \} \right) \cup \{ {\tilde \gamma}_1, {\tilde \gamma}_2 \}$, where we use the same notations as in Figure \ref{fig:flip_move_poles};
\item The flip move may transform configurations with divisor points in generic position into configurations with divisor points at double point and vice versa.
\item The flip move may transform configurations without trivial divisor points into configurations with trivial divisor points and vice versa.
\end{enumerate} 
\end{lemma}

\begin{figure}
 \centering
	\includegraphics[width=0.45\textwidth]{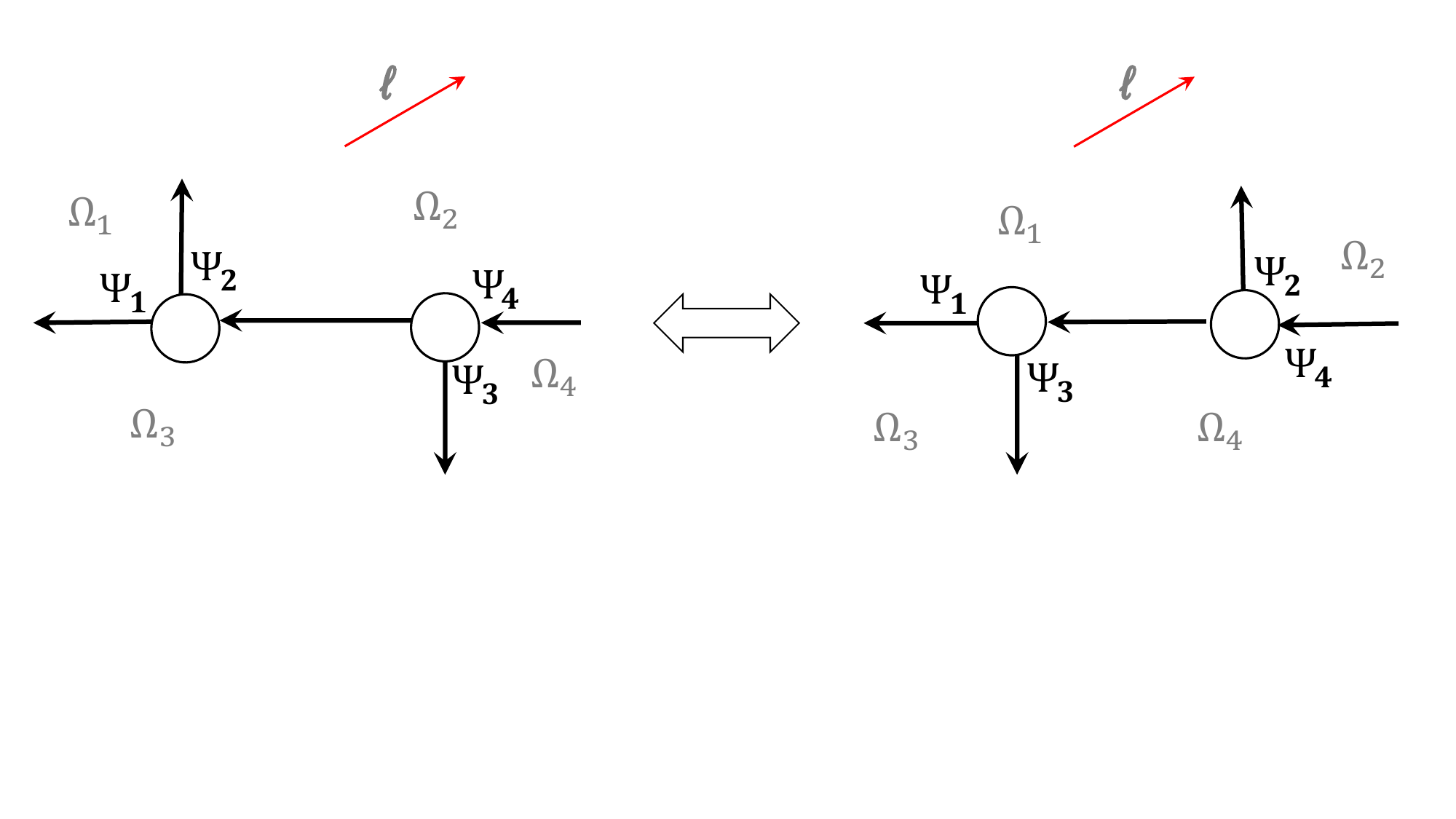}
	\hfill
	\includegraphics[width=0.47\textwidth]{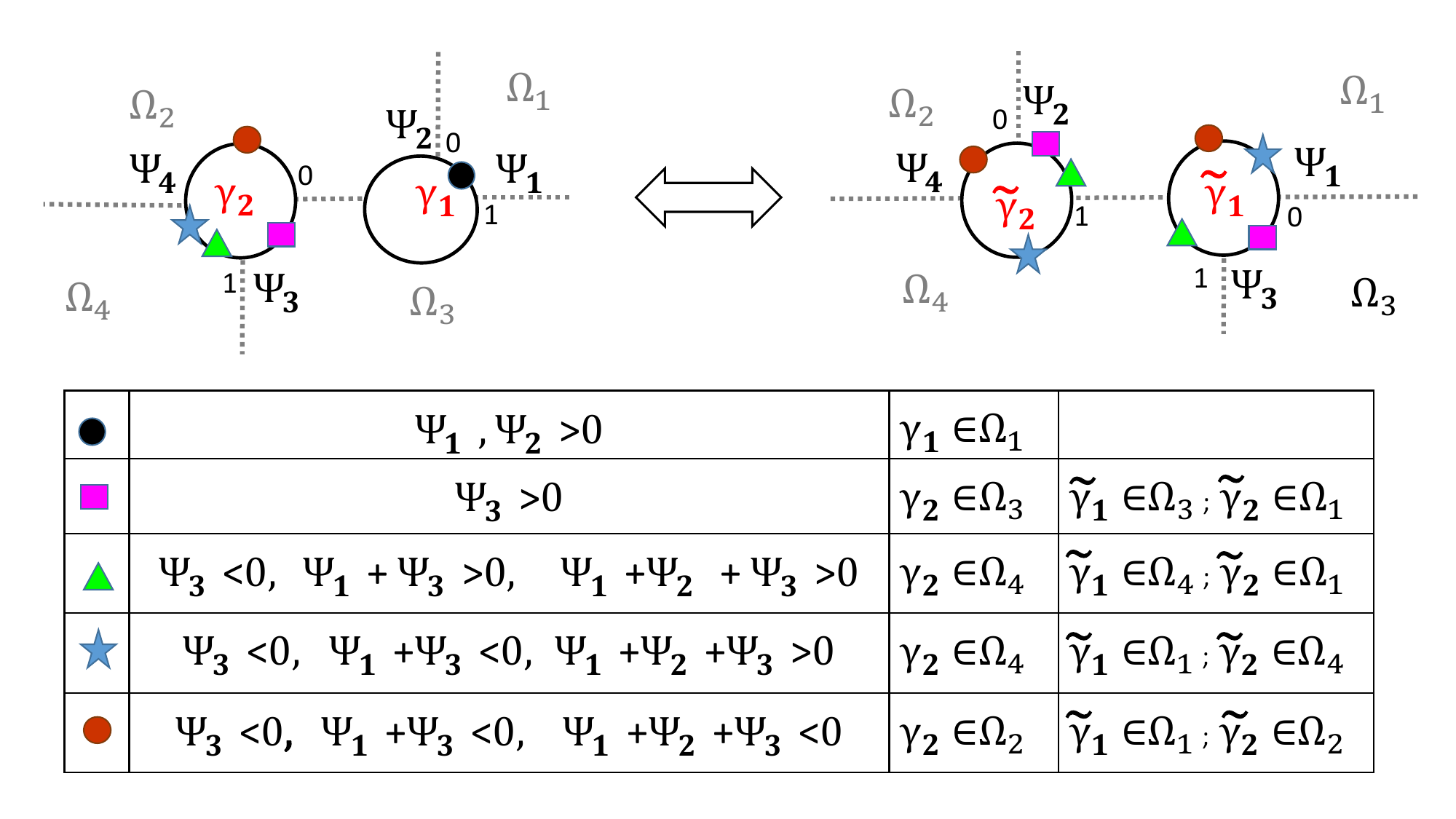}
	\includegraphics[width=0.47\textwidth]{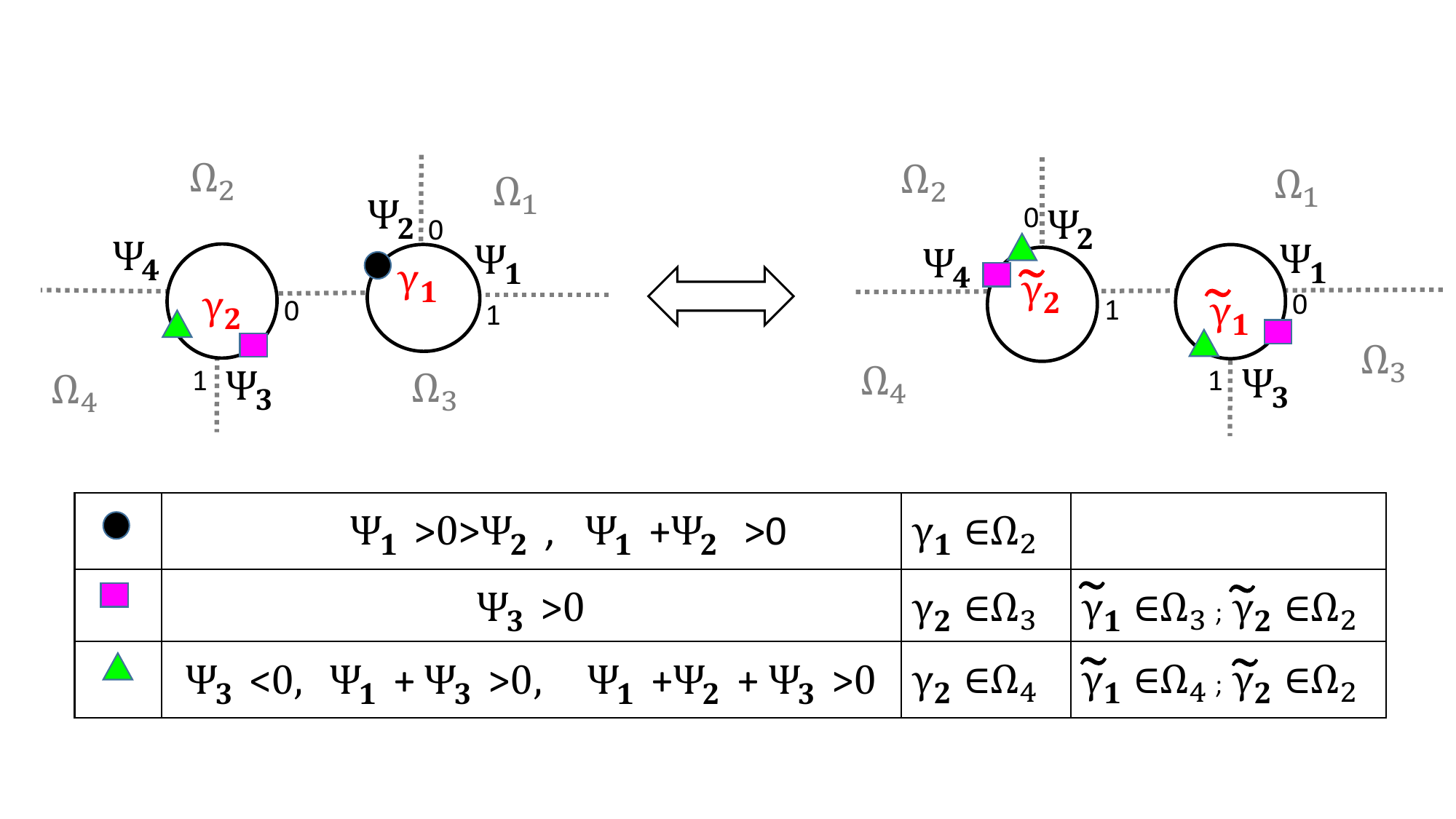}
	\hfill
	\includegraphics[width=0.47\textwidth]{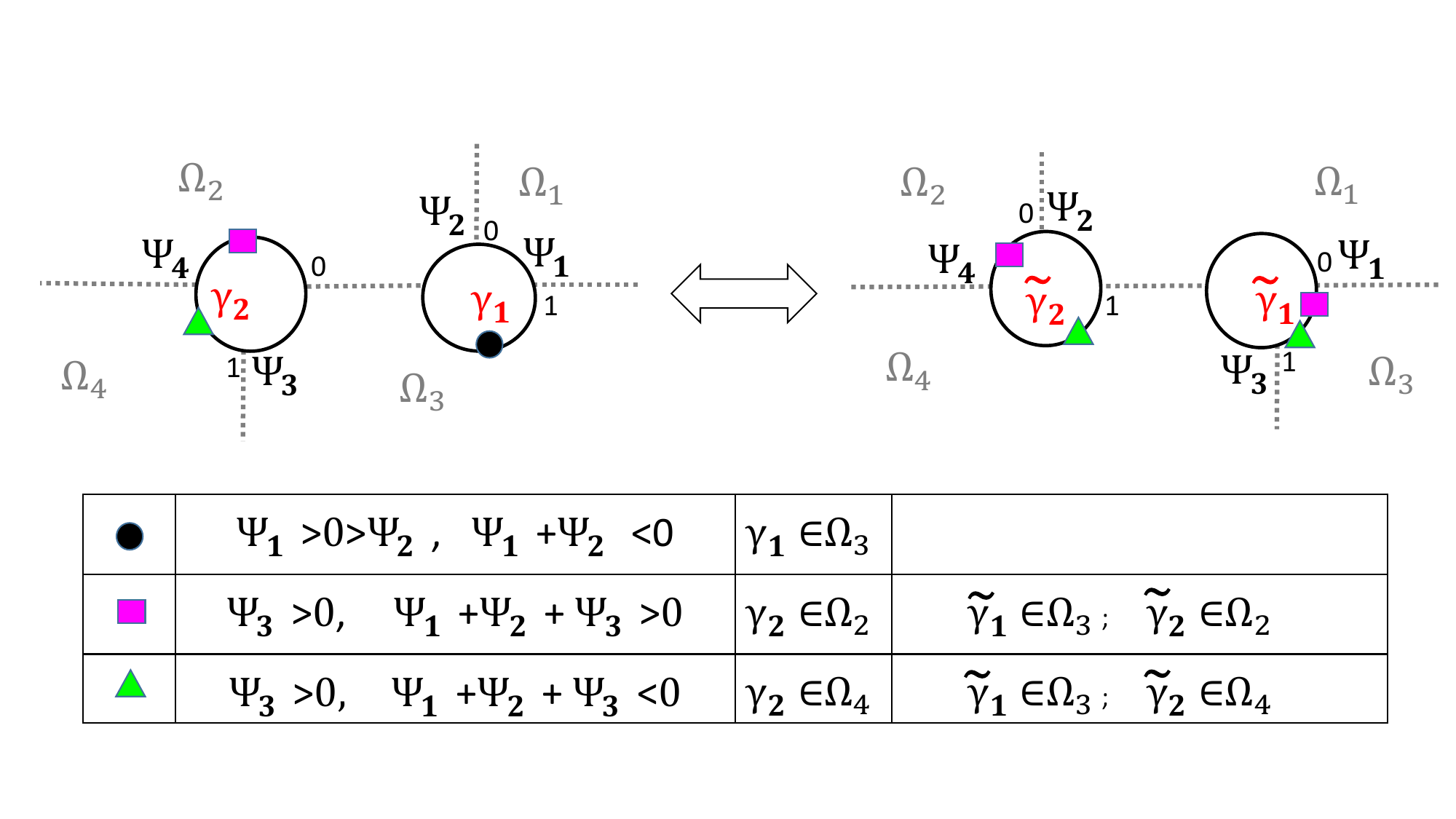}
  \caption{\small{\sl The effect of the flip move [top-left] on the possible configurations of dressed divisor points [top-right],[bottom-left],[bottom-right].}\label{fig:flip_move_poles}}
\end{figure}

The proof is omitted. 

\begin{corollary}\textbf{The effect of the flip move on the divisor}
Let the geometric signature be the same on each edge of the graph before and after the flip move, and with a choice of gauge such that $\epsilon_{e_0}=1$ at $e_0$.
Let the local coordinates on $\Gamma_i$, ${\tilde \Gamma_i}$, $i=1,2$, be as in Figure \ref{fig:flip_move_poles} and let
$\Psi_j = \Psi_{V_1,e_j}(\vec t_0)$, $j=1,2$, $\Psi_j = \Psi_{V_2,e_j}(\vec t_0)$, $j=3,4$, be the value of the half--edge dressed wave function in the initial configuration and ${\tilde \Psi}_j = \Psi_{{\tilde V}_1,e_j}(\vec t_0)$, $j=1,3$, ${\tilde \Psi}_j = \Psi_{{\tilde V}_2,e_j}(\vec t_0)$, $j=2,4$, be the value of the half--edge dressed wave function after the flip move. Then the value of the half--edge dressed wave function at $({\tilde V}_i,{\tilde e}_j)$ is unchanged,  {\sl i.e}
\[
{\tilde \Psi}_{i} = \Psi_i, \quad i\in [4],
\]
and the local coordinates of the divisor points involved in flip move are
\[
\gamma_1 = \frac{\Psi_2}{\Psi_1+\Psi_2},\quad\quad \gamma_2 = \frac{\Psi_1+\Psi_2}{\Psi_1+\Psi_2+\Psi_3},\quad\quad{\tilde \gamma}_1 = \frac{\Psi_1}{\Psi_1+\Psi_3},\quad\quad {\tilde\gamma}_2 = \frac{\Psi_2}{\Psi_1+\Psi_2+\Psi_3}.
\]
\end{corollary}

The proof is straightforward using the definition of divisor coordinates.
In Figure \ref{fig:flip_move_poles} we show the position of the divisor points before and after the flip move [top-left], in function of the relative signs of the value of the dressed wave function at the double points of $\Gamma$. Faces $f_l$ correspond to ovals $\Omega_l$, the orientation of the edges in the graph at each vertex induces the local coordinates at each copy of $\mathbb{CP}^1$ in $\Gamma$.
We remark that not all combinations of signs are realizable at the finite ovals for real regular divisors.
For instance the following choice of signs of the half--edge wave function 
$\Psi_2(\vec t_0), \Psi_3(\vec t_0)<0< \Psi_1 (\vec t_0)$, $\Psi_1(\vec t_0)+ \Psi_3(\vec t_0)<0<\Psi_1(\vec t_0)+ \Psi_2(\vec t_0)$,
would imply the divisor configurations $\gamma_1,\gamma_2 \in \Omega_2$ and ${\tilde \gamma}_1, {\tilde\gamma}_2 \in \Omega_1$, which are not allowed for real regular divisors, since every finite oval may contain only one divisor point.

Finally, if the half edge vectors $F_1,F_3$ are linearly dependent, {\sl i.e.} $F_3= c_{13}F_1$ for some $c_{13}\not =0$, then also $\Psi_3(\vec t)= c_{13} \Psi_1(\vec t)$, for all $\vec t$. In such case $\gamma_l$, $l=1,2$, and ${\tilde \gamma}_2$ are untrivial divisor points, whereas ${\tilde \gamma}_1$ is a trivial divisor point, {\sl i.e.} the normalized wave function is constant on the corresponding copy of $\mathbb{CP}^1$.

\begin{figure}
  \centering{\includegraphics[width=0.55\textwidth]{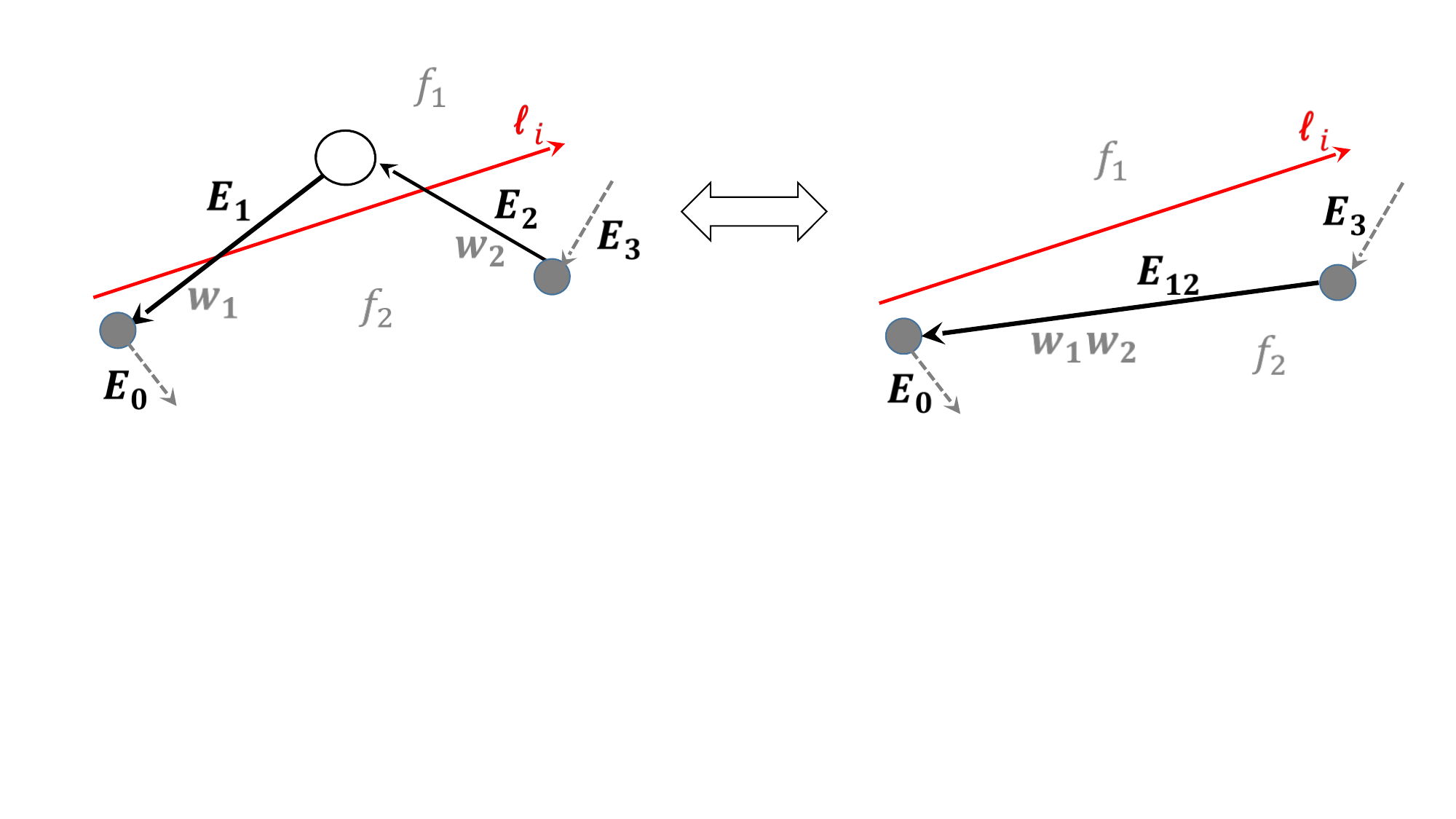}}
	\vspace{-2.5 truecm}
  \caption{\small{\sl The middle edge insertion/removal.}\label{fig:middle}}
\end{figure}

\textbf{(M3) The middle edge insertion/removal}
The middle edge insertion/removal concerns bivalent vertices (see Figure \ref{fig:middle}) without changing the face configuration. If we remove a white vertex, then we keep the signature invariant on all edges, and assign $\epsilon_{12}=\epsilon_{1}+\epsilon_{2}+1 \mod 2$. If we remove a black vertex, then we keep the signature invariant on all edges, and assign $\epsilon_{12}=\epsilon_{1}+\epsilon_{2}\mod 2$. This move does not affect neither the number of ovals nor the divisor configuration.

\begin{figure}
  \centering{\includegraphics[width=0.55\textwidth]{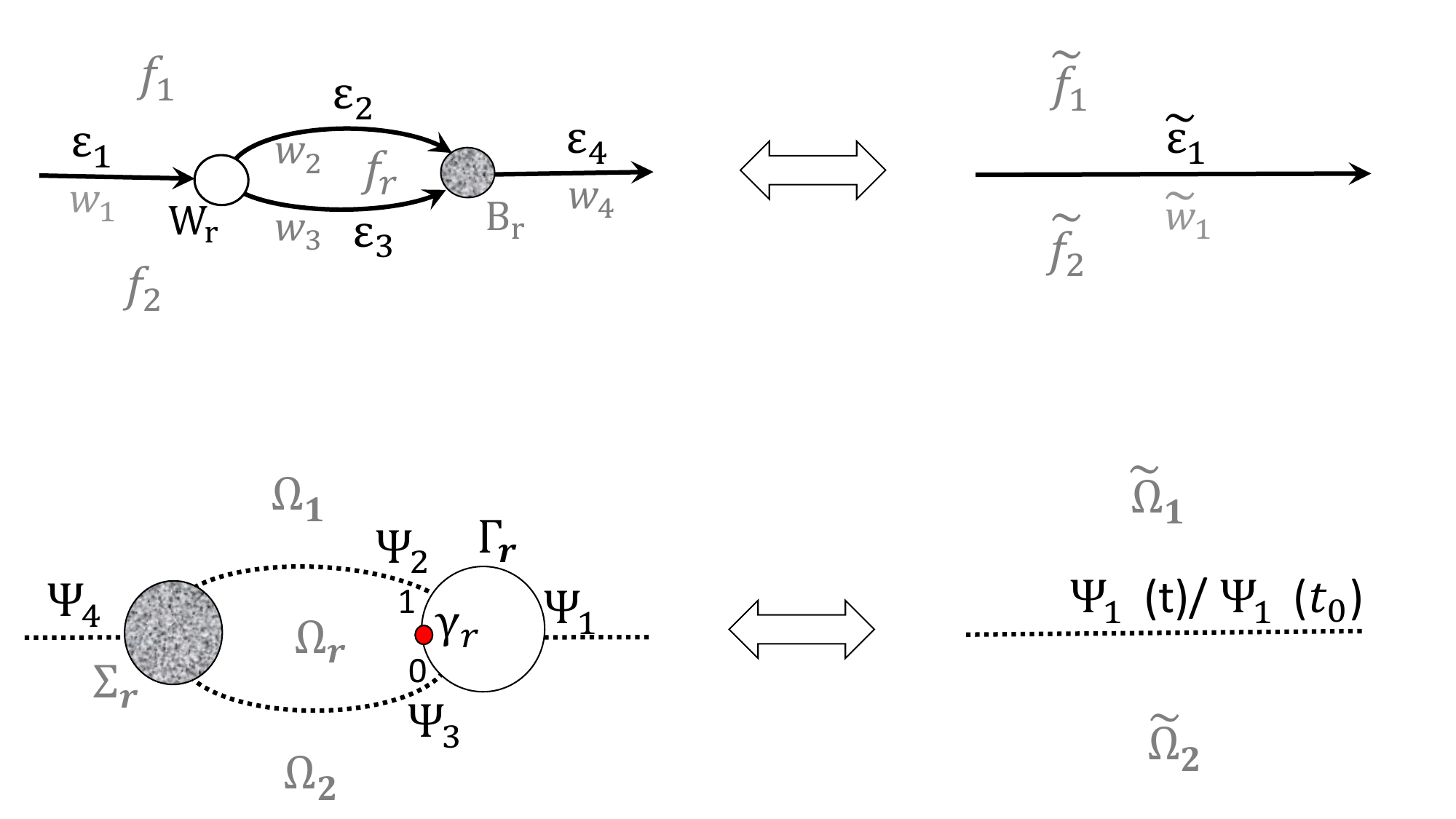}}
  \caption{\small{\sl The parallel edge reduction [top] eliminates an oval, diminishes by one the genus and eliminates a divisor point $\gamma_{\mbox{r}}$ [bottom].}\label{fig:parall_red_poles}}
\end{figure}
\textbf{(R1) The parallel edge reduction}
The parallel edge reduction consists in the removal of two trivalent vertices of different color connected by a pair of parallel edges (Figure \ref{fig:parall_red_poles}[top]), therefore it eliminates an oval and a trivial divisor point on the corresponding curve:

\begin{lemma}\label{lemma:poles_red1}\textbf{The effect of  (R1) on the divisor}
Let ${\tilde {\mathcal N}}$ be obtained from ${\mathcal N}$ via the parallel edge reduction (R1), and denote
${\tilde \Gamma}$ and $\Gamma$ the curve after and before such reduction. Then
\begin{enumerate}
\item The genus ${\tilde g}$ of ${\tilde \Gamma}$ is one less than that in $\Gamma$: ${\tilde g}=g-1$ ;
\item The oval $\Omega_{\mbox{r}}$ corresponding to the face $f_{\mbox{r}}$ and the components $\Gamma_{\mbox{r}},\Sigma_{\mbox{r}}$ corresponding to the white and black vertices $W_{\mbox{r}},B_{\mbox{r}}$, are removed by effect of the parallel edge reduction;
\item The trivial divisor point $\gamma_{\mbox{r}}\in \Gamma_{\mbox{r}}\cap \Omega_{\mbox{r}}$ is removed;
\item All other divisor points are not effected by the reduction.
\end{enumerate} 
\end{lemma}
In Figure \ref{fig:parall_red_poles}[bottom] we show the effect of the parallel edge reduction on the curve. We use a gauge such that the signatures $\epsilon_2=\epsilon_3=0$, so that the half-edge wave functions satisfy: $\Psi_{2} (\vec t)= w_2 \Psi_4(\vec t)$, $\Psi_{3} (\vec t) = w_3 \Psi_4(\vec t)$, $\Psi_{1} (\vec t)= -(\Psi_2 (\vec t)+\Psi_3(\vec t))$, for all $\vec t$.

The divisor point in $\Gamma^{(1)}_{\mbox{red}}$ is trivial since it is independent on time, and in the coordinates induced by the orientation on the Figure~\ref{fig:parall_red_poles}, it takes the value:
\[
\gamma_{\mbox{r}} = \frac{\Psi_{3} (\vec t_0)}{\Psi_{2} (\vec t_0)+\Psi_{3} (\vec t_0)} = \frac{w_{3} }{w_{2}+w_{3} }.
\]

On both $\Gamma^{(1)}_{\mbox{red}}$ and $\Sigma^{(1)}_{\mbox{red}}$ the normalized wave function is independent of the spectral parameter and takes the value $\frac{\Psi_{1} (\vec t)}{\Psi_{1} (\vec t_0)}\equiv \frac{\Psi_{4} (\vec t)}{\Psi_{4} (\vec t_0)}$. Therefore, such value is also the one of the normalized wave function at the double point in ${\tilde \Gamma}$ corresponding to the edge $\tilde e_1$ created by the reduction. 

\section{Example: plane curves and divisors for soliton data in ${\mathcal S}_{34}^{\mbox{\tiny TNN}}\subset Gr^{\mbox{\tiny TNN}}(2,4)$}\label{sec:example}

\begin{figure}
  \centering{
\includegraphics[width=0.4\textwidth]{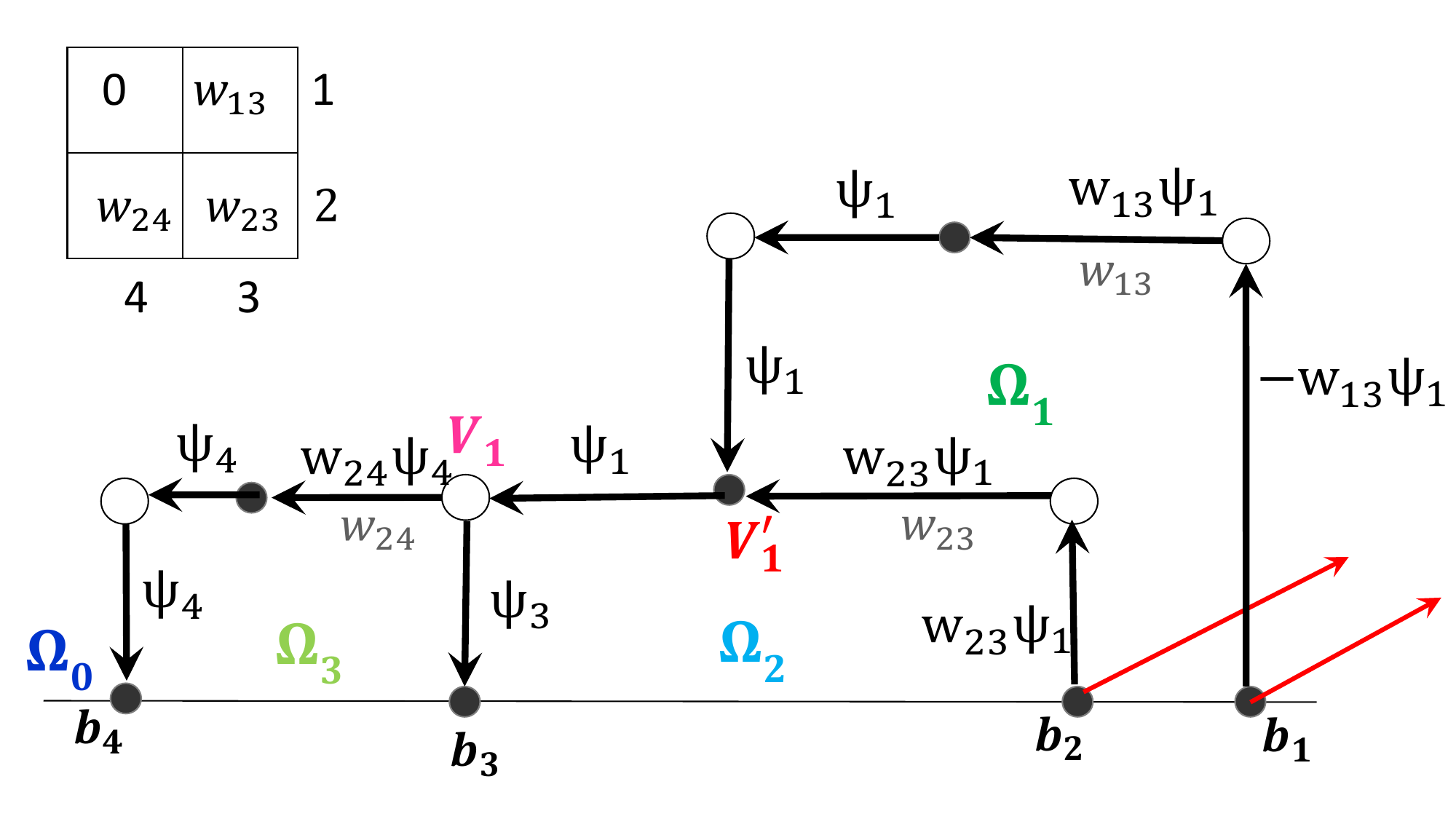}
   \hfill 
  \includegraphics[width=0.4\textwidth]{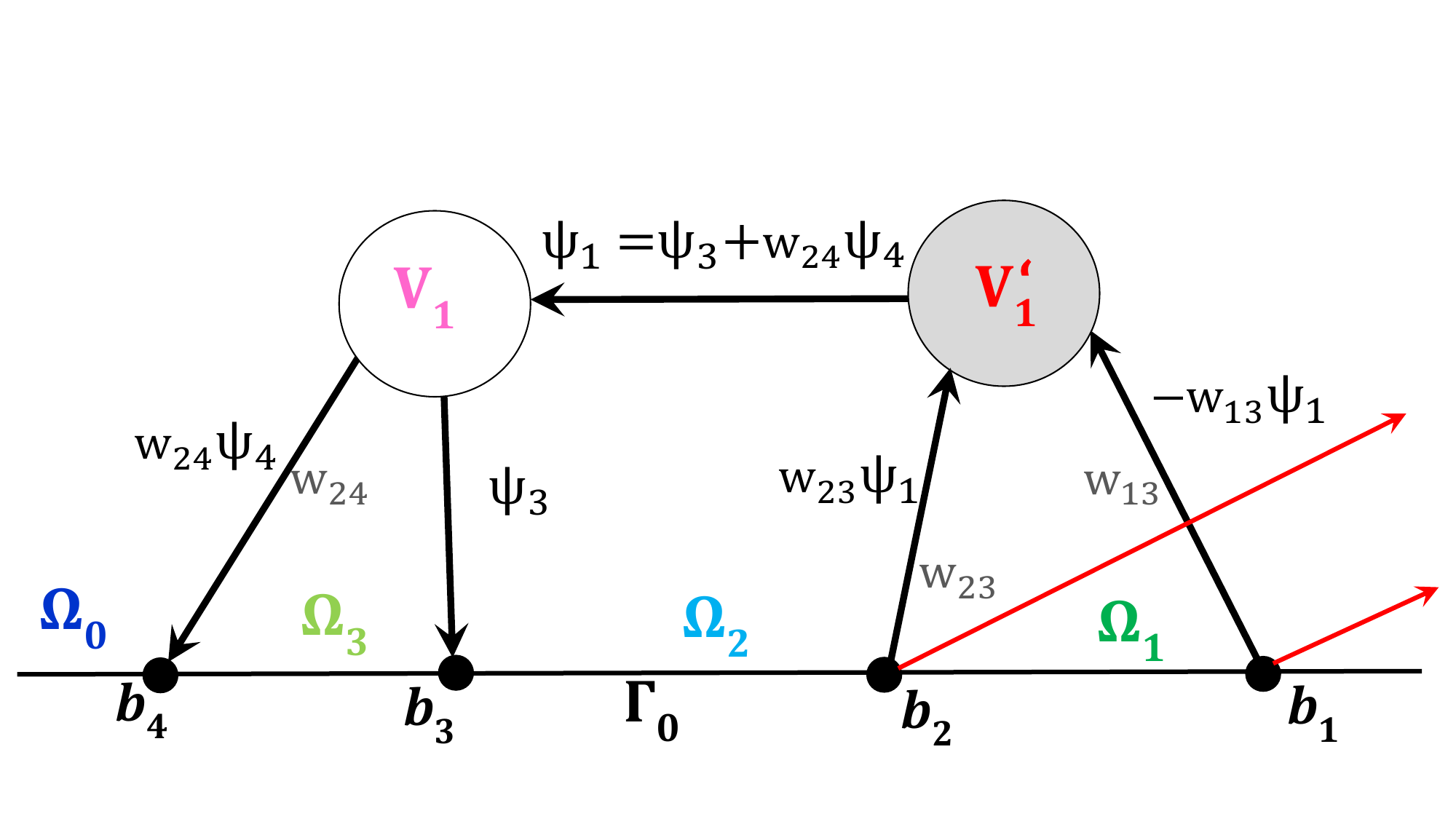}
\\
\includegraphics[width=0.4\textwidth]{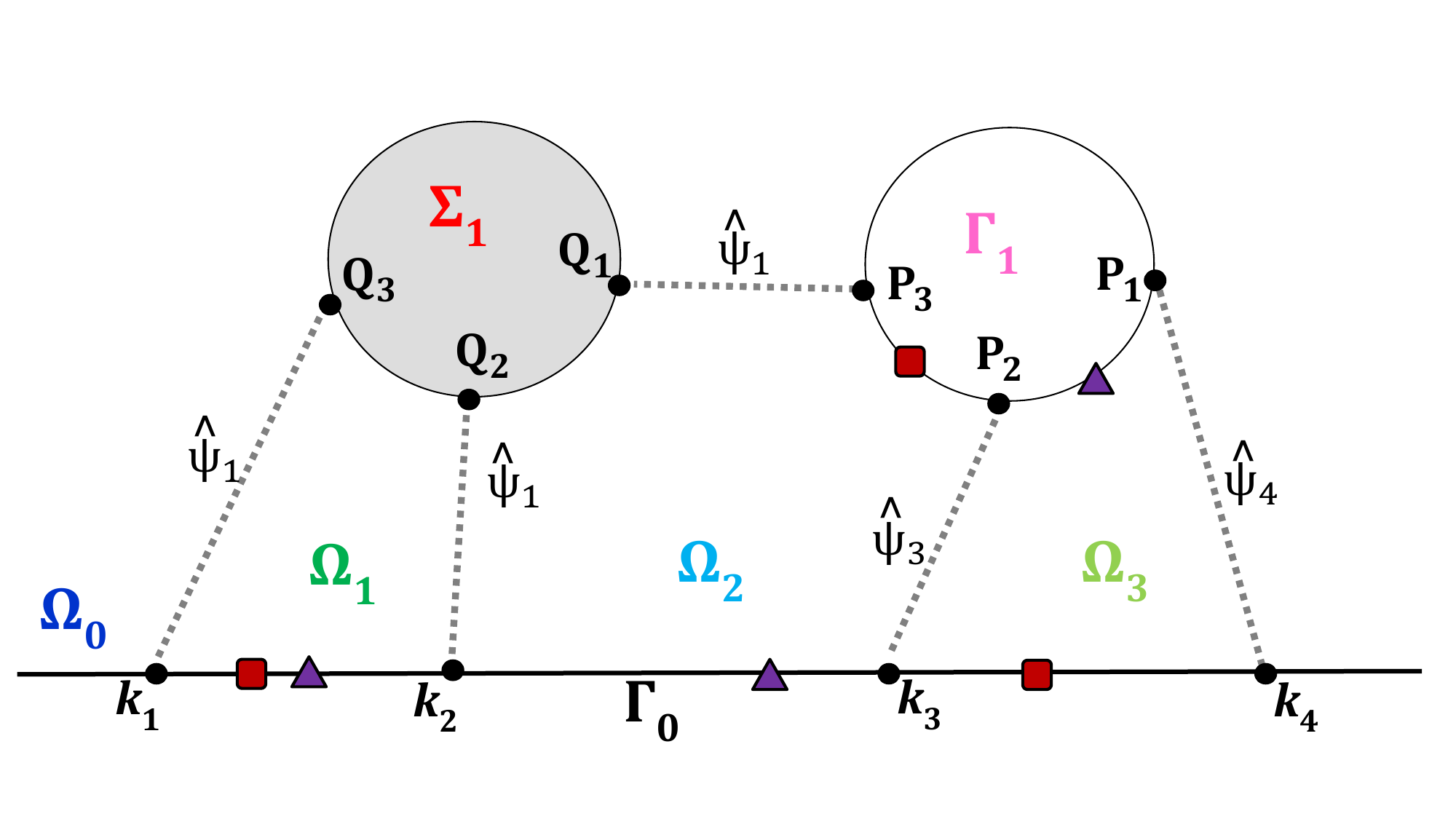}
\hfill
\includegraphics[,width=0.4\textwidth]{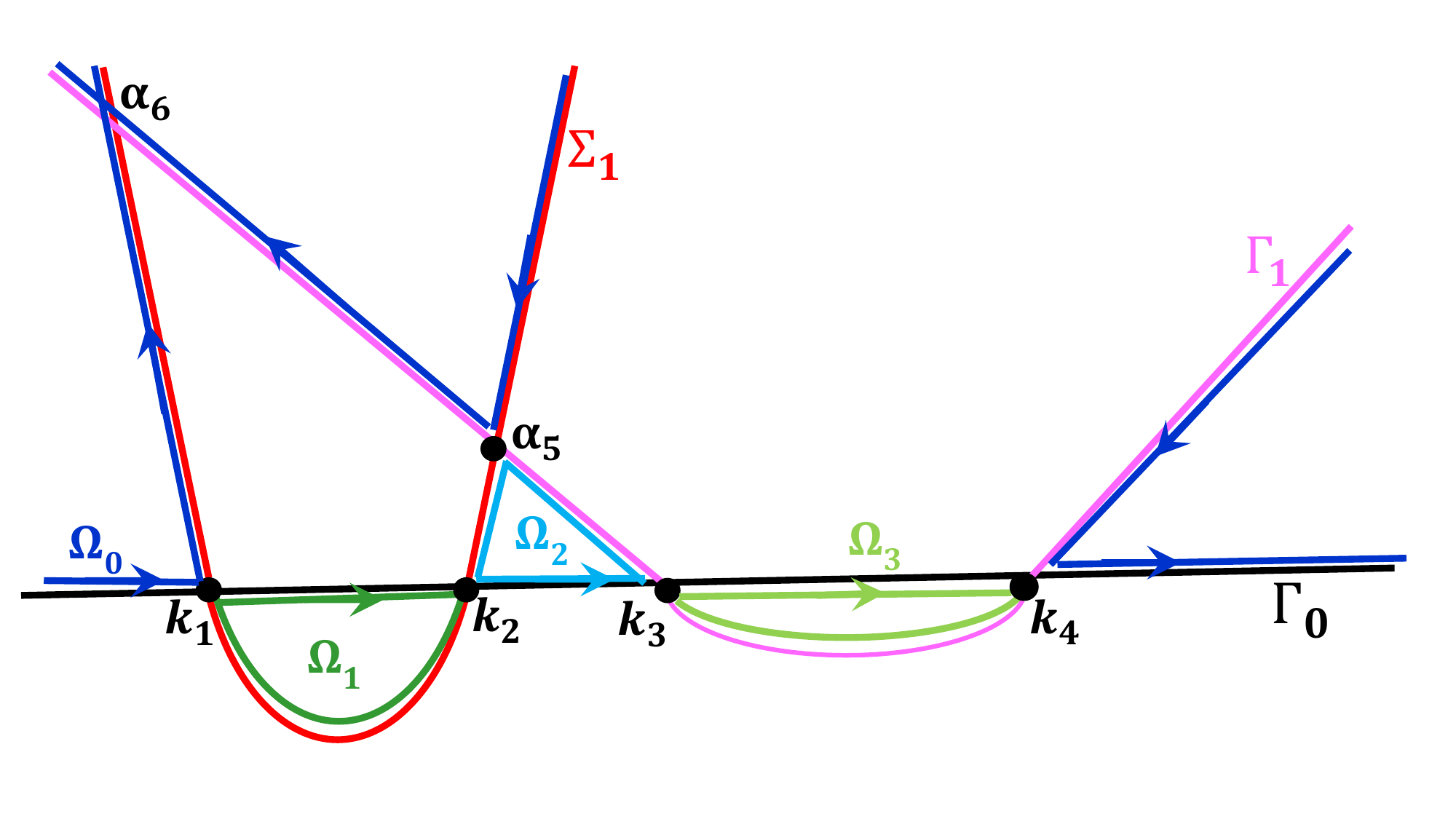}
}
\caption{\small{\sl The network ${\mathcal N}$ [top left], the reduced network ${\mathcal N}_{\mbox{\scriptsize red}}$ [top right], the topological model of the corresponding spectral curve $\Gamma_{T,\mbox{\scriptsize red}}$ [bottom left] and its plane curve representation [bottom right] for soliton data in ${\mathcal S}_{34}^{\mbox{\tiny TNN}}\subset Gr^{\mbox{\tiny TNN}}(2,4)$. $\Psi_{1} (\vec t)=\Psi_{3} (\vec t)+w_{24}\Psi_{4} (\vec t)$. The configurations of the KP divisor (triangles/squares) depend only on the sign of $\Psi_3 (\vec t_0)$.
On the curve double points are represented as dotted segments and $\hat \psi_l \equiv \hat\psi_l (\vec t)$ is as in (\ref{eq:hat_psi}).}} 
\label{fig:Gr24_net}
\end{figure}

${\mathcal S}_{34}^{\mbox{\tiny TNN}}$ is the 3--dimensional positroid cell in $Gr^{\mbox{\tiny TNN}} (2,4)$ corresponding to the matroid
\[
{\mathcal M} = \{ \ 12 \ , \ 13 \ , \ 14 \ ,\ 23 \ , \ 24\ \},
\]
and its elements $[A]$ are equivalence classes of real $2\times 4$ matrices with all maximal minors positive, except $\Delta_{34}=0$. The three positive weights $w_{13},w_{23},w_{24}$ of the Le-tableau (see Figure \ref{fig:Gr24_net}[top,left]) parametrize ${\mathcal S}_{34}^{\mbox{\tiny TNN}}$ and correspond to the matrix in the reduced row echelon form (RREF),
\begin{equation}
\label{eq:RREF}
A = \left( \begin{array}{cccc}
1 & 0 & - w_{13} & -w_{13} w_{24}\\
0 & 1 & w_{23}   & w_{23}w_{24}
\end{array}
\right).
\end{equation}
The generators of the Darboux transformation ${\mathfrak D}=\partial_x^2 -{\mathfrak w}_1 (\vec t)\partial_x -{\mathfrak w}_2(\vec t)$ are
$f^{(1)} (\vec t) = e^{\theta_1(\vec t)} -w_{13}e^{\theta_3(\vec t)}-w_{13} w_{24}e^{\theta_4(\vec t)}$, $
f^{(2)} (\vec t) = e^{\theta_2(\vec t)} +w_{23}e^{\theta_3(\vec t)}+w_{23}w_{24}e^{\theta_4(\vec t)}$.

In the following sections we construct a reducible rational curve $\Gamma_{T,\mbox{\scriptsize red}}$ and the divisor for soliton data $({\mathcal K}, [A])$ with $\mathcal K = \{ \kappa_1 < \kappa_2<\kappa_3<\kappa_4\}$ and $[A] \in {\mathcal S}_{34}^{\mbox{\tiny TNN}}$. We represent $\Gamma_{T,\mbox{\scriptsize red}}$ as a plane curve given by the intersection of a line and two quadrics (see (\ref{eq:lines}) and (\ref{eq:curveGr24})) and we verify that it is a rational degeneration of the genus 3 $\mathtt M$--curve $\Gamma_{\varepsilon}$ ($0<\varepsilon \ll 1$) in (\ref{eq:curveGr24_pert}). We then apply a parallel edge unreduction and a flip move to the reduced network and compute the transformed KP divisor on the transformed curves.

\subsection{Spectral curves for the reduced Le--network and their desingularizations}\label{sec:gamma_24}

We briefly illustrate the construction of a rational spectral curve $\Gamma_{T,\mbox{\scriptsize red}}$ for soliton data in ${\mathcal S}_{34}^{\mbox{\tiny TNN}}$. We choose the reduced Le--graph $\mathcal N_{T,\mbox{\scriptsize red}}$ as dual to the reducible rational curve and use Postnikov rules to assign the weights to construct the corresponding Le--network representing $[A]\in {\mathcal S}_{34}^{\mbox{\tiny TNN}}$  (Figure \ref{fig:Gr24_net}[top right]). In Figure \ref{fig:Gr24_net}[bottom left], we show the topological model of the curve $\Gamma_{T,\mbox{\scriptsize red}}$. 

The reducible rational curve $\Gamma_{T,\mbox{\scriptsize red}}$ is obtained gluing three copies of $\mathbb{CP}^1$, $\Gamma_{T,\mbox{\scriptsize red}}=\Gamma_0\sqcup \Gamma_{1}\sqcup\Sigma_{1}$, and it may be represented as a plane curve given by the intersection of a line ($\Gamma_0$) and two quadrics ($\Gamma_1,\Sigma_1$). We plot both the topological model and the plane curve for this example in Figure \ref{fig:Gr24_net}[bottom]. 
To simplify its representation, we impose that $\Gamma_0$ is one of the coordinate axis in the $(\lambda,\mu)$--plane, say $\mu=0$, that $P_0\in \Gamma_0$ is the infinite point, that the quadrics $\Sigma_{1}$ and  $\Gamma_{1}$ are parabolas with two real finite intersection points $\alpha_5=(\lambda_5,\mu_5)$, $\alpha_6=(\lambda_6,\mu_6)$:
\begin{equation}\label{eq:lines}
\Gamma_0:  \mu=0, \ \ \Gamma_{1}: \mu- (\lambda - \kappa_3)(\lambda-\kappa_4) =0, \ \   \Sigma_{1}: \mu-c_1(\lambda -\kappa_1)(\lambda-\kappa_2)=0.
\end{equation}
In the following we also take $c_1>1$ and choose $\lambda(P_3)=\lambda(Q_1)=\lambda_5$. Then, by construction, $\lambda_5\in ]\kappa_2,\kappa_3[$ and $\lambda_6<\kappa_1$.
As usual we denote $\Omega_0$ the infinite oval, that is $P_0\in \Omega_0$, and $\Omega_j$, $j\in[3]$, the finite ovals. 
Since the singularity at infinity is completely resolved, the quadrics $\Sigma_{1}$ and $\Gamma_{1}$ do not intersect at infinity. The intersection point $\alpha_6$ does not correspond to any of the marked points of the topological model of $\Gamma$. Such singularity is resolved in the partial normalization and therefore there are no extra conditions to be satisfied by the dressed wave functions at $\alpha_6$.

The relation between the coordinate $\lambda$ in the plane curve representation and the coordinate $\zeta$ introduced in Definition \ref{def:loccoor} may be easily worked out at each component of
$\Gamma_{T,\mbox{\scriptsize red}}$.
On $\Gamma_{1}$, we have 3 real ordered marked points $P_m$, $m\in [3]$, with $\zeta$--coordinates: $\zeta(P_1)=0<\zeta(P_2)=1<\zeta(P_3)=\infty$.
Comparing with (\ref{eq:lines}) we then easily conclude that
\[
\lambda = \frac{\lambda_5 (\kappa_4-\kappa_3) \zeta +(\kappa_3-\lambda_5)\kappa_4}{(\kappa_4-\kappa_3) \zeta +\kappa_3-\lambda_5}.
\]
Similarly, on $\Sigma_{1}$, we have 3 real ordered marked points $Q_m$, $m\in [3]$, with $\zeta$--coordinates: $\zeta(Q_1)=0<\zeta(Q_2)=1<\zeta(Q_3)=\infty$.
Comparing with (\ref{eq:lines}) we then easily conclude that
\[
\lambda = \frac{\kappa_1 (\lambda_5-\kappa_2) \zeta +(\kappa_2-\kappa_1)\lambda_5}{(\lambda_5-\kappa_2) \zeta +\kappa_2-\kappa_1}.
\]
$\Gamma_{T,\mbox{\scriptsize red}}$ is represented by the reducible plane curve $\Pi_0(\lambda,\mu)=0$, with
\begin{equation}
\label{eq:curveGr24}
\Pi_0(\lambda,\mu)=\mu\cdot\big(\mu- (\lambda - \kappa_3)(\lambda-\kappa_4)\big)\cdot\big(\mu-c_1 (\lambda - \kappa_1)(\lambda-\kappa_2)\big).
\end{equation}
and is a rational degeneration of the genus 3 $\mathtt M$--curve $\Gamma_{\varepsilon}$ ($0<\varepsilon \ll 1$):
\begin{equation}
\label{eq:curveGr24_pert}
\Gamma_{\varepsilon} \; : \quad\quad \Pi(\lambda, \mu;\varepsilon)=  \Pi_0(\lambda,\mu)-\varepsilon^2\left(\lambda -\lambda_6\right)^2=0.
\end{equation}

\begin{remark}
The plane curve representation for a given cell is not unique. For example, a rational spectral curve $\Gamma_{T,\mbox{\scriptsize red}}$ for soliton data in ${\mathcal S}_{34}^{\mbox{\tiny TNN}}$ can be also represented as the union of one quadric and two lines. 
\end{remark}

\subsection{The KP divisor on $\Gamma_{T,\mbox{\scriptsize red}}$}\label{sec:div_24}
We now construct the wave function and the KP divisor on $\Gamma_{T,\mbox{\scriptsize red}}$. 
For this example, the KP wave function  may take only three possible values at the marked points:
\begin{equation}
\label{eq:hat_psi}
\hat \psi_l (\vec t) = \frac{{\mathfrak D} e^{\theta_l(\vec t)}}{{\mathfrak D} e^{\theta_l(\vec t_0)}}, \ \ l=1,3,4,
\end{equation}
where $\theta_l(\vec t) = \kappa_l x+  \kappa_l^2 y+ \kappa_l^3 t$. In Figure~\ref{fig:Gr24_net} [bottom left] we show which double point carries which of the above values of $\hat \psi$. At each marked point the value $\hat \psi_l (\vec t)$ is independent on the choice of local coordinates on the components, {\sl i.e.} of the orientation in the network. 

The  local coordinate of each divisor point may be computed using (\ref{eq:formula_div}).
On $\Gamma=\Gamma_{T,\mbox{\scriptsize red}}$, the KP divisor $\DKP$ consists of the degree $k=2$ Sato divisor $(\gamma_{S,1} ,\gamma_{S,2} )=(\gamma_{S,1} (\vec t_0),\gamma_{S,2} (\vec t_0))$ defined in (\ref{eq:Satodiv}) and of $1$ simple pole $\gamma_{1}=\gamma_{1} (\vec t_0)$ belonging to the intersection of 
$\Gamma_{1}$ with the union of the finite ovals. In the local coordinates induced by the orientation of the network (see Definition \ref{def:loccoor}), we have
\begin{equation}\label{eq:ex_div1}
\zeta(\gamma_{S,1}) +\zeta(\gamma_{S,2}) = {\mathfrak w}_1 (\vec t_0), \quad  \zeta(\gamma_{S,1})\zeta(\gamma_{S,2}) = -{\mathfrak w}_2 (\vec t_0), \quad
\zeta(\gamma_{1}) = \frac{w_{24} {\mathfrak D} e^{\theta_4(\vec t_0)}}{{\mathfrak D} e^{\theta_3(\vec t_0)}+ w_{24} 
{\mathfrak D} e^{\theta_4(\vec t_0)}}.
\end{equation}
It is easy to verify that ${\mathfrak D} e^{\theta_1(\vec t)}$, ${\mathfrak D} e^{\theta_4(\vec t)}>0$ and ${\mathfrak D} e^{\theta_2(\vec t)} = - \frac{w_{23}}{w_{13}}{\mathfrak D} e^{\theta_1(\vec t)}$. Therefore, for generic soliton data $[A]\in {\mathcal S}_{34}^{\mbox{\tiny TNN}}$, the KP--II pole divisor configuration is one of the two shown in Figure \ref{fig:Gr24_net} [bottom left]:
\begin{enumerate}
\item If ${\mathfrak D} e^{\theta_3(\vec t_0)}>0$, then $\gamma_{S,1} \in \Omega_1$, $\gamma_{S,2} \in \Omega_2$ and $\gamma_{1} \in \Omega_3$. One such configuration is illustrated by triangles in the Figure;
\item If ${\mathfrak D} e^{\theta_3(\vec t_0)}<0$, then $\gamma_{S,1} \in \Omega_1$, $\gamma_{S,2} \in \Omega_3$ and $\gamma_{1} \in \Omega_2$. One such configuration is illustrated by squares in the Figure.
\end{enumerate}
As expected, there is exactly one KP divisor point in each finite oval of $\Gamma$, where we use the counting rule established in \cite{AG1} for non--generic soliton data satisfying ${\mathfrak D} e^{\theta_3(\vec t_0)}=0$.

\subsection{The effect of Postnikov moves and reductions on the KP divisor}

\begin{figure}
\centering{
\includegraphics[width=0.4\textwidth]{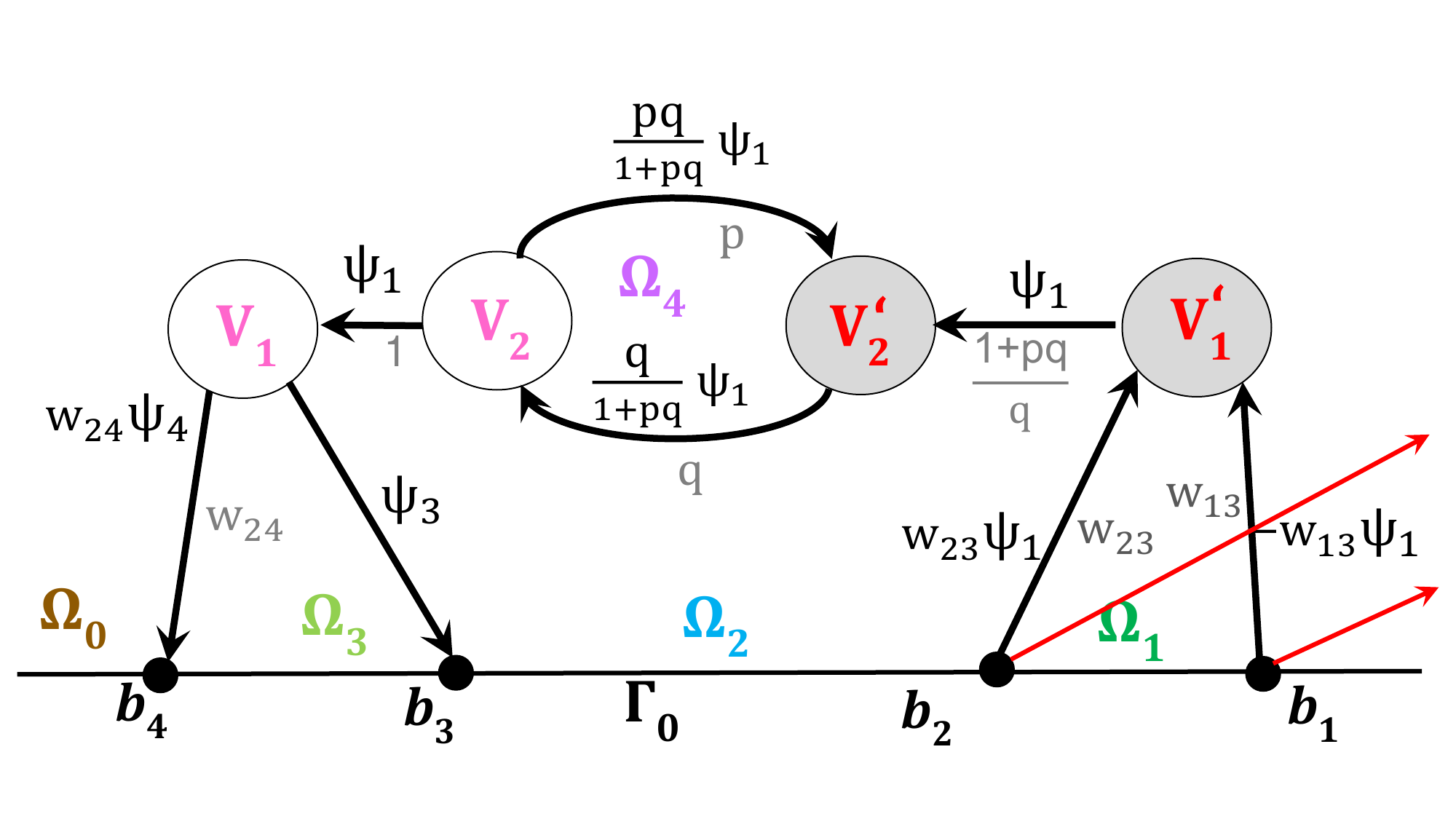}
\hfill
\includegraphics[width=0.4\textwidth]{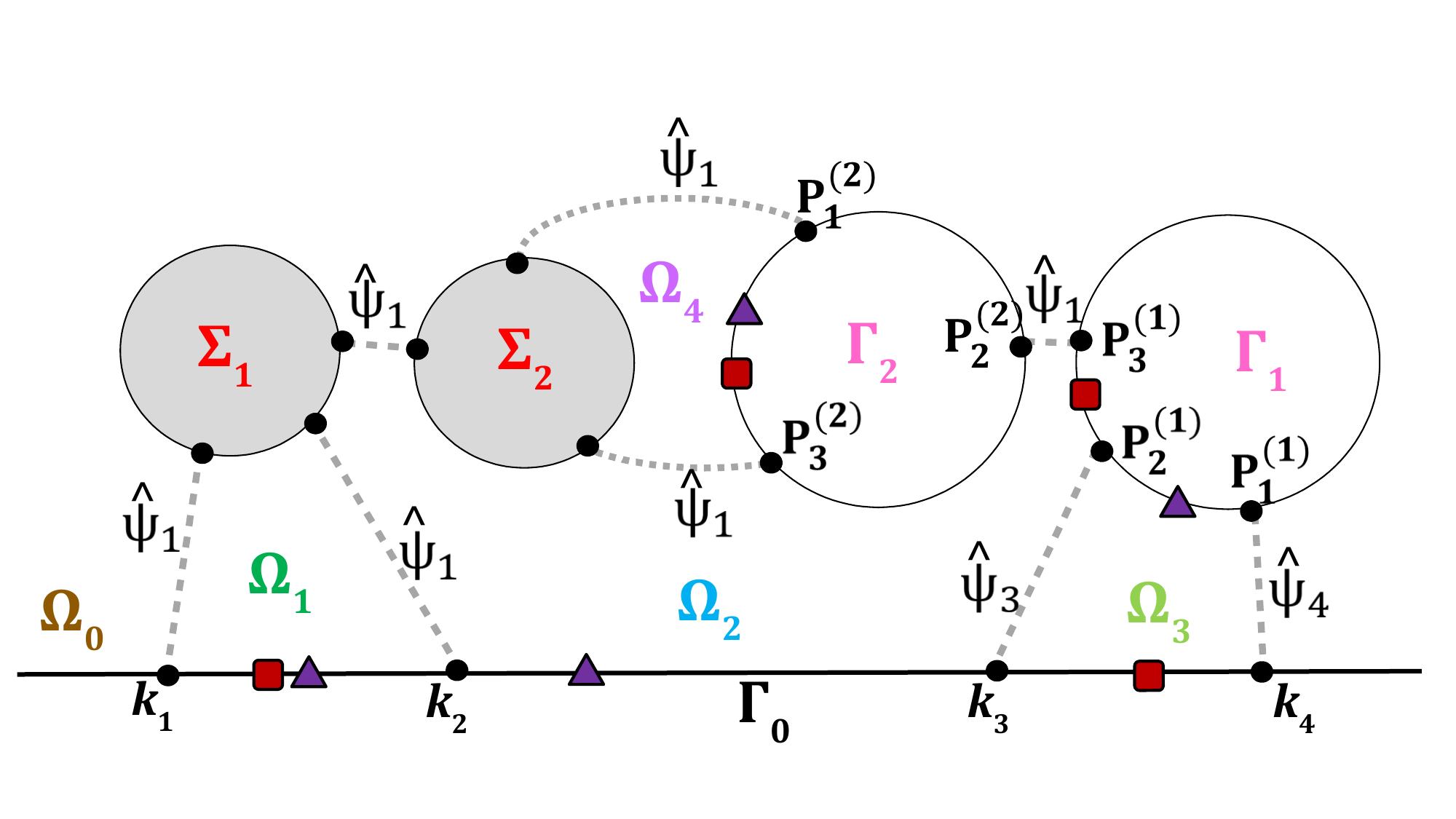}
\includegraphics[width=0.4\textwidth]{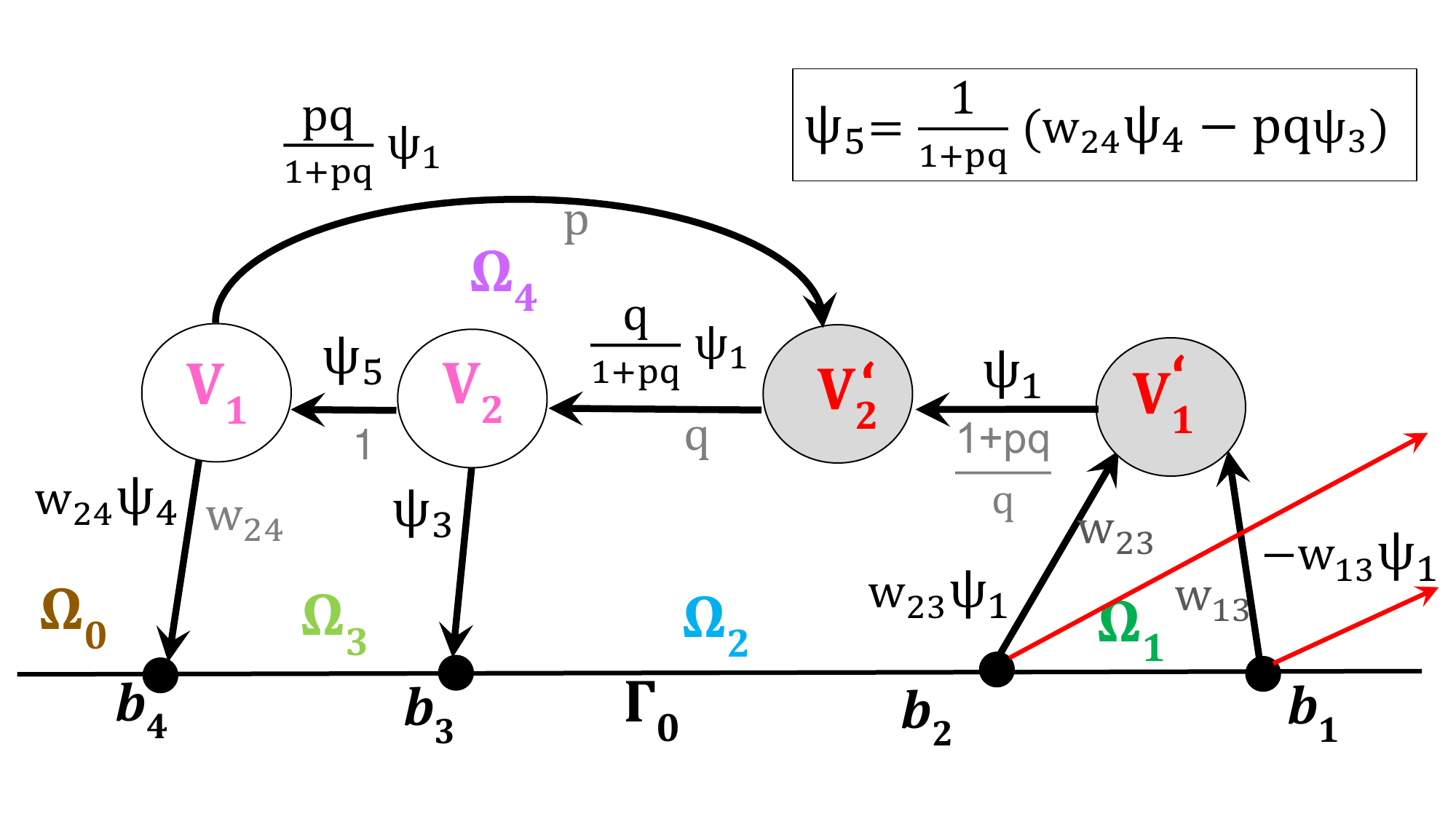}
\hfill
\includegraphics[width=0.4\textwidth]{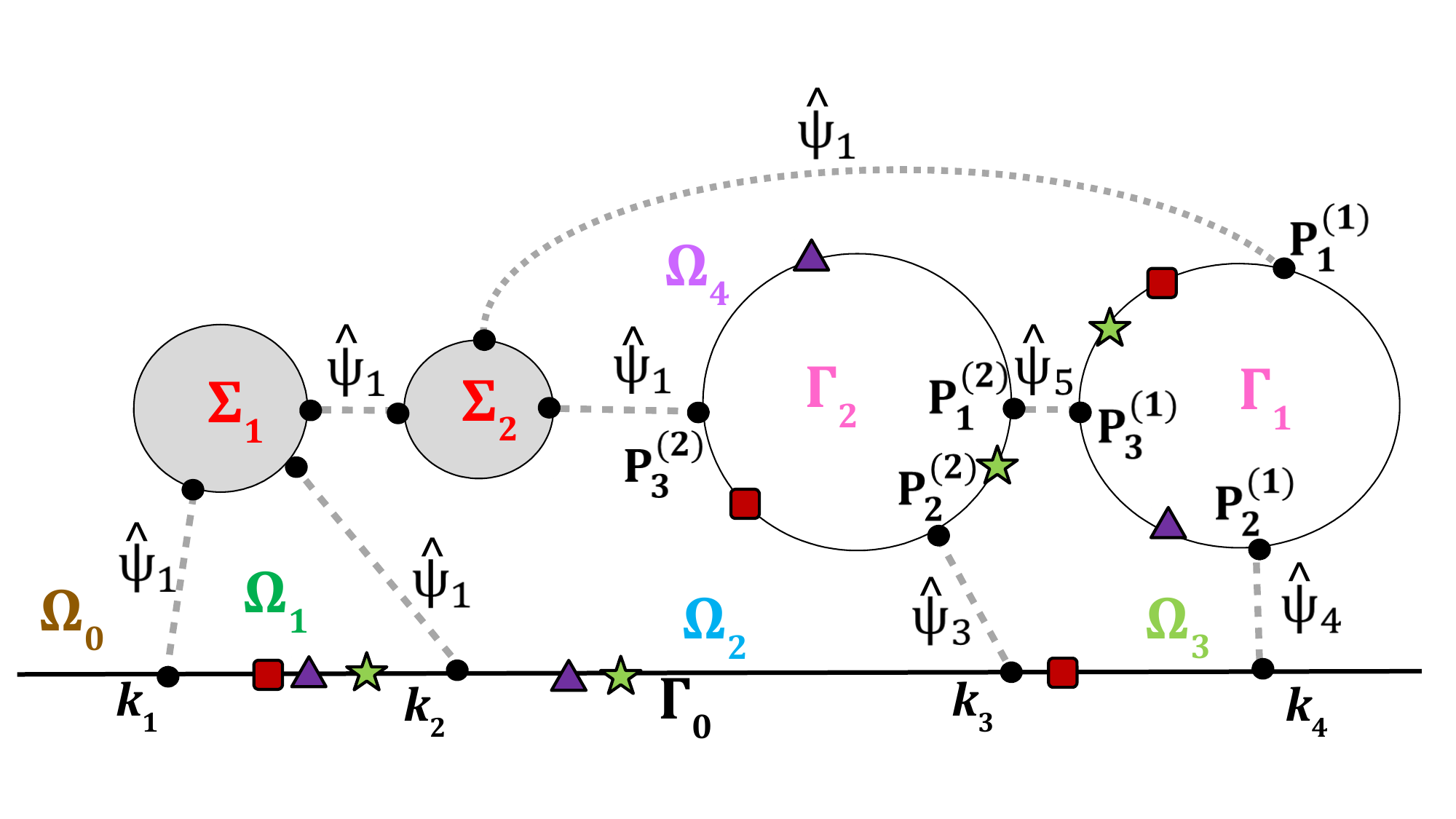}
}
\caption{\small{\sl Left: ${\mathcal N}_{\mbox{\scriptsize par}}$ [top] is obtained from the reduced Le-network ${\mathcal N}_{T,{\mbox{\scriptsize red}}}$ of Figure \ref{fig:Gr24_net} applying the parallel edge unreduction, whereas ${\mathcal N}_{\mbox{\scriptsize flip}}$ [bottom] is obtained from ${\mathcal N}_{\mbox{\scriptsize par}}$ applying a flip move.
Right: The partial normalization of the corresponding curves ${\Gamma}_{\mbox{\scriptsize par}}$ [top], ${\Gamma}_{\mbox{\scriptsize flip}}$ [bottom] and the possible divisor configurations.
}}
\label{fig:Gr24_moves}
\end{figure}

Next we show the effect of moves and reductions on the divisor position. 

We first apply a parallel edge unreduction to ${\mathcal N}_{T,{\mbox{\scriptsize red}}}$ (Figure \ref{fig:Gr24_net}) and obtain the network  ${\mathcal N}_{{\mbox{\scriptsize par}}}$ in Figure \ref{fig:Gr24_moves}[top,left]. We remark that we have a gauge freedom in assigning the weights to the edges involved in this transformation provided that $p,q>0$. The values of the unnormalized dressed wave function are shown in the Figure. The corresponding curve,  ${\Gamma}_{\mbox{\scriptsize par}}$, is presented in Figure \ref{fig:Gr24_moves}[top,right]. By construction the KP divisor $\DKP$ consists of the degree $k=2$ Sato divisor $(\gamma_{S,1} ,\gamma_{S,2} )=(\gamma_{S,1} (\vec t_0),\gamma_{S,2} (\vec t_0))$ computed in (\ref{eq:ex_div1}) and of the simple poles $\gamma_{i}=\gamma_{i} (\vec t_0)$ belonging to the intersection of 
$\Gamma_{i}$, $i=1,2$, with the union of the finite ovals. In the local coordinates induced by the orientation of the network (Definition \ref{def:loccoor}), we have
\begin{equation}\label{eq:ex_div2}
\zeta(\gamma_{1}) = \frac{w_{24} {\mathfrak D} e^{\theta_4(\vec t_0)}}{{\mathfrak D} e^{\theta_3(\vec t_0)}+ w_{24} 
{\mathfrak D} e^{\theta_4(\vec t_0)}}, \quad\quad   \zeta(\gamma_{2}) = -pq.
\end{equation}
For generic soliton data $[A]\in {\mathcal S}_{34}^{\mbox{\tiny TNN}}$, the KP--II pole divisor configurations are shown in Figure \ref{fig:Gr24_moves} [top,right]:
\begin{enumerate}
\item $\gamma_{2} \in \Omega_4$ independently of the sign of ${\mathfrak D} e^{\theta_3(\vec t_0)}>0$;
\item If ${\mathfrak D} e^{\theta_3(\vec t_0)}>0$, then $\gamma_{S,1} \in \Omega_1$, $\gamma_{S,2} \in \Omega_2$ and $\gamma_{1} \in \Omega_3$. One such configuration is illustrated by triangles in the Figure;
\item If ${\mathfrak D} e^{\theta_3(\vec t_0)}<0$, then $\gamma_{S,1} \in \Omega_1$, $\gamma_{S,2} \in \Omega_3$ and $\gamma_{1} \in \Omega_2$. One such configuration is illustrated by squares in the Figure.
\end{enumerate}
Again, for any given $[A]\in {\mathcal S}_{34}^{\mbox{\tiny TNN}}$, there is exactly one KP divisor point in each finite oval.

Next we apply a flip move to ${\mathcal N}_{\mbox{\scriptsize par}}$ and obtain the network  ${\mathcal N}_{{\mbox{\scriptsize flip}}}$ in Figure \ref{fig:Gr24_moves}[bottom,left] and show the values on the un--normalized dressed wave function in the Figure. The corresponding curve,  ${\Gamma}_{\mbox{\scriptsize flip}}$, is presented in Figure \ref{fig:Gr24_moves}[bottom,right]. By construction the KP divisor $\DKP$ consists of the degree $k=2$ Sato divisor $(\gamma_{S,1} ,\gamma_{S,2} )=(\gamma_{S,1} (\vec t_0),\gamma_{S,2} (\vec t_0))$ computed in (\ref{eq:ex_div1}) and of the simple poles $\tilde \gamma_{i}=\tilde \gamma_{i} (\vec t_0)$ belonging to the intersection of 
$\Gamma_{i}$, $i=1,2$, with the union of the finite ovals. In the local coordinates induced by the orientation of the network, we have
\begin{equation}\label{eq:ex_div3}
\zeta(\gamma_{1}) = \frac{pq \left({\mathfrak D} e^{\theta_3(\vec t_0)} + w_{24} {\mathfrak D} e^{\theta_4(\vec t_0)} \right)}{{pq\mathfrak D} e^{\theta_3(\vec t_0)}- w_{24} 
{\mathfrak D} e^{\theta_4(\vec t_0)}}, \quad\quad   \zeta(\gamma_{2}) = -\frac{pq{\mathfrak D} e^{\theta_3(\vec t_0)}- w_{24} 
{\mathfrak D} e^{\theta_4(\vec t_0)}}{{\mathfrak D} e^{\theta_3(\vec t_0)} + w_{24} {\mathfrak D} e^{\theta_4(\vec t_0)}}.
\end{equation}
For generic soliton data $[A]\in {\mathcal S}_{34}^{\mbox{\tiny TNN}}$, the KP--II pole divisor configuration is one of the three shown in Figure \ref{fig:Gr24_net} [right]:
\begin{enumerate}
\item If ${\mathfrak D} e^{\theta_3(\vec t_0)}>0$  and $pq{\mathfrak D} e^{\theta_3(\vec t_0)}> w_{24} 
{\mathfrak D} e^{\theta_4(\vec t_0)}$, then $\gamma_{S,1} \in \Omega_1$, $\gamma_{S,2} \in \Omega_2$,
 and $\tilde \gamma_{1} \in \Omega_3$ and $\tilde \gamma_{2} \in \Omega_4$. One such configuration is illustrated by triangles in the Figure;
\item If ${\mathfrak D} e^{\theta_3(\vec t_0)}>0$  and $pq{\mathfrak D} e^{\theta_3(\vec t_0)}< w_{24} 
{\mathfrak D} e^{\theta_4(\vec t_0)}$, then $\gamma_{S,1} \in \Omega_1$, $\gamma_{S,2} \in \Omega_2$,
 and $\tilde \gamma_{1} \in \Omega_4$ and $\tilde \gamma_{2} \in \Omega_3$. One such configuration is illustrated by stars in the Figure;
\item If ${\mathfrak D} e^{\theta_3(\vec t_0)}<0$, then $\gamma_{S,1} \in \Omega_1$, $\gamma_{S,2} \in \Omega_3$, $\tilde \gamma_{1} \in \Omega_4$ and $\tilde \gamma_{2} \in \Omega_2$. One such configuration is illustrated by squares in the Figure.
\end{enumerate}
Again there is exactly one KP divisor point in each finite oval.

\section{Example: effect of the square move on the KP divisor for soliton data in $Gr^{\mbox{\tiny TP}}(2,4)$}\label{sec:ex_Gr24top}
\begin{figure}
\centering{
\includegraphics[width=0.4\textwidth]{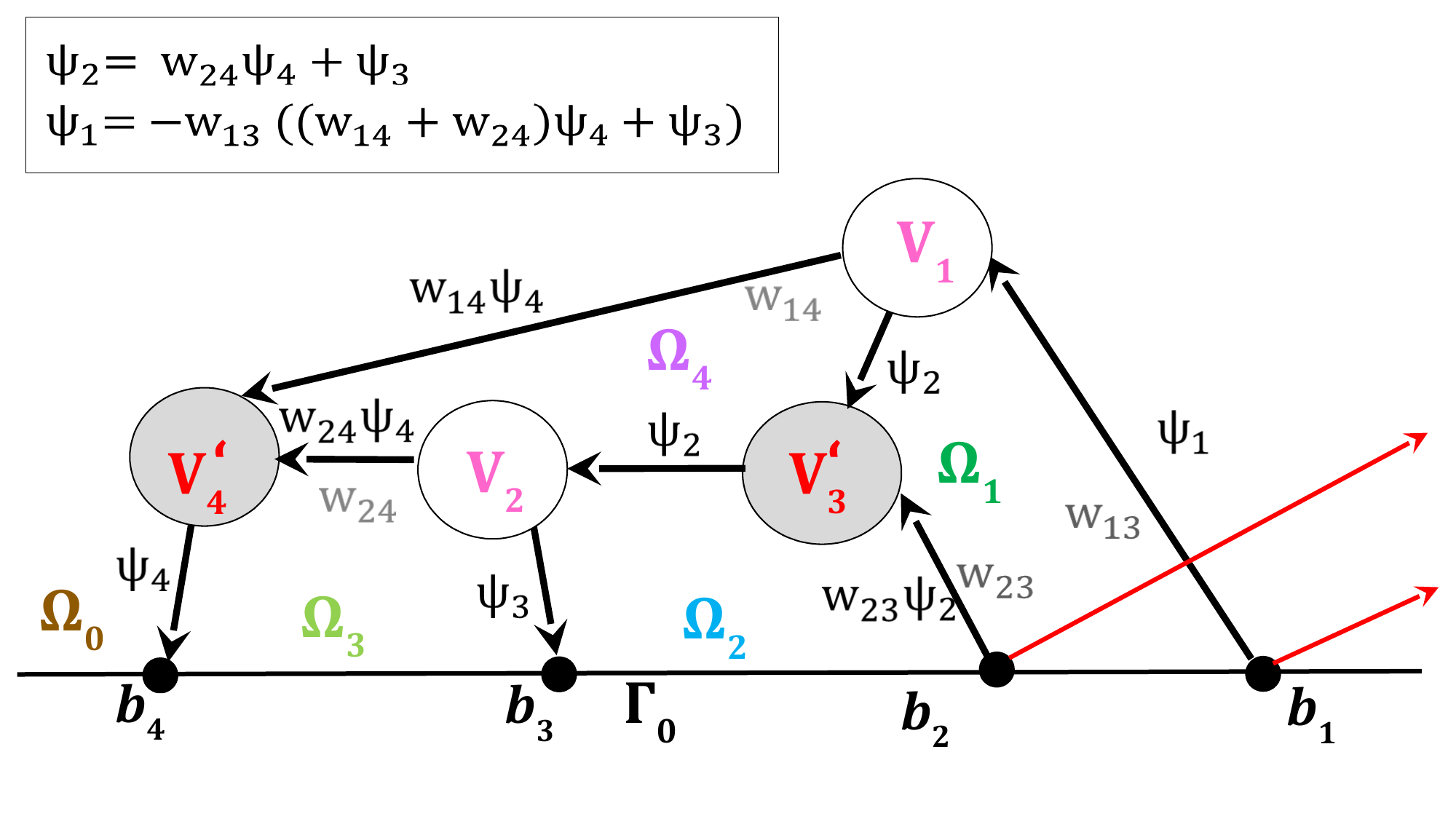}
\hfill
\includegraphics[width=0.4\textwidth]{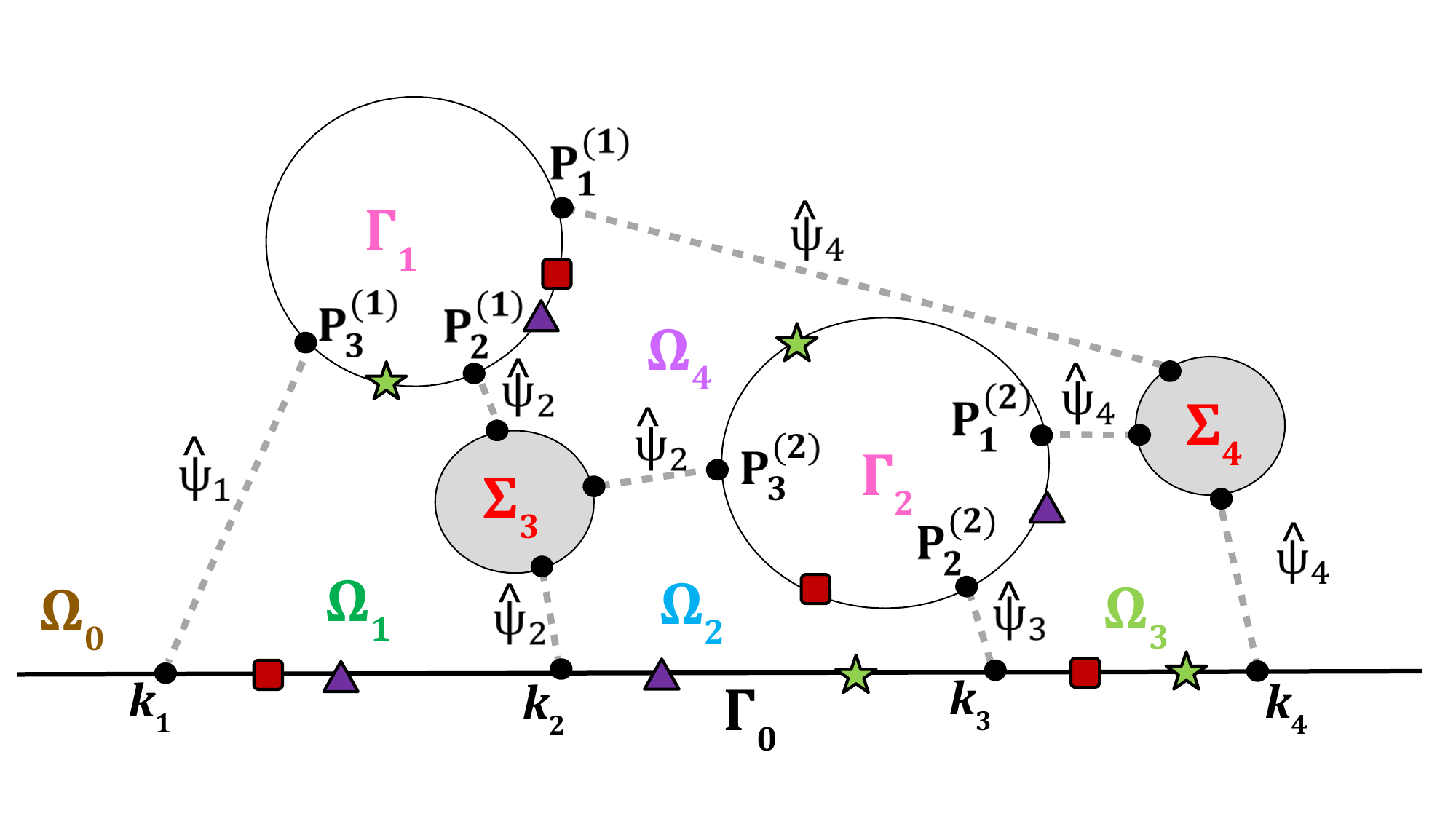}
\includegraphics[width=0.4\textwidth]{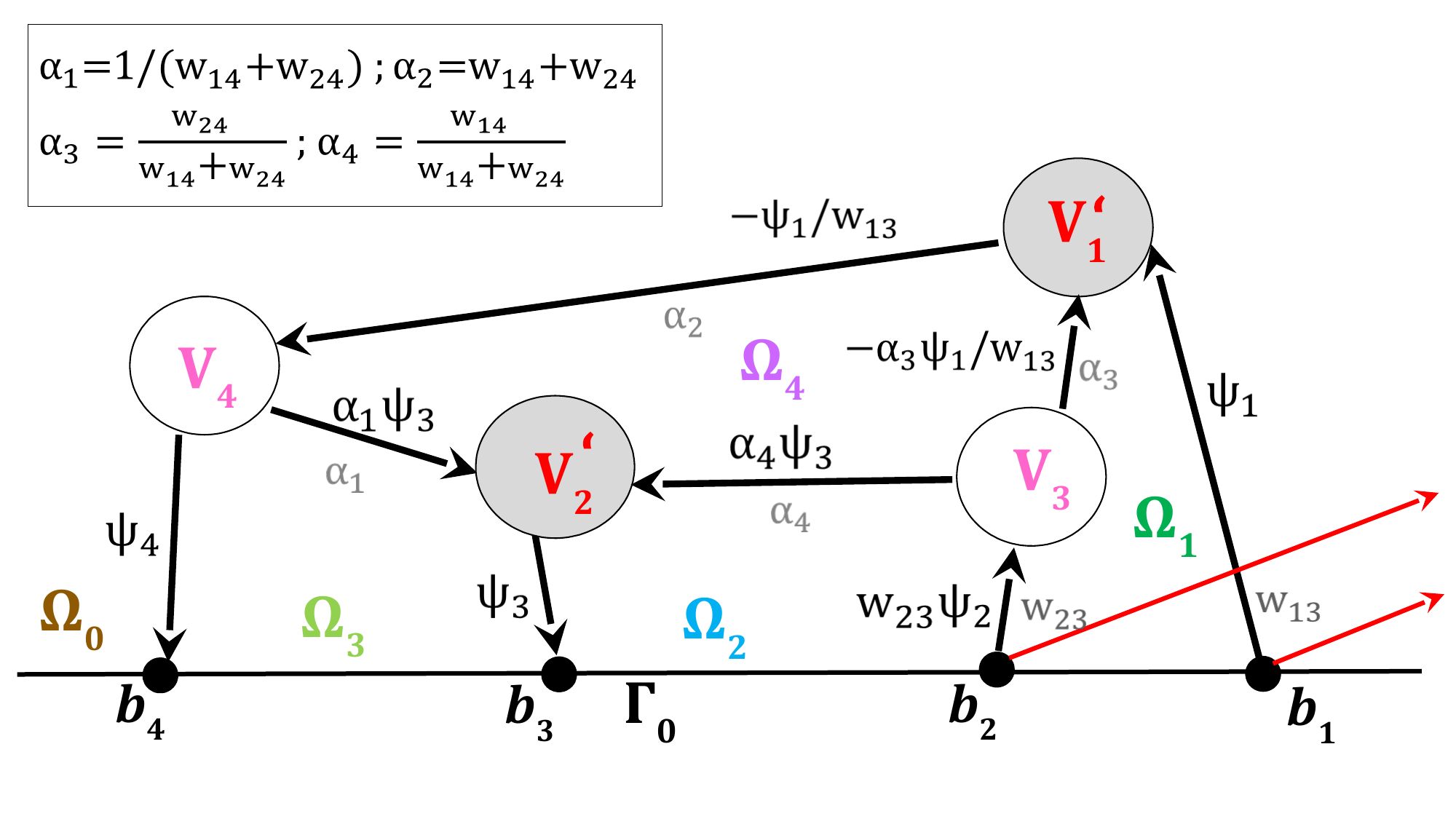}
\hfill
\includegraphics[width=0.4\textwidth]{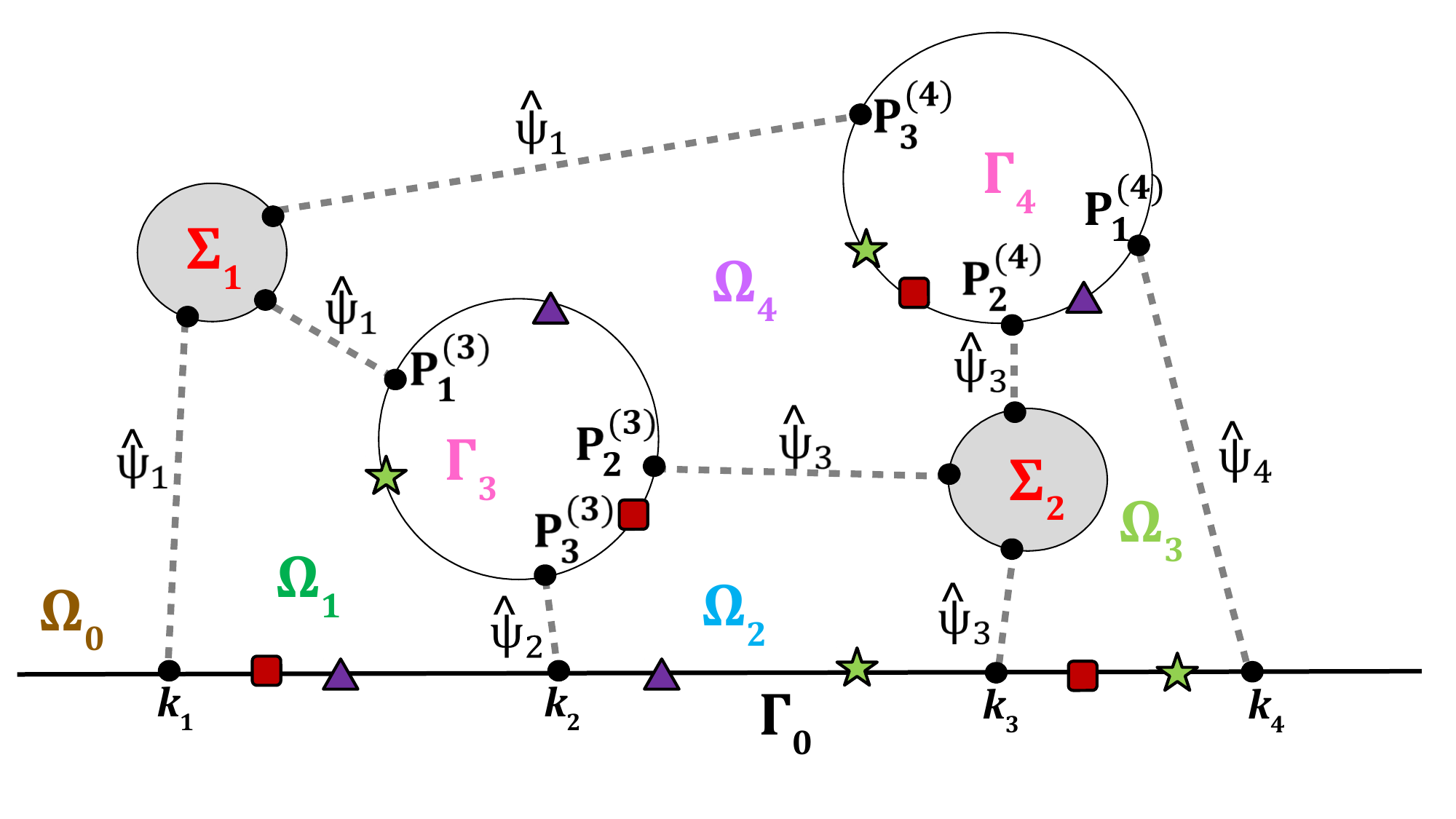}
}
\caption{\small{\sl Left: ${\mathcal N}_{\mbox{\scriptsize sq-mv}}$ [bottom] is obtained from the reduced Le-network ${\mathcal N}_{T,{\mbox{\scriptsize top}}}$ [top] applying a square move. Right: The partial normalization of the corresponding curves ${\Gamma}_{\mbox{\scriptsize top}}$ [top], ${\Gamma}_{\mbox{\scriptsize sq-mv}}$ [bottom] and the possible divisor configurations.}}
\label{fig:Gr24_square}
\end{figure}

The simplest network to which the square move is applicable is the reduced Le--network ${\mathcal N}_{T,{\mbox{\scriptsize top}}}$ associated to soliton data in $Gr^{\mbox{\tiny TP}} (2,4)$ and corresponds to the change of colour of all internal vertices. The latter transformation may be also interpreted as a self--dual transformation in $Gr^{\mbox{\tiny TP}} (2,4)$.

The reduced networks and the topological models of the curves before and after the square move for soliton data in $Gr^{\mbox{\tiny TP}} (2,4)$ are shown in Figure \ref{fig:Gr24_square}.
In \cite{A2,AG3} we have already computed a plane curve representation and its desingularization, and discussed the divisor configurations on ${\mathcal N}_{T,{\mbox{\scriptsize top}}}$. The duality transformation implies that the same plane curve representation is associated both to ${\Gamma}_{\mbox{\scriptsize top}}$ and ${\Gamma}_{\mbox{\scriptsize sq-mv}}$, by conveniently relabeling $\mathbb{CP}^1$ components from $\Sigma_i,\Gamma_j$ to $\Gamma_i,\Sigma_j$ (compare the topological models of curves in Figure \ref{fig:Gr24_square}).

Here we just compute the KP divisor after the square move and refer to \cite{AG2} for more details on this example.
The values on the dressed edge wave functions are shown in Figure \ref{fig:Gr24_square}. 

By definition, the Sato divisor is not affected by the square move since the Darboux transformation is the same.
Therefore the degree $k=2$ Sato divisor $(\gamma_{S,1} ,\gamma_{S,2} )=(\gamma_{S,1} (\vec t_0),\gamma_{S,2} (\vec t_0))$ is obtained solving
$\zeta(\gamma_{S,1}) +\zeta(\gamma_{S,2}) = {\mathfrak w}_1 (\vec t_0)$, $\zeta(\gamma_{S,1})\zeta(\gamma_{S,2}) = -{\mathfrak w}_2 (\vec t_0)$, 
where the Darboux transformation $\mathfrak D = \partial_x^2 -{\mathfrak w}_1 (\vec t) \partial_x -{\mathfrak w}_2 (\vec t)$ is generated by the heat hierarchy solutions
$f^{(1)} (\vec t) = e^{\theta_1(\vec t)}-w_{13} e^{\theta_3(\vec t)}-w_{13}(w_{14}+w_{24}) e^{\theta_4(\vec t)}$, $
f^{(2)} (\vec t) = e^{\theta_2(\vec t)}+w_{23} e^{\theta_3(\vec t)}+w_{23}w_{24} e^{\theta_4(\vec t)}$.

On $\Gamma_{\mbox{\scriptsize top}}$, $\DKP = (\gamma_{S,1} ,\gamma_{S,2},\gamma_1,\gamma_2 )$
where the simple poles $\gamma_{i}=\gamma_{i} (\vec t_0)$ belong to the intersection of 
$\Gamma_{i}$, $i=1,2$, with the union of the finite ovals. In the local coordinates induced by the orientation of ${\mathcal N}_{T,{\mbox{\scriptsize top}}}$, we have
\begin{equation}\label{eq:ex_div_red}
\zeta(\gamma_{1}) = \frac{w_{14} {\mathfrak D} e^{\theta_4(\vec t_0)}}{{\mathfrak D} e^{\theta_3(\vec t_0)}+(w_{14}+ w_{24})
{\mathfrak D} e^{\theta_4(\vec t_0)}}, \quad\quad   \zeta(\gamma_{2}) = \frac{w_{24}{\mathfrak D} e^{\theta_4(\vec t_0)}}{{\mathfrak D} e^{\theta_3(\vec t_0)}+ w_{24}
{\mathfrak D} e^{\theta_4(\vec t_0)}}.
\end{equation}
It is straightforward to verify that ${\mathfrak D} e^{\theta_1(\vec t)}, {\mathfrak D} e^{\theta_4(\vec t)}>0$ for all $\vec t$.
As observed in \cite{A2,AG2}, there are three possible generic configurations of the KP--II pole divisor depending on the signs of 
${\mathfrak D} e^{\theta_2(\vec t_0)}$ and ${\mathfrak D} e^{\theta_3(\vec t_0)}$ (see also Figure \ref{fig:Gr24_square} [top,right]):
\begin{enumerate}
\item If ${\mathfrak D} e^{\theta_2(\vec t_0)}<0<{\mathfrak D} e^{\theta_3(\vec t_0)}$, then $\gamma_{S,1} \in \Omega_1$, $\gamma_{S,2} \in \Omega_2$, $\gamma_{1} \in \Omega_4$ and $\gamma_{2} \in \Omega_3$. One such configuration is illustrated by triangles in the Figure;
\item If ${\mathfrak D} e^{\theta_2(\vec t_0)},{\mathfrak D} e^{\theta_3(\vec t_0)}<0$, then $\gamma_{S,1} \in \Omega_1$, $\gamma_{S,2} \in \Omega_3$, $\gamma_{1} \in \Omega_4$ and $\gamma_{2} \in \Omega_2$. One such configuration is illustrated by squares in the Figure;
\item If ${\mathfrak D} e^{\theta_3(\vec t_0)}<0<{\mathfrak D} e^{\theta_2(\vec t_0)}$, then $\gamma_{S,1} \in \Omega_2$, $\gamma_{S,2} \in \Omega_3$, $\gamma_{1} \in \Omega_1$ and $\gamma_{2} \in \Omega_4$. One such configuration is illustrated by stars in the Figure.
\end{enumerate}
On $\Gamma_{\mbox{\scriptsize sq-mv}}$, $\DKP = (\gamma_{S,1} ,\gamma_{S,2},\gamma_3,\gamma_4 )$
where the simple poles $\gamma_{i}=\gamma_{i} (\vec t_0)$ belong to the intersection of 
$\Gamma_{i}$, $i=3,4$, with the union of the finite ovals. In the local coordinates induced by the orientation of ${\mathcal N}_{\mbox{\scriptsize sq-mv}}$, we have
\begin{equation}\label{eq:ex_div_square}
\zeta(\gamma_{3}) = \frac{w_{24} \left( {\mathfrak D} e^{\theta_3(\vec t_0)} + (w_{14}+ w_{24})
{\mathfrak D} e^{\theta_4(\vec t_0)}\right)}{(w_{14}+ w_{24})\left({\mathfrak D} e^{\theta_3(\vec t_0)}+ w_{24}
{\mathfrak D} e^{\theta_4(\vec t_0)}\right)}, \quad\quad   \zeta(\gamma_{4}) = \frac{(w_{14}+w_{24}){\mathfrak D} e^{\theta_4(\vec t_0)}}{{\mathfrak D} e^{\theta_3(\vec t_0)}+ (w_{14}+w_{24})
{\mathfrak D} e^{\theta_4(\vec t_0)}}.
\end{equation}
The three possible generic configurations of the KP--II pole divisor after the square move are  then (see also Figure \ref{fig:Gr24_square} [bottom,right]):
\begin{enumerate}
\item If ${\mathfrak D} e^{\theta_2(\vec t_0)}<0<{\mathfrak D} e^{\theta_3(\vec t_0)}$, then $\gamma_{S,1} \in \Omega_1$, $\gamma_{S,2} \in \Omega_2$, $\gamma_{3} \in \Omega_4$ and $\gamma_{4} \in \Omega_3$. One such configuration is illustrated by triangles in the Figure;
\item If ${\mathfrak D} e^{\theta_2(\vec t_0)},{\mathfrak D} e^{\theta_3(\vec t_0)}<0$, then $\gamma_{S,1} \in \Omega_1$, $\gamma_{S,2} \in \Omega_3$, $\gamma_{3} \in \Omega_2$ and $\gamma_{4} \in \Omega_4$. One such configuration is illustrated by squares in the Figure;
\item If ${\mathfrak D} e^{\theta_3(\vec t_0)}<0<{\mathfrak D} e^{\theta_2(\vec t_0)}$, then $\gamma_{S,1} \in \Omega_2$, $\gamma_{S,2} \in \Omega_3$, $\gamma_{3} \in \Omega_1$ and $\gamma_{4} \in \Omega_4$. One such configuration is illustrated by stars in the Figure.
\end{enumerate}
The transformation rule of the divisor points is in agreement with the effect of the square move discussed in Section \ref{sec:moves_reduc}.
As expected, for any given $[A]\in Gr^{\mbox{\tiny TP}}(2,4)$, there is exactly one KP divisor point in each finite oval. Non generic divisor configurations (squares) correspond either to
${\mathfrak D} e^{\theta_2(\vec t_0)}=0$ or to ${\mathfrak D} e^{\theta_3(\vec t_0)}=0$ since ${\mathfrak D} e^{\theta_2(\vec t)}
+w_{23}{\mathfrak D} e^{\theta_3(\vec t)}<0$, for all $\vec t$. We plan to discuss the exact definition of global parametrization in the case of non generic divisor configurations using resolution of singularities in a future publication, see also next Section~\ref{sec:global} for the case $Gr^{TP}(1,3)$.

\section{Amalgamation of positroid cells and divisor structure}\label{sec:amalg}

In \cite{FG1} Fock and Goncharov introduced amalgamation of cluster varieties, which has turned out to be relevant both for constructing integrable systems on Poisson cluster varieties \cite{KG} and for computation of scattering amplitudes on on-shell diagrams in $N=4$ SYM theory \cite{AGP1}. Amalgamation of positroid varieties admits a very simple representation in terms of simple operations on the corresponding plabic graphs. In our setting the vertices are those of the graph and the frozen ones are those at the boundary. Since the planarity property is essential, amalgamation is represented by compositions of the following elementary operations:
\begin{enumerate}
\item Disjoint union of a pair of planar graphs (see Figure~\ref{fig:amalg1}) corresponding to the direct sum of the corresponding Grassmannians via a map $Gr^{\mbox{\tiny{TNN}}}(k_1,n_1)\times Gr^{\mbox{\tiny{TNN}}}(k_2,n_2)\rightarrow Gr^{\mbox{\tiny{TNN}}}(k_1+k_2,n_1+n_2)$;
\item Defrosting of a pair of consecutive boundary vertices (see Figure~\ref{fig:amalg2})  corresponding to a projection map $Gr^{\mbox{\tiny{TNN}}}(k,n)\rightarrow Gr^{\mbox{\tiny{TNN}}}(k-1,n-2)$. 
\end{enumerate}

\begin{figure}
  \centering{\includegraphics[width=0.45\textwidth]{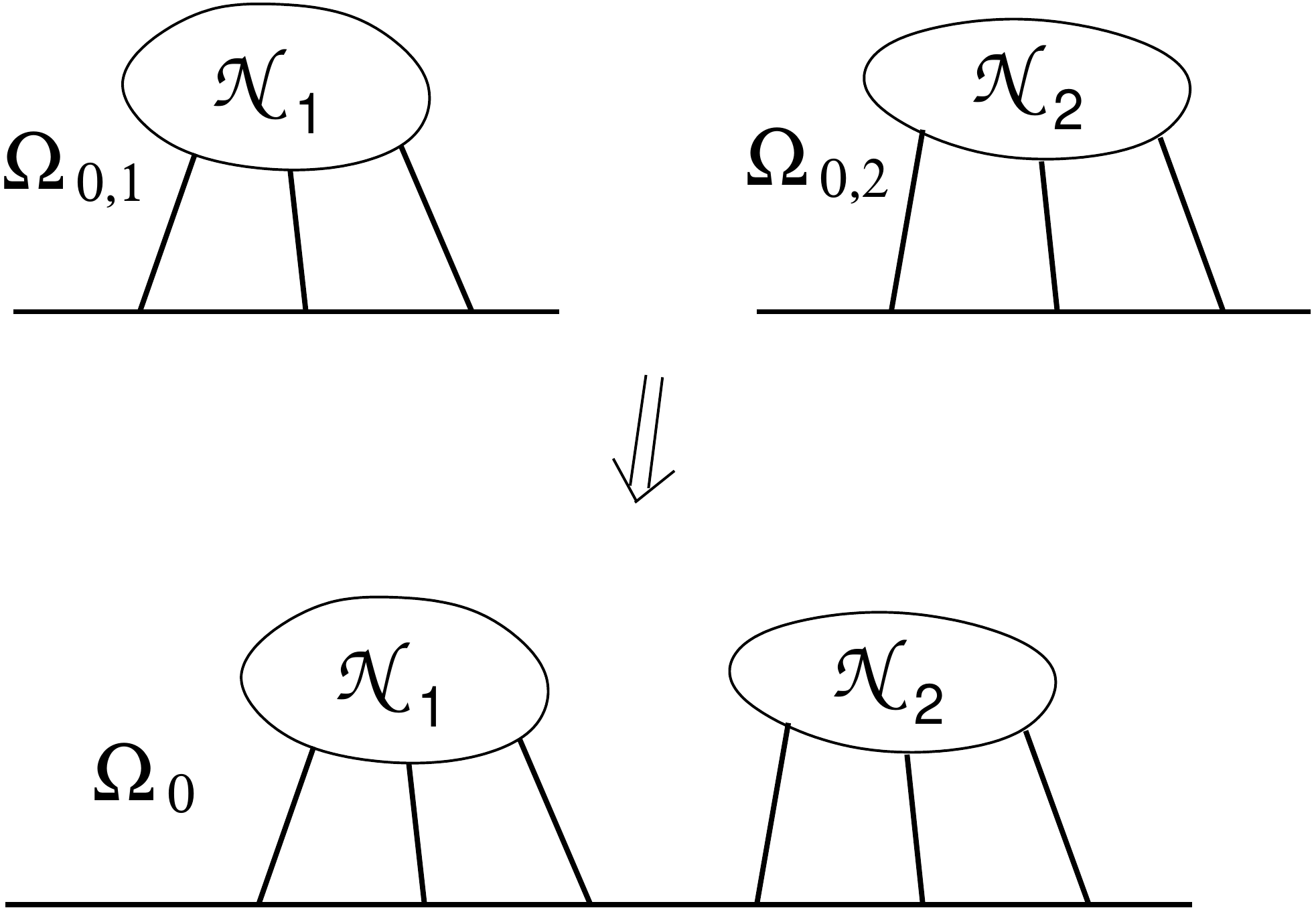}
	\hfill
	\includegraphics[width=0.45\textwidth]{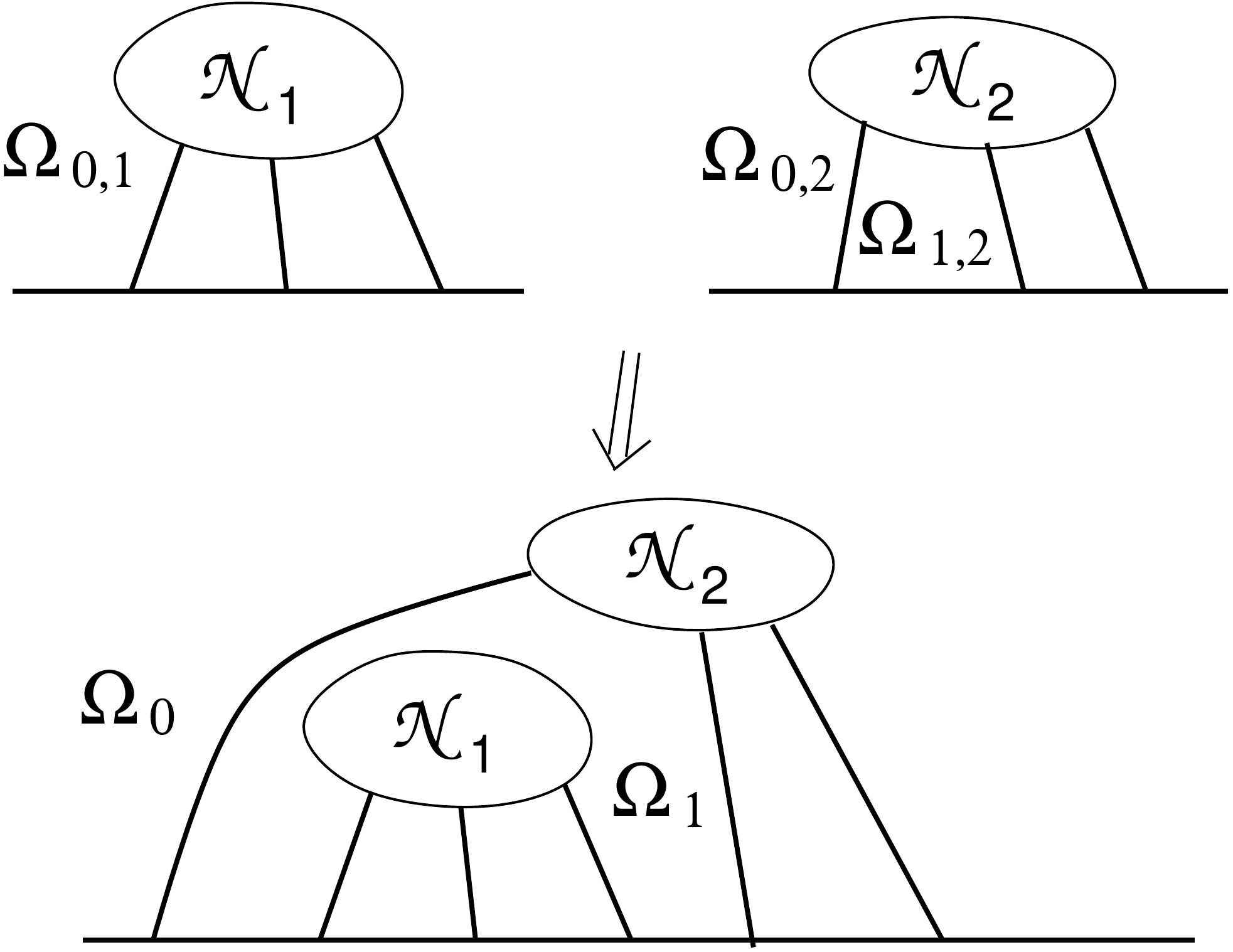}}
  \caption{\small{\sl The two possible ways to construct disjoint union of two given networks in the disk preserving planarity.}\label{fig:amalg1}}
\end{figure}
 
We recall that the soliton data consist of a point $[A]\in\GTNN$ and a set ordered real phases $\mathcal K=\{\kappa_1<\kappa_2<\ldots<\kappa_n\}$. Moreover, in our construction one face at the boundary of the disk plays a distinguished role since we place the essential singularity of the KP wave function in the corresponding oval in the curve. Let the initial soliton data be $[A_i]\in Gr^{\mbox{\tiny{TNN}}}(k_i,n_i)$, ${\mathcal K}_i=\{\kappa^{(i)}_1<\kappa^{(i)}_2<\ldots<\kappa^{(i)}_{n_i}\}$, $i=1,2$. Then we have exactly two ways to perform the disjoint union preserving the total non-negativity property:
\begin{enumerate}
\item All boundary vertices of one network precede all boundary vertices of the second one (see Figure~\ref{fig:amalg1} [left]). This situation occures when $\kappa^{(2)}_{n_2}< \kappa^{(1)}_{1}$. In this case the resulting infinite oval $\Omega_0$ is the union of the infinite ovals  $\Omega_{0,1}$, $\Omega_{0,2}$ of the initial networks and all finite ovals are not modified;
\item All boundary vertices of one network are located between two consecutive boundary vertices of the other one (see Figure~\ref{fig:amalg1} [right]). This situation occures when $\kappa^{(2)}_{j}< \kappa^{(1)}_{1}<\ldots <\kappa^{(1)}_{n_1}<\kappa^{(2)}_{j+1}$. Let us denote $\Omega_{1,2}$,  $\Omega_{1}$  respectively the finite oval  containing this pair of boundary vertices in the initial and final networks. In this case the infinite oval $\Omega_0$ coincides with $\Omega_{0,2}$, the infinite oval of the ``external'' network (${\mathcal N}_2$ in Figure~\ref{fig:amalg1} [right]), whereas $\Omega_1$ is built out of  $\Omega_{1,2}$ and of $\Omega_{0,1}$, the infinite oval of the  ``internal'' network  (${\mathcal N}_1$ in Figure~\ref{fig:amalg1} [right]). All other ovals are not modified.
\end{enumerate}

Let us remark that from the point of view of soliton solutions this procedure is non-trivial and generates some KP-II families of solutions discussed in the literature such as O-solitons and P-solitons \cite{CK}.

\begin{figure}
  \centering{\includegraphics[width=0.47\textwidth]{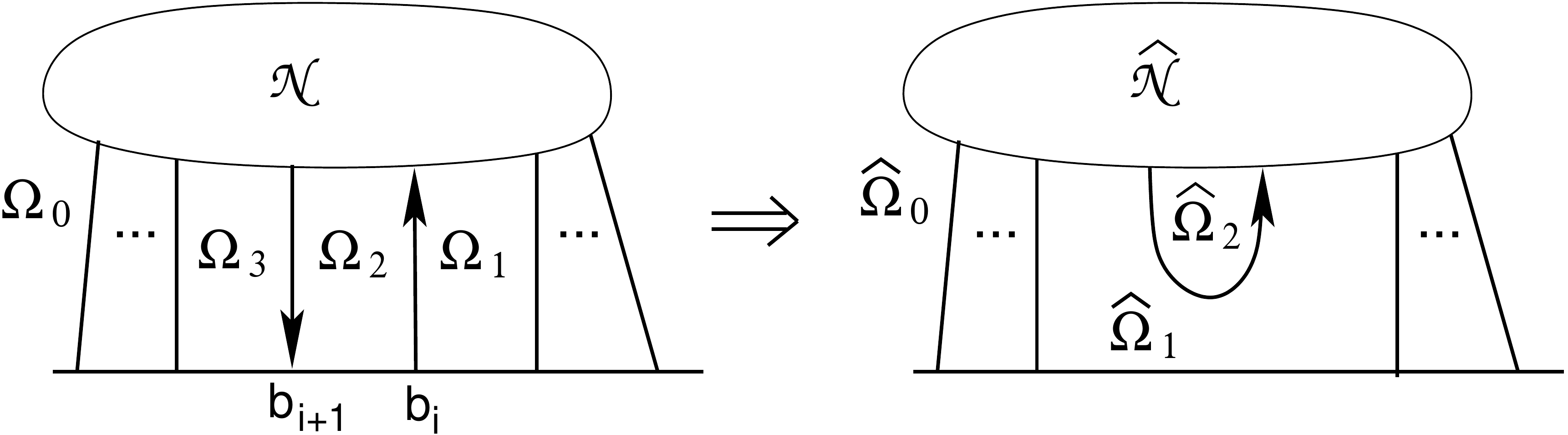} \hfill \includegraphics[width=0.47\textwidth]{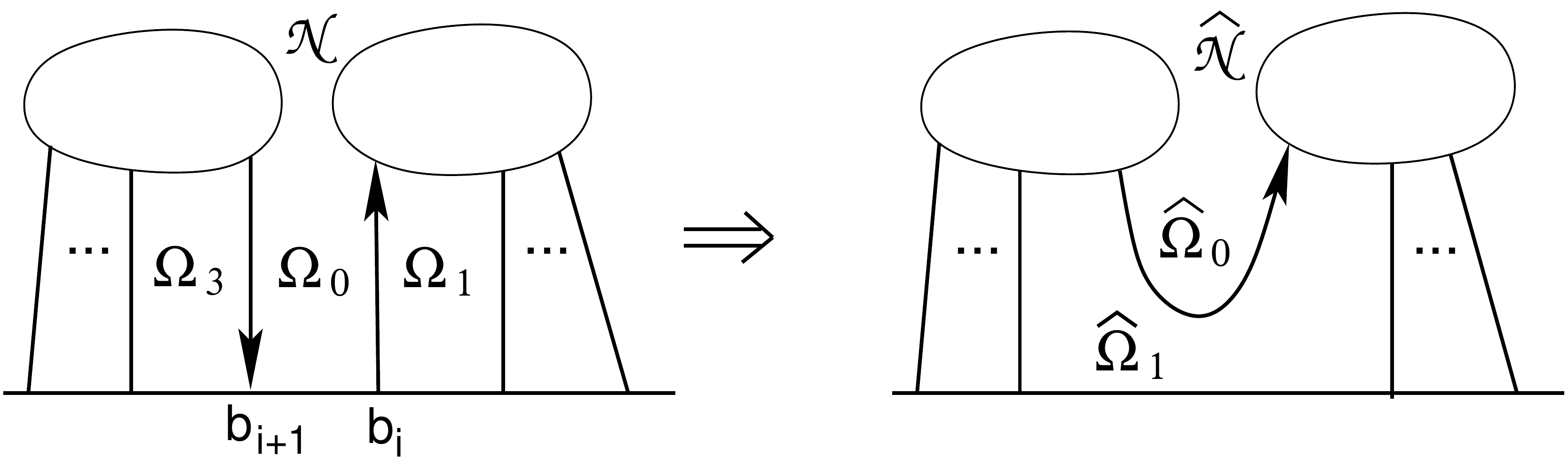}}
  
  \vspace{5mm}
   \includegraphics[width=0.47\textwidth]{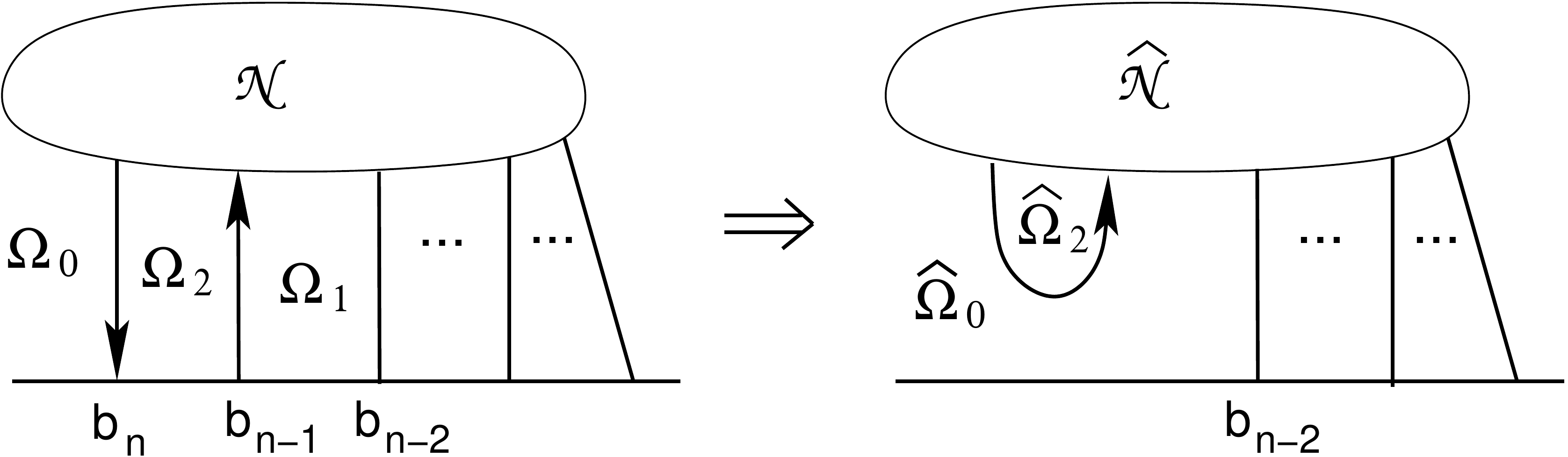}
  \caption{\small{\sl The projection procedure for defrostures of the graph preserving planarity.}\label{fig:amalg2}}
\end{figure}

Since we work only with planar directed graphs and we assume that any internal edge belongs to at least one path starting and ending at the boundary of the disk, in our setting defrosture corresponds to the elimination of two consecutive boundary vertices, one of which is a source  and the other one is a sink, and to gluing the resulting directed half-edges (see Figure~\ref{fig:amalg2}). Defrosture transforms the face $\Omega_2$ into the face $\hat\Omega_2$, and the faces $\Omega_1$, $\Omega_3$ are merged into the face $\hat\Omega_1$.

It is easy to check that any planar graph in the disk considered in our text can be obtained starting from several copies of Le-graphs associated with the small positive Grassmannians $Gr^{TP}(1,3)$, $Gr^{TP}(2,3)$ and $Gr^{TP}(1,2)$ in such a way that at any step planarity is preserved and any edge of the resulting graph belongs at least to one directed path starting and ending at the boundary of the disk. 

Let us now explain the effect of amalgamation on the total edge signature and on the divisor structure. At this aim we use Thereom~\ref{theo:sign_face} to compute the edge signatures of the faces of the amalgamated networks in tmers of that of the initial networks. As in the previous section let $\mathcal N$ be a plabic network representing a point in $\GTNN$, and for any given face $\Omega$, let the indices $\epsilon(\Omega)$ and $n_{\mbox{\scriptsize{white}}}(\Omega)$ respectively denote the edge signature and the number of white vertices $\partial \Omega$. Then the proof of the next Lemmas follows from:
\begin{equation}\label{eq:sign_face2}
\epsilon(\Omega) \; = \;  
\left\{ \begin{array}{ll}
n_{\mbox{\scriptsize{white}}}(\Omega) \; + \; 1 \quad \mod 2, & \quad \mbox{if } \Omega \mbox{ is a finite face}; \\
\\
n_{\mbox{\scriptsize{white}}}(\Omega) \; + \; k \quad \mod 2, & \quad \mbox{if } \Omega \mbox{ is the infinite face}.
\end{array}
\right.
\end{equation}

\begin{lemma}\textbf{Edge signature of the direct sum}\label{lem:dir_sum}
Let ${\mathcal N}_i$ be plabic networks representing points in $Gr^{\mbox{\tiny{TNN}}}(k_i,n_i)$, $i=1,2$ and $\mathcal N$ be their disjoint union representing a point in $Gr^{\mbox{\tiny{TNN}}}(k_1+k_2,n_1+n_2)$ with notations as in Figure \ref{fig:amalg1}. Then, the edge signature behaves as follows: 
\begin{enumerate}
\item If all boundary vertices of $\mathcal N_2$ precede all boundary vertices of $\mathcal N_1$ (Figure~\ref{fig:amalg1} [left]),
\begin{equation}\label{eq:inf_case1}
\epsilon(\Omega_0) \; = \;  \epsilon(\Omega_{0,1}) + \epsilon(\Omega_{0,2})  \quad \mod 2,
\end{equation}
and is unchanged in all other faces;
\item If all boundary vertices of ${\mathcal N}_1$ are located between two consecutive boundary vertices of ${\mathcal N}_2$ (Figure~\ref{fig:amalg1} [right]),
\begin{equation}\label{eq:inf_case2}
\epsilon(\Omega_0) \; = \;  \epsilon(\Omega_{0,2}) + k_1  \quad \mod 2,\qquad\qquad
\epsilon(\Omega_1) \; = \;  \epsilon(\Omega_{0,1}) + \epsilon(\Omega_{1,2}) + k_1, \quad \mod 2,
\end{equation}
and is unchanged in all other faces.
\end{enumerate}
\end{lemma}

We now describe the action of defrosting on the edge signatures. We remark that if both ovals $\Omega_1$, $\Omega_3$ are finite, then $\hat\Omega_1$ is also a finite oval, otherwise it is the infinite oval. Similarly, $\hat\Omega_2$ is the infinite oval if and only if $\Omega_2$ is the infinite oval .

\begin{lemma}\textbf{Effect of defrosting on edge signatures}\label{lem:defrost}
  Let $\mathcal N$ be plabic network representing a point in $Gr^{\mbox{\tiny{TNN}}}(k,n)$,  and $\hat{\mathcal N}$ be the defrosted network representing a point in $Gr^{\mbox{\tiny{TNN}}}(k-1,n-2)$ with notations as in Figure \ref{fig:amalg2}. Then, the edge signature behaves as follows:

\begin{enumerate}
\item If $b_i\ne b_1,b_{n-1}$ and $\Omega_2$ is not the infinite oval, then 
  \begin{equation}\label{eq:inf_case3}
 \begin{split}   
   \epsilon(\hat\Omega_1) \; &= \;  \epsilon(\Omega_{1}) + \epsilon(\Omega_{3}) +1  \quad \mod 2,\\
   \epsilon(\hat\Omega_2) \; &= \;  \epsilon(\Omega_{2}),   \quad \mod 2,\\
   \epsilon(\hat\Omega_0) \; &= \;  \epsilon(\Omega_{0})  +1   \quad \mod 2,\\
  \end{split} 
\end{equation}
\item If $b_i\ne b_1,b_{n-1}$ and $\Omega_2$ is the infinite oval $\Omega_2=\Omega_0$ , then 
  \begin{equation}\label{eq:inf_case3_bis}
 \begin{split}   
   \epsilon(\hat\Omega_1) \; &= \;  \epsilon(\Omega_{1}) + \epsilon(\Omega_{3}) +1  \quad \mod 2,\\
   \epsilon(\hat\Omega_2) \; &= \;  \epsilon(\Omega_{2}) +1 ,   \quad \mod 2,\\
  \end{split} 
\end{equation}

\item If $b_i=b_{n-1}$ then 
  \begin{equation}\label{eq:inf_case4}
 \begin{split}   
   \epsilon(\hat\Omega_0) \; &= \;  \epsilon(\Omega_{0}) +  \epsilon(\Omega_{1})  \quad \mod 2,\\
   \epsilon(\hat\Omega_2) \; &= \;  \epsilon(\Omega_{2}),   \quad \mod 2.
 \end{split}
\end{equation}                                                             
\item The case  $b_i=b_{1}$ is similar to the previous one.
\end{enumerate}
In all other faces the signature is unchanged.
\end{lemma}

Next we discuss the effect of amalgamation on real regular divisor configurations. The direct sum is trivial in terms of matrices, but non-trivial superposition in terms of the corresponding soliton solutions. The possible real regular divisor configurations depend on the type of the disjoint union. If we place one graph next to the other one, the allowed real regular divisor configurations coincide with the union of all possible configurations on both initial networks. If one graph is placed inside another one, we have more admissible divisor configurations provided that the infinite face $\Omega_{0,1}$ contains at least one trivalent white vertex. 
\begin{lemma}\textbf{Real regular divisor configurations for direct sum of networks}\label{lem:direct_sum}
Let ${\mathcal N}_i$ be plabic networks representing points in $Gr^{\mbox{\tiny{TNN}}}(k_i,n_i)$, $i=1,2$ and $\mathcal N$ be their disjoint union representing a point in $Gr^{\mbox{\tiny{TNN}}}(k_1+k_2,n_1+n_2)$. Then the degree of the divisor (Sato divisor) associated to  ${\mathcal N}$ is the sum of the degrees of the divisors (Sato divisors) associated to ${\mathcal N}_1$ and ${\mathcal N}_2$.

Moreover, 
\begin{enumerate}
\item If all boundary vertices of $\mathcal N_2$ precede all boundary vertices of $\mathcal N_1$ (Figure~\ref{fig:amalg4}), then all admissible divisors ${\mathcal D}$ on  $\mathcal N$ are sums $ {\mathcal D}= {\mathcal D}_1+{\mathcal D}_2$, where  ${\mathcal D}_i$, are admissible divisors on  $\mathcal N_i$, $i=1,2$;
\item If all boundary vertices of ${\mathcal N}_1$ are located between two consecutive boundary vertices of ${\mathcal N}_2$ (Figure~\ref{fig:amalg5}) and the infinite oval of ${\mathcal N}_1$ contains no trivalent white vertices, then the admissible divisors are exclusively sums $ {\mathcal D}= {\mathcal D}_1+{\mathcal D}_2$, where  ${\mathcal D}_i$, are admissible divisors on  $\mathcal N_i$, $i=1,2$ such that either there is no Sato divisor point in the oval $\Omega_{1,2}$, or the Sato divisor point in the oval $\Omega_{1,2}$ lies outside the interval $[\kappa^{(1)}_{1},\kappa^{(1)}_{n_1}]$;
\item If all boundary vertices of ${\mathcal N}_1$ are located between two consecutive boundary vertices of ${\mathcal N}_2$ (Figure~\ref{fig:amalg5}) and the infinite oval of ${\mathcal N}_1$ contains al least one trivalent white vertex, then either $\mathcal D$ is as in Item~2, or is a configuration obtained starting from an admissible configuration ${\mathcal D}_2$ on  $\mathcal N_2$, by eliminating the Sato divisor point from the oval $\Omega_{1,2}$, adding a divisor on a trivalent white vertex in $\Omega_{0,1}$, placing $k_1+1$ Sato divisor points in ${\mathcal N}_1$ and completing the configuration respecting the regularity divisor rules for ${\mathcal N}_1$. 
\end{enumerate}
\end{lemma}
\begin{figure}
  \centering{\includegraphics[width=0.7\textwidth]{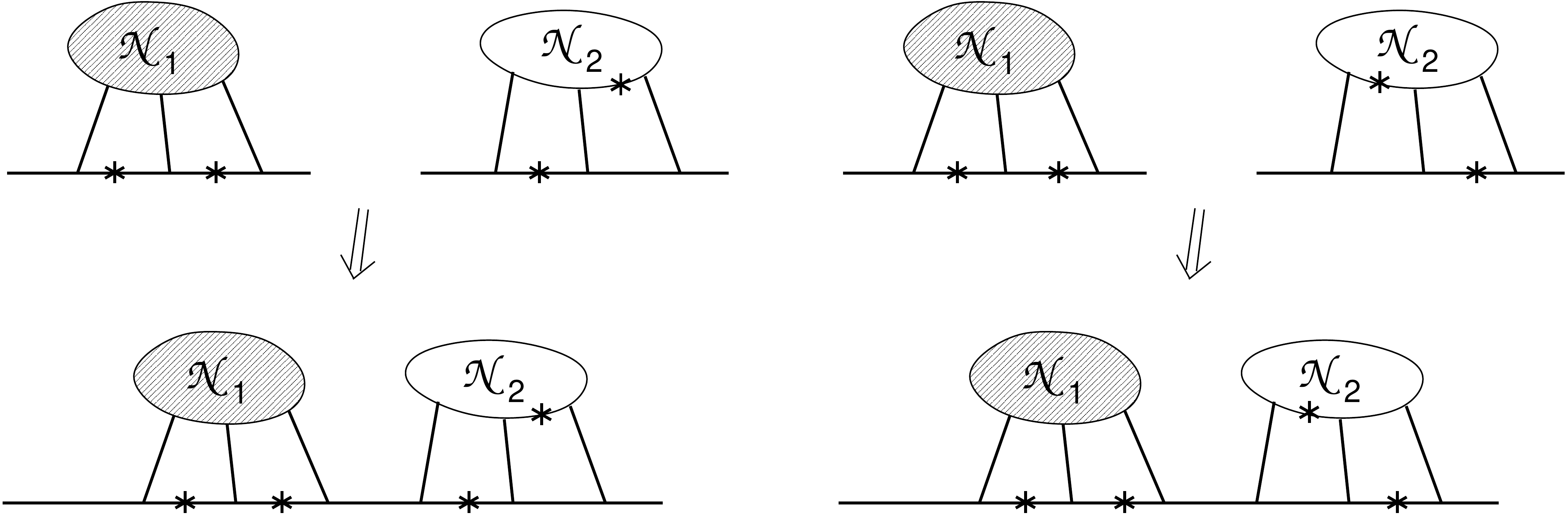}}
  \caption{\small{\sl Real regular divisor configurations for disjoint union of positroid cells if we place one graph next to another one.}\label{fig:amalg4}}
\end{figure}

\begin{figure}
  \centering{\includegraphics[width=0.8\textwidth]{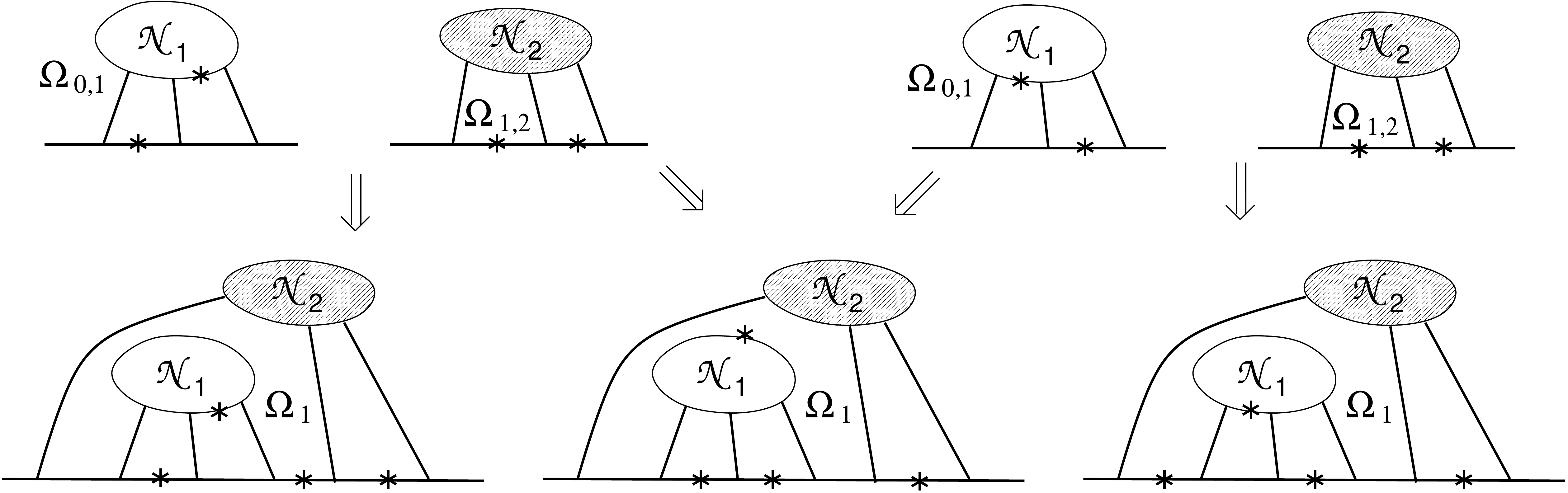}}
  \caption{\small{\sl Real regular divisor configurations for disjoint union of positroid cells if we place one graph inside the other one. }\label{fig:amalg5}}
\end{figure}

In Figures~\ref{fig:amalg4},~\ref{fig:amalg5} we illustrate the direct sum of  $Gr^{\mbox{\tiny{TP}}}(1,3)$ and  $Gr^{\mbox{\tiny{TP}}}(2,3)$. The new configuration on Figure~\ref{fig:amalg5} is the middle one.

Finally we discuss the effect of defrosting consecutive boundary vertices on the divisor structure. 

\begin{lemma}\textbf{Effect of defrosting on the regular divisors}\label{lem:defrost_divisor}
Let $\mathcal N$ be plabic network representing a point in $Gr^{\mbox{\tiny{TNN}}}(k,n)$,  and $\hat{\mathcal N}$ be the defrosted network representing a point in $Gr^{\mbox{\tiny{TNN}}}(k-1,n-2)$. Then the degrees of both the KP divisor and the Sato divisor associated to  $\hat{\mathcal N}$ are one less than degrees of the corresponding initial divisors.

Moreover, the admissible real regular divisor structures of $\hat{\mathcal N}$ correspond to the divisor structures of  $\mathcal N$ such that:
\begin{enumerate}
\item The oval $\Omega_2$ does not contain a Sato divisor point;
\item The union of ovals $\Omega_1\cup\Omega_3$ contains at least one Sato divisor point,
\end{enumerate}
and are obtained by elimination of one Sato divisor point from  $\Omega_1\cup\Omega_3$.
\end{lemma}

We schematically illustrate the above Lemma in Figures~\ref{fig:amalg6}--\ref{fig:amalg8}.
\begin{figure}
  \centering{\includegraphics[width=0.7\textwidth]{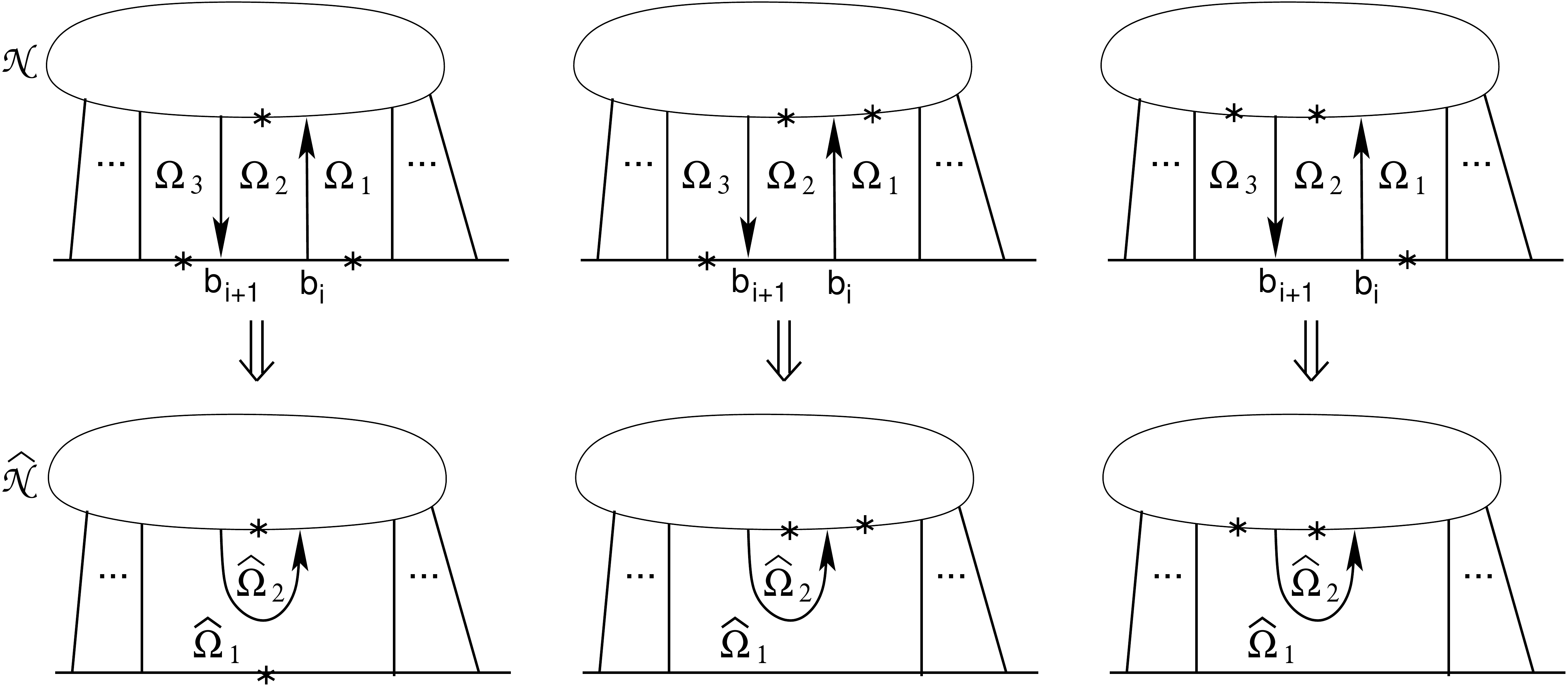}}
  \caption{\small{\sl Real regular divisor configurations for projections of positroid cells involving 3 finite ovals.}\label{fig:amalg6}}
\end{figure}

\begin{figure}
  \centering{\includegraphics[width=0.8\textwidth]{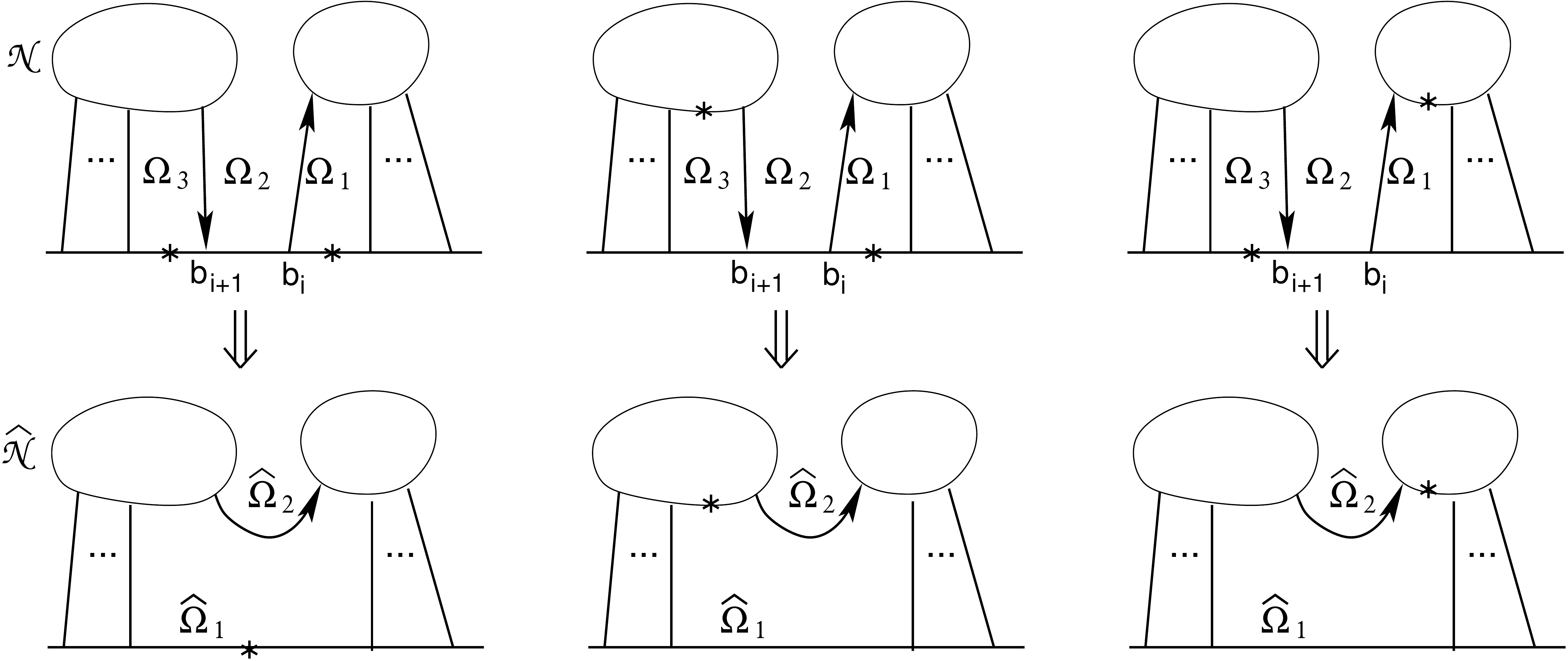}}
  \caption{\small{\sl Real regular divisor configurations for projections of positroid cells when $\Omega_2$ is the infinite oval. }\label{fig:amalg7}}
\end{figure}
\begin{figure}
  \centering{\includegraphics[width=0.6\textwidth]{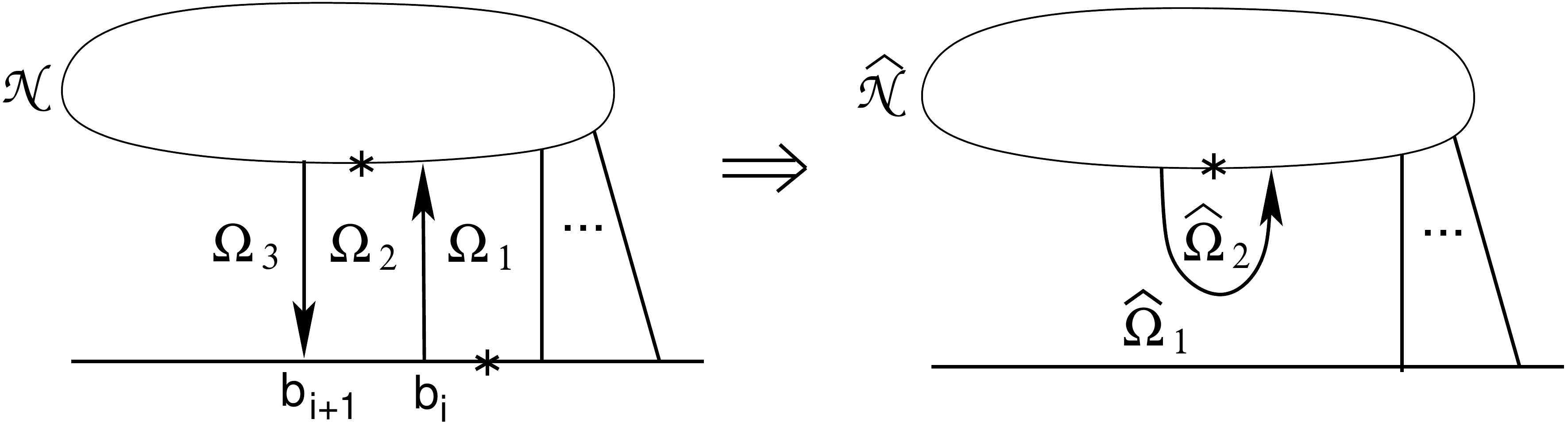}}
  \caption{\small{\sl Real regular divisor configurations for dprojections of positroid cells when either $\Omega_1$ or $\Omega_3$ is the infinite oval. }\label{fig:amalg8}}
\end{figure}

\part{Singularities of divisors}\label{sec:global_null}

The parametrization of a given positroid cell $\S$ via KP-II divisors constructed in this paper, see also \cite{AG3}, is local in the following sense: for each point in 
$\S$ and a collection of phases $\mathcal K$, we choose a fixed time $\vec t_0$ such that near this point the parametrization is locally regular. But globally we cannot exclude the situation in which al least two divisor points simultaneously approach the same node of the curve. We have 3 possible situations:
\begin{enumerate}
\item There exists a time $\vec t_0$ such that a pair of divisor point are on the same node, but for generic $\vec t$ the divisor is generic. In this case it is necessary to apply an appropriate blow-up procedure to resolve the singularity. We remark in the case of reduced graphs only such degenerations may occur. We plan to study this problem in a future paper.  Here in Section~\ref{sec:global} we solve this problem in the simplest non-trivial case  $Gr^{TP}(1,3)$. 
\item There exists a collection of positive weights such that for any time $\vec t$ a pair of divisor point are on the same node, but for generic collection of weights and genetic $\vec t$ the divisor is generic. This situation may occur for the reducible graphs studied in this paper. We briefly discuss this case in Section~\ref{sec:constr_null}. 
\item For a given graph, any collection of positive weights and any time $\vec t$ a pair of divisor point are on the same node. This situation may occur only if we release the condition that for any edge there exists a path from boundary to boundary containing it. We present an example in Section~\ref{sec:counterexample}. 
\end{enumerate}

\section{Global parametrization of positroid cells via KP divisors: the case $Gr^{TP}(1,3)$}\label{sec:global}

If the graph is reduced, the non-normalized wave function is never identically zero at a given node. However, for some non-generic time $\vec t_0$ a zero of the non-normalized wave function may coincide with the node, therefore we have a pair of divisor points at such node. The simplest example discussed in this Section shows that it is necessary to resolve the singularity in the variety of divisors. 
\begin{figure}
  \centering{\includegraphics[width=0.5\textwidth]{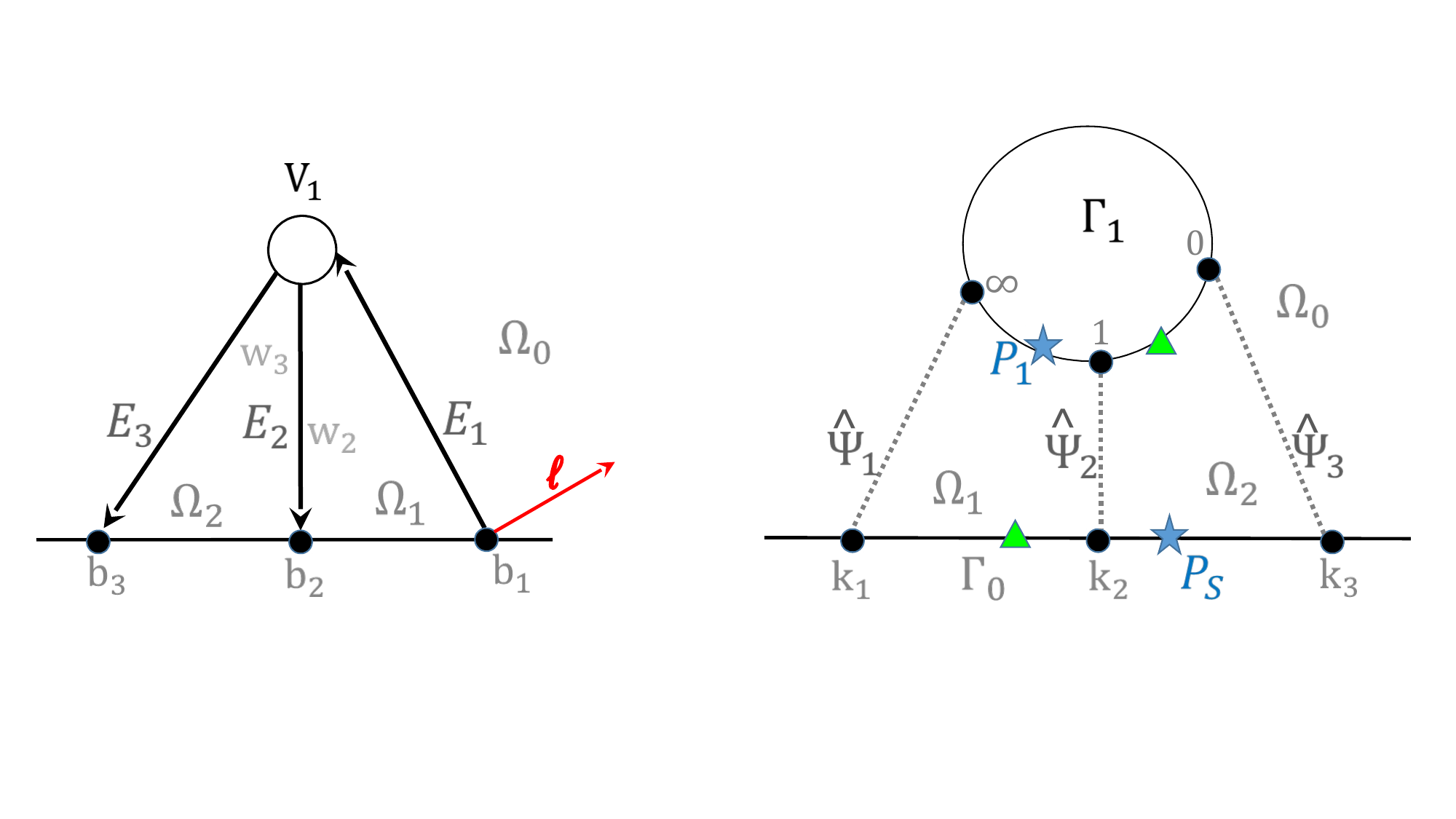}}
	\vspace{-1.2 truecm}
  \caption{\small{\sl We consider the issue of the global parametrization of positroid cells via divisors in the simplest example $Gr^{TP}(1,3)$.}\label{fig:global}}
\end{figure}

Consider the totally positive Grassmannian $Gr^{TP} (1,3)$.
With usual affine coordinates, we have the cell parametrization $[1,w_2,w_3]$, the heat hierarchy solution $f(\vec t)= e^{\theta_1 (\vec t)}+w_2 e^{\theta_2 (\vec t)}+w_3 e^{\theta_3 (\vec t)}$ and the Darboux transformation ${\mathfrak D}=\partial_x -\frac{\partial_x f(\vec t)}{f(\vec t)}$. 
Then on the oriented network (Figure \ref{fig:global}[left]) the vectors are 
$E_1 =E_2+E_3$, $E_2 =(0,w_2,0)$, $E_3 =(0,0,w_3)$,
and the edge wave function takes the values
$\Psi_1 (\vec t) =\Psi_2(\vec t)+\Psi_3(\vec t)$, $\Psi_2(\vec t) =w_2 (\kappa_2-\gamma_S) e^{\theta_2(\vec t)}$, $\Psi_3(\vec t) =w_3 (\kappa_3-\gamma_S) e^{\theta_3(\vec t)}$,
where $\gamma_S$ is the coordinate of the Sato divisor point $P_S$. If we fix the reference time $\vec t_0=\vec 0=(0,0,0,\dots)$, then
$\zeta(P_s)= \gamma_S=  \frac{\kappa_1 +w_2 \kappa_2 + w_3 \kappa_3}{1 +w_2  + w_3 }$.
On the curve $\Gamma=\Gamma_0\sqcup\Gamma_1$, at the double points the normalized KP wave function is $\hat \Psi_j(\vec t) =\frac{\Psi_j(\vec t)}{\Psi_j(\vec 0)}$ and the divisor point $P_1 \in \Gamma_1$ has local coordinate 
$\zeta(P_1) =\gamma_1 = \frac{w_3(\kappa_3-\gamma_S)}{w_3(\kappa_3-\gamma_S)+w_2(\kappa_2-\gamma_S)}$.

It is easy to check that the positivity of the weights is equivalent to $\gamma_S \in]\kappa_1,\kappa_3[$, $\gamma_1>0$ and the fact that there is exactly one divisor point in each one of the finite ovals, $\Omega_1$ and $\Omega_2$, that is
\begin{enumerate}
\item Either $\kappa_1 <\gamma_S<\kappa_2$ and $\gamma_1<1$, {\sl i.e.} $P_S\in\Omega_1$ and $P_1\in\Omega_2$. In Figure \ref{fig:global} [right] we illustrate this case representing divisor points by triangles;
\item Or $\kappa_2 <\gamma_S<\kappa_3$ and $\gamma_1>1$, {\sl i.e.} $P_S\in \Omega_2$ and $P_1\in \Omega_1$. In Figure \ref{fig:global} [right] we illustrate this case representing divisor points by stars.
\end{enumerate}
The transformation from $(w_2,w_3)$ to $(\gamma_S,\gamma_1)$ looses injectivity and full rank Jacobian along the line $w_3=\frac{\kappa_2-\kappa_1}{\kappa_3-\kappa_2}$ so that, for any $w_2>0$,
\[
\gamma_S ( w_2, \frac{\kappa_2-\kappa_1}{\kappa_3-\kappa_2} )=\kappa_2, \quad\quad \gamma_1 ( w_2, \frac{\kappa_2-\kappa_1}{\kappa_3-\kappa_2} )=1.
\]
If we invert the relation between divisor numbers and weights, we get
\[
w_2 (\gamma_S,\gamma_1)= \frac{(\gamma_1 -1)(\gamma_S-\kappa_1)}{\gamma_S-\kappa_2},\quad\quad w_3 (\gamma_S,\gamma_1)= \frac{\gamma_1(\gamma_S-\kappa_1)}{\kappa_3-\gamma_S}.
\]

Therefore in the non--generic case when $\gamma_S\to\kappa_2$ and $\gamma_1\to 1$, we need to apply the blow-up procedure at the point $(\gamma_S,\gamma_1)=(\kappa_2,1)$, by setting $\gamma_S=\kappa_2+\epsilon$, $\gamma_1=1+z\epsilon$, where
\[
\epsilon = \frac{ \omega_3 (\kappa_3- \kappa_2) - (\kappa_2-\kappa_1) }{1+\omega_2 + \omega_3}, \qquad
z = \frac{\omega_2(1+\omega_2+\omega_3)}{\omega_3(\kappa_3-\kappa_1)+\omega_2(\kappa_2-\kappa_1)},
\]
and take the limit $\epsilon\to 0$, so that
\[
w_2 (\kappa_2,1) = z(\kappa_2-\kappa_1),\quad\quad w_3 (\kappa_2,1) = \frac{\kappa_2-\kappa_1}{\kappa_3-\kappa_2}.
\]

\section{Construction of divisor if the wave function is identically zero at a node}\label{sec:constr_null}
In the class of graphs considered in this paper, the condition that every edge belongs to at least one path from boundary to boundary implies that the wave function may be identically zero at a node only for positive edge weights lying in subvarieties of codimension at least one, and just for reducible graphs.  

We may extend the construction of the real regular KP divisor to this case by taking a proper limit of both the wave function and the divisor in the space of weights, and again we require a resolution of singularities analogous to that discussed in the previous Section. Since for reducible graphs there is extra gauge weights freedom, we conjecture that such degeneracy may be always avoided by a proper choice of positive weights. In this Section we discuss a simple example leaving the detailed study to a future publication. 
 
\begin{figure}
  \centering{\includegraphics[width=0.4\textwidth]{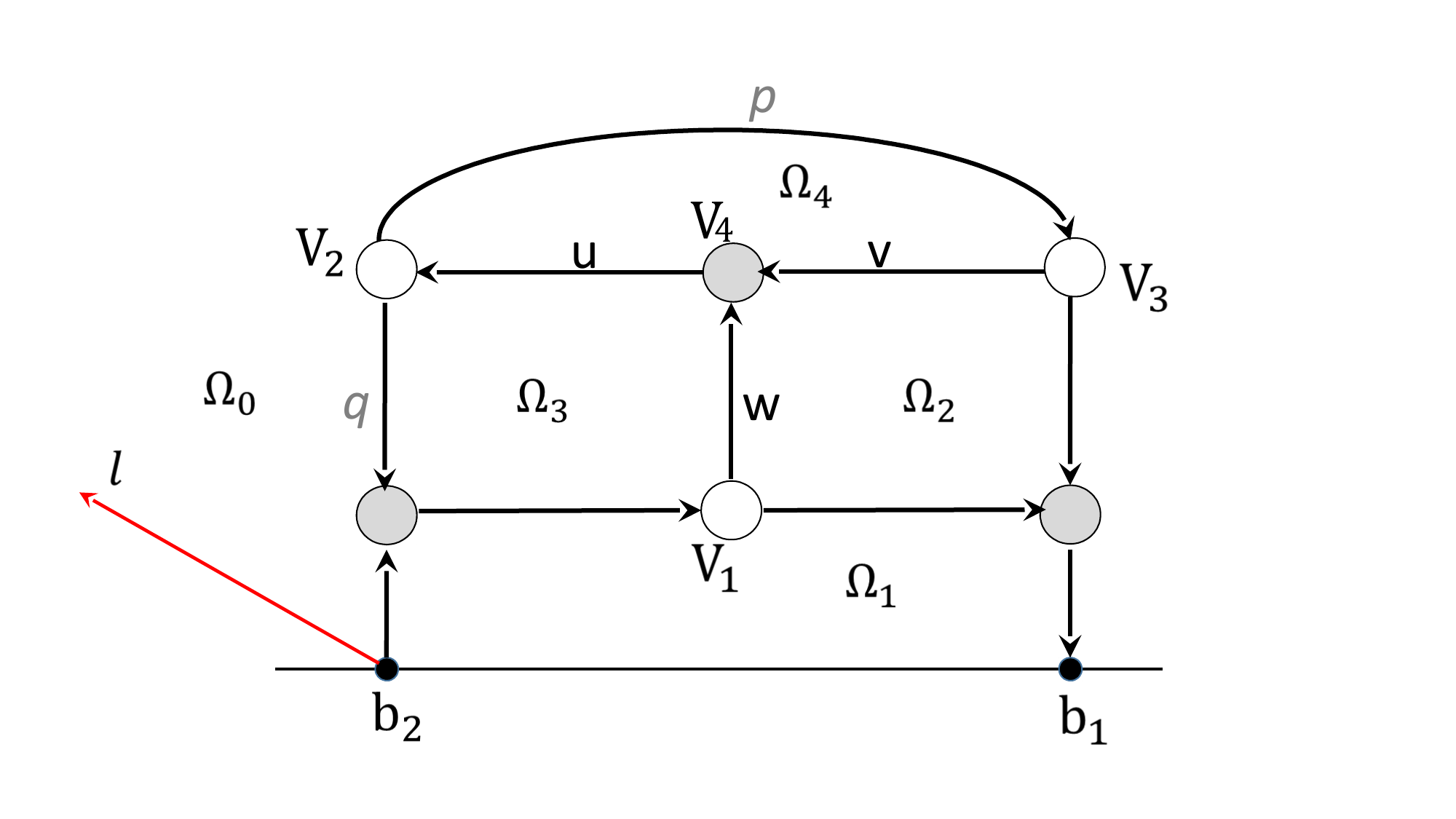}
	\hspace{.5 truecm}
	\includegraphics[width=0.4\textwidth]{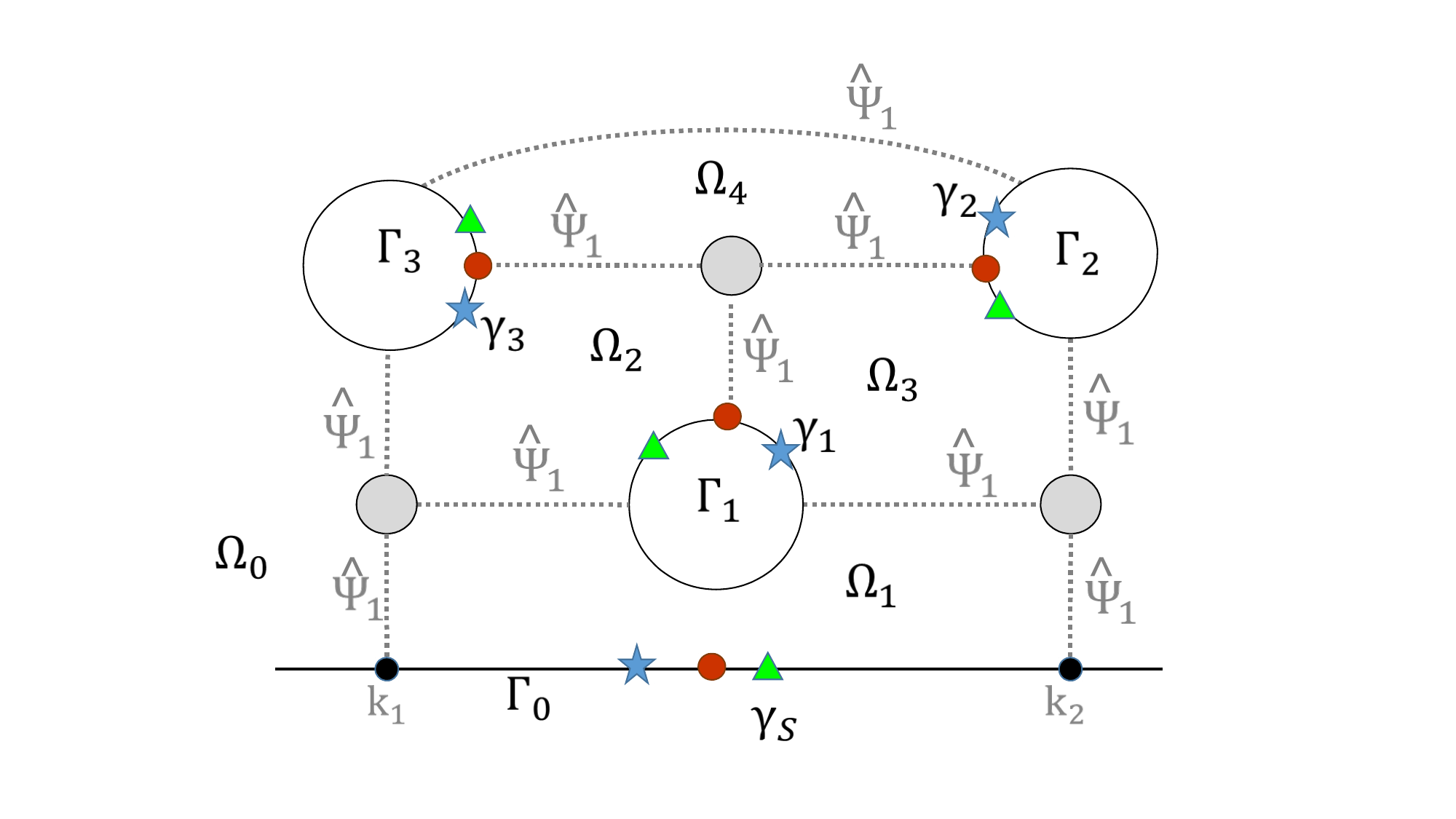}}
	\vspace{-.7 truecm}
  \caption{\small{\sl Left: the network possesses identically zero unnormalized wave function at the edges $u$, $v$, $w$  when $p=q$ [left]. Right: the KP divisor on the curve in the case $q<p$ (stars), $q=p$ (balls) and $q>p$ (triangles), where $\gamma_S$ is the Sato divisor point.}\label{fig:div_null_new}}
\end{figure}

The network in Figure \ref{fig:div_null_new} [left] represents the point $[a,1]  \in Gr^{\mbox{\tiny TP}}(1,2)$, with $a=(2p+1)/(1+p+q)$. Here all edges have weight 1 except the two edges going out from $V_2$. It is not difficult to check that the unnormalized wave function is identically 0 at the edges $u$, $v$, $w$ if $a=1$ ($p=q$).

The Sato divisor point  $\gamma_S$ belongs to the intersection of the oval $\Omega_1$ and the Sato component $\Gamma_0$ and has local coordinate $\gamma_S= \frac{a \kappa_1 e^{\theta_1(\vec t)}+ \kappa_2 e^{\theta_2(\vec t)}}{a e^{\theta_1(\vec t)} + e^{\theta_2(\vec t)}}$ in $]\kappa_1,\kappa_2[$. Since the wave function is proportional to $(\kappa_1-\gamma_s)e^{\theta_1(\vec t)}$ at all double points, the divisor points at the components $\Gamma_i$, $i\in[3]$ are independent on KP times, In the local coordinates associated to the orientation of the picture, they are $\gamma_1 = \frac{1}{a}$, $\gamma_2 = 1+\frac{aq}{a-1}$ and $\gamma_3 = \frac{a-1}{a-2}$. Let us remark that the Sato divisor point depends only on the soliton data, but other divisor points may depend on the choice of the weights representing the given soliton data because the graph is reducible, and we have extra gauge freedom.  In Figure \ref{fig:div_null_new}[right], we show the possible divisor configurations using the convention that no divisor point is attributed to black vertices in case of zero wave function:
\begin{enumerate}
\item If $0<a<1$ ($0<p<q$), then $\gamma_1 \in \Gamma_1 \cap \Omega_3$, $\gamma_2 \in \Gamma_2 \cap \Omega_4$ and $\gamma_3 \in \Gamma_3 \cap \Omega_2$. We represent one such configuration with stars in Figure \ref{fig:div_null_new} [right];
\item If $a>1$ ($0<q<p$), then $\gamma_1 \in \Gamma_1 \cap \Omega_2$, $\gamma_2 \in \Gamma_2 \cap \Omega_3$ and $\gamma_3 \in \Gamma_3 \cap \Omega_4$. We represent one such configuration with triangles in Figure \ref{fig:div_null_new} [right];
\item If $a=1$, then 3 divisor points coincide with the double points. The divisor is represented by balls in Figure \ref{fig:div_null_new} [right]). In this case a resolution of singularity similar to the one used in the previous Section is required.
\end{enumerate}

\begin{figure}
  \centering{
	\includegraphics[width=0.4\textwidth]{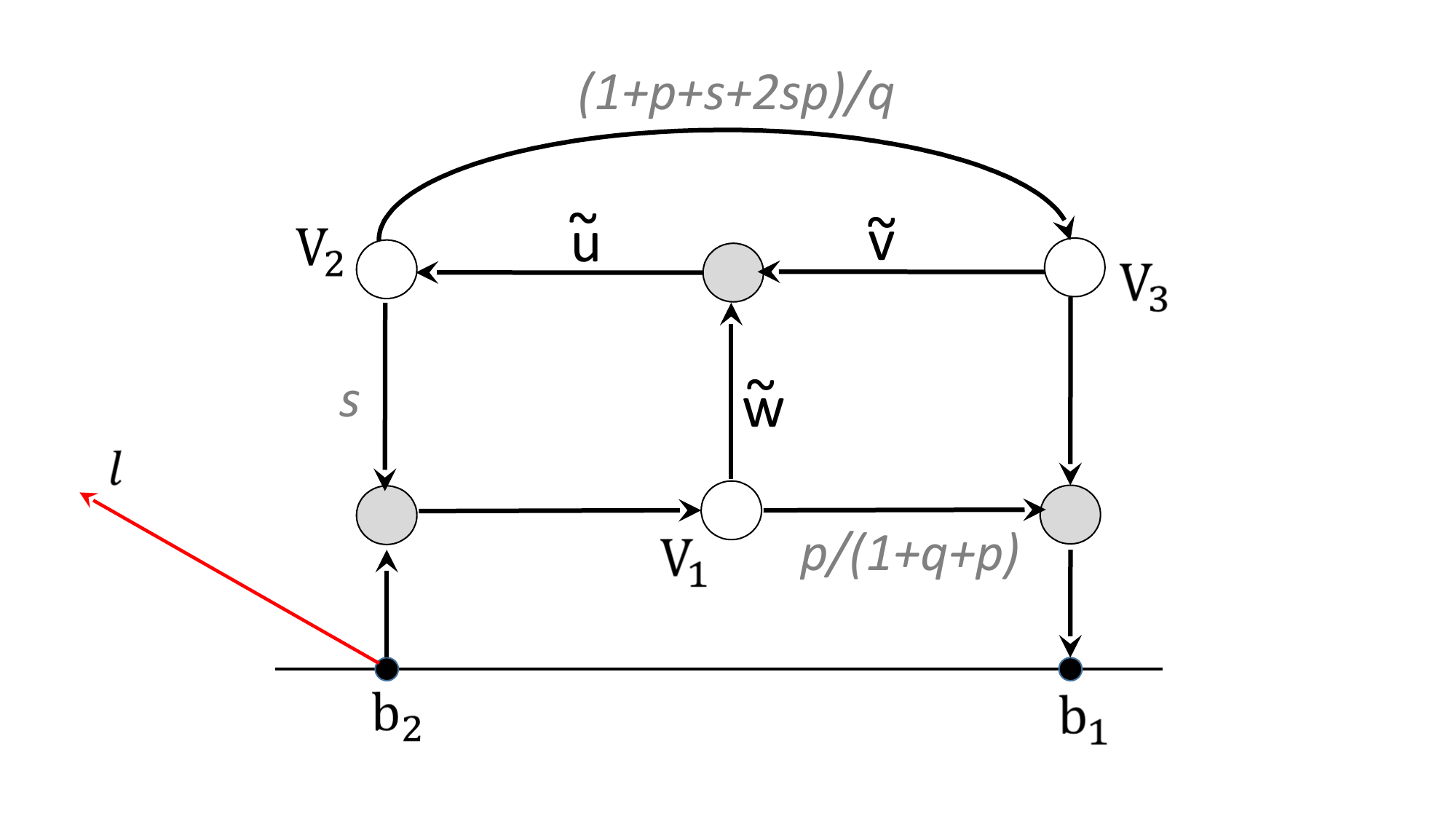}
	\hspace{.5 truecm}
	\includegraphics[width=0.4\textwidth]{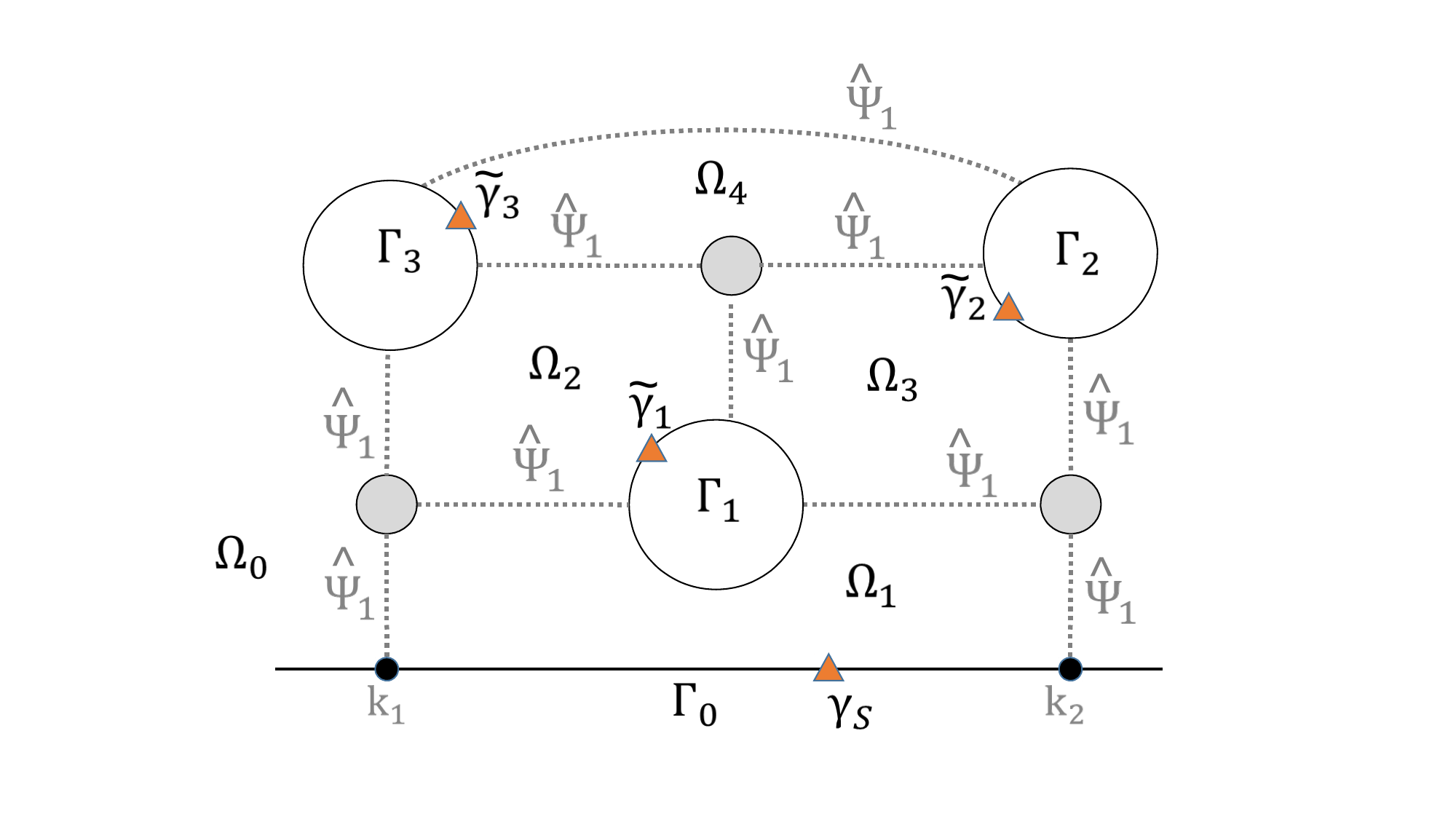}
  \caption{\small{\sl Left: the elimination of zero edge vectors for the example of Figure \ref{fig:div_null_new} using the gauge freedom for unreduced graphs. Right: the KP divisor on the curve for the gauged network.}}\label{fig:zero-vector}}
\end{figure}

\textbf{Elimination of zero edge wave function using the gauge freedom for unreduced graphs}\label{rem:div_unred}. On  unreduced graphs there is an extra gauge freedom in assigning edge weights in addition to the standard weight gauge freedom (see Remark~\ref{rem:gauge_weight}). In contrast with the standard  weight gauge freedom,  this extra gauge freedom  acts untrivially on network divisor numbers and we conjecture that such gauge may be used to eliminate identically zero wave functions. Indeed, for any fixed $s>0$, the network in Figure~\ref{fig:zero-vector}[left] represents the same point $[ a,1] \in Gr^{\mbox{\tiny TP}}(1,2)$, $a=(2p+1)/(1+p+q)$ as in Figure~\ref{fig:div_null_new}[left], but it never possesses identically zero edge wave functions (see also \cite{AG4}). The divisor is $(\gamma_S,\tilde \gamma_1=\frac{p}{1+2p},\tilde \gamma_2= 1+s\frac{1+2p}{1+p},\tilde \gamma_3= \frac{(1+p)a}{(1+p)a-2p-1})$ and again it depends on the choice of the parameters $p,q,s$ for given $a$.

\section{Example of network with  wave function vanishing at a node for arbitrary choice of weights}\label{sec:counterexample}

If we release the assumption that for any edge there exits a path from boundary to boundary containing it, then the wave function may be identically zero at some nodes for any choice of weights. In this case our counting rule for divisors in the ovals needs some modification, and we plan to study it in the future. In Figure~\ref{fig:conterexample} we show a simple example:

\begin{figure}
  \centering{
	\includegraphics[width=0.4\textwidth]{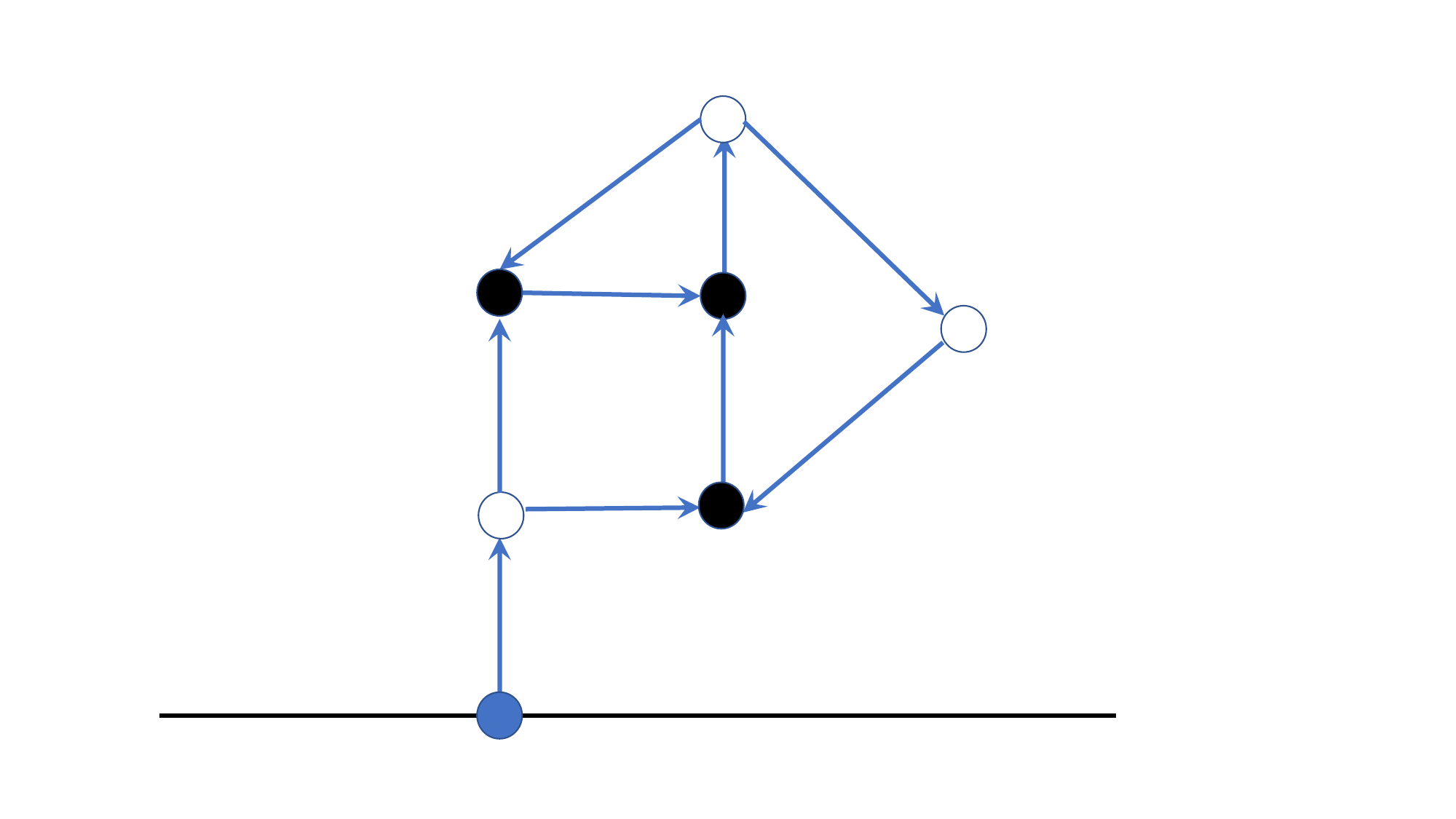}
  \caption{\small{\sl The network containing zero edge wave functions for all positive weights.}}\label{fig:conterexample}}
\end{figure}

The graph has 3 internal ovals and only 2 trivalent white vertices, therefore at least one oval has no divisor point. 

\bibliographystyle{alpha}

\end{document}